\documentclass[acmsmall,screen,nonacm]{acmart}
\usepackage{xspace}
\usepackage{enumerate}
\usepackage{extarrows}
\usepackage[e]{esvect}
\usepackage{thm-restate}

\graphicspath{{figs/}}

\theoremstyle{acmplain}
\newtheorem{theorem}{Theorem}[section]
\newtheorem{lemma}[theorem]{Lemma}
\newtheorem{proposition}[theorem]{Proposition}
\newtheorem{corollary}[theorem]{Corollary}
\newtheorem{claim}[theorem]{Claim}
\newtheorem{conjecture}[theorem]{Conjecture}
\newtheorem{thiconjecture}{Conjecture}
\newtheorem{question}[theorem]{Question}
\theoremstyle{acmdefinition}
\newtheorem{definition}[theorem]{Definition}
\newtheorem{remark}[theorem]{Remark}
\newtheorem{example}[theorem]{Example}

\numberwithin{equation}{section}

\DeclareMathOperator{\conv}{conv}
\DeclareMathOperator{\St}{St}
\DeclareMathOperator{\Link}{Link}
\DeclareMathOperator{\last}{last}
\DeclareMathOperator{\dimension}{dim}
\DeclareMathOperator{\degree}{deg}
\DeclareMathOperator{\MSO}{MSO}
\DeclareMathOperator{\inc}{inc}
\DeclareMathOperator{\Aut}{Aut}

\newcommand{\ovr}[1]{\ensuremath{\vv{#1}}\xspace}
\newcommand{\oG}{\ensuremath{\ovr{G}}\xspace}
\newcommand{\oLambda}{\ensuremath{\ovr{\Lambda}}\xspace}

\newcommand{\NN}{\ensuremath{\mathbb{N}\xspace}}
\newcommand{\JJ}{\ensuremath{\mathbb{J}\xspace}}

\newcommand{\be}{\ensuremath{\mathbf{e}\xspace}}
\newcommand{\bvv}{\ensuremath{\mathbf{v}\xspace}}
\newcommand{\bG}{\ensuremath{\mathbf{G}\xspace}}
\newcommand{\bS}{\ensuremath{\mathbf{S}\xspace}}
\newcommand{\bX}{\ensuremath{\mathbf{X}\xspace}}
\newcommand{\bY}{\ensuremath{\mathbf{Y}\xspace}}
\newcommand{\bZ}{\ensuremath{\mathbf{Z}\xspace}}

\newcommand{\cA}{\ensuremath{\mathcal{A}\xspace}}
\newcommand{\ccC}{\ensuremath{\mathcal{C}\xspace}}
\newcommand{\cD}{\ensuremath{\mathcal{D}\xspace}}
\newcommand{\cE}{\ensuremath{\mathcal{E}\xspace}}
\newcommand{\cF}{\ensuremath{\mathcal{F}\xspace}}
\newcommand{\cH}{\ensuremath{\mathcal{H}\xspace}}
\newcommand{\cI}{\ensuremath{\mathcal I}\xspace}
\newcommand{\cL}{\ensuremath{\mathcal{L}\xspace}}
\newcommand{\cP}{\ensuremath{\mathcal{P}\xspace}}
\newcommand{\cR}{\ensuremath{\mathcal{R}\xspace}}
\newcommand{\cS}{\ensuremath{\mathcal{S}\xspace}}
\newcommand{\FT}{\ensuremath{\mathcal{FT}\xspace}}
\newcommand{\PFT}{\ensuremath{\mathcal{PFT}\xspace}}
\newcommand{\GT}{\ensuremath{\mathcal{GT}\xspace}}
\newcommand{\PGT}{\ensuremath{\mathcal{PGT}\xspace}}

\newcommand{\tc}{\ensuremath{\widetilde{c}}\xspace}
\newcommand{\te}{\ensuremath{\widetilde{e}}\xspace}
\newcommand{\tm}{\ensuremath{\widetilde{m}}\xspace}
\newcommand{\tp}{\ensuremath{\widetilde{p}}\xspace}
\newcommand{\tq}{\ensuremath{\widetilde{q}}\xspace}
\newcommand{\tu}{\ensuremath{\widetilde{u}}\xspace}
\newcommand{\tv}{\ensuremath{\widetilde{v}}\xspace}
\newcommand{\tw}{\ensuremath{\widetilde{w}}\xspace}
\newcommand{\tH}{\ensuremath{\widetilde{H}}\xspace}
\newcommand{\tQ}{\ensuremath{\widetilde{Q}}\xspace}
\newcommand{\tX}{\ensuremath{\widetilde{X}}\xspace}
\newcommand{\tY}{\ensuremath{\widetilde{Y}}\xspace}
\newcommand{\tZ}{\ensuremath{\widetilde{Z}}\xspace}
\newcommand{\tpi}{\ensuremath{\widetilde{\pi}}\xspace}
\newcommand{\tlambda}{\ensuremath{\widetilde{\lambda}}\xspace}
\newcommand{\tnu}{\ensuremath{\widetilde{\nu}}\xspace}
\newcommand{\tildo}{\ensuremath{\widetilde{o}}\xspace}

\newcommand{\tbX}{\ensuremath{\widetilde{\mathbf{X}}\xspace}}
\newcommand{\tbZ}{\ensuremath{\widetilde{\mathbf{Z}}\xspace}}
\newcommand{\tbY}{\ensuremath{\widetilde{\mathbf{Y}}\xspace}}

\newcommand{\dotc}{\ensuremath{\dot{c}\xspace}}
\newcommand{\dotf}{\ensuremath{\dot{f}\xspace}}
\newcommand{\dotm}{\ensuremath{\dot{m}\xspace}}
\newcommand{\dotD}{\ensuremath{\dot{D}\xspace}}
\newcommand{\dotE}{\ensuremath{\dot{E}\xspace}}
\newcommand{\dotF}{\ensuremath{\dot{F}\xspace}}
\newcommand{\dotG}{\ensuremath{\dot{G}\xspace}}
\newcommand{\dotI}{\ensuremath{\dot{I}\xspace}}
\newcommand{\dotM}{\ensuremath{\dot{M}\xspace}}
\newcommand{\dotN}{\ensuremath{\dot{N}\xspace}}
\newcommand{\dotS}{\ensuremath{\dot{S}\xspace}}
\newcommand{\dotX}{\ensuremath{\dot{X}\xspace}}
\newcommand{\dotlambda}{\ensuremath{\dot{\lambda}\xspace}}
\newcommand{\dotSigma}{\ensuremath{\dot{\Sigma}\xspace}}
\newcommand{\dotES}{\ensuremath{\dot{ES}\xspace}}
\newcommand{\dotdiese}{\ensuremath{\dot{\#}}\xspace}
\newcommand{\dotleq}{\mathbin{\dot{\leq}}}
\newcommand{\dotprec}{\mathbin{\dot{\prec}}}

\newcommand{\dotcD}{\ensuremath{\dot{\mathcal{D}}\xspace}}
\newcommand{\dotcE}{\ensuremath{\dot{\mathcal{E}}\xspace}}
\newcommand{\dotcF}{\ensuremath{\dot{\mathcal{F}}\xspace}}

\newcommand{\dottX}{\ensuremath{\dot{\widetilde{X}}\xspace}}
\newcommand{\tdotX}{\ensuremath{\widetilde{\dot{X}}\xspace}}
\newcommand{\tdotm}{\ensuremath{\widetilde{\dot{m}}\xspace}}

\newcommand{\dotoG}{\ensuremath{\dot{\ovr{G}}}\xspace}

\newcommand{\bua}{\ensuremath{{^{\bullet}a}}\xspace}
\newcommand{\abu}{\ensuremath{a^\bullet}\xspace}
\newcommand{\bub}{\ensuremath{{^{\bullet}b}}\xspace}
\newcommand{\bbu}{\ensuremath{b^\bullet}\xspace}
\newcommand{\buh}{\ensuremath{{^{\bullet}h}}\xspace}
\newcommand{\hbu}{\ensuremath{h^\bullet}\xspace}
\newcommand{\bup}{\ensuremath{{^{\bullet}p}}\xspace}
\newcommand{\pbu}{\ensuremath{p^\bullet}\xspace}
\newcommand{\buv}{\ensuremath{^{\bullet}v}\xspace}
\newcommand{\vbu}{\ensuremath{v^\bullet}\xspace}

\newcommand{\trace}[1]{\ensuremath{{\langle {#1} \rangle}\xspace}}
\newcommand{\trsigma}{\ensuremath{\trace{\sigma}\xspace}}

\title[1-Safe Petri Nets and Special Cube Complexes]{1-Safe Petri Nets
  and Special Cube Complexes: Equivalence and Applications}

\author[J.\ Chalopin]{J\' er\'emie Chalopin}
\orcid{0000-0002-2988-8969}
\affiliation{%
  \department{LIS}
  \institution{CNRS, Aix-Marseille Universit\'e, and Universit\'e
    de Toulon}
  \streetaddress{163 Avenue de Luminy -- case 901 -- BP 5}
  \postcode{13288}
  \city{Marseille}
  \country{France}
}
\email{jeremie.chalopin@lis-lab.fr}

\author[V.\ Chepoi]{Victor Chepoi}
\orcid{0000-0002-0481-7312}
\affiliation{%
  \department{LIS}
  \institution{Aix-Marseille Universit\'e, CNRS, and
    Universit\'e de Toulon}
  \streetaddress{163 Avenue de Luminy -- case 901 -- BP 5}
  \postcode{13288}
  \city{Marseille}
  \country{France}
}
\email{victor.chepoi@lis-lab.fr}



\begin{abstract}  
  Nielsen, Plotkin, and Winskel (1981) proved that every 1-safe Petri
  net $N$ unfolds into an event structure $\cE_N$.  By a result of
  Thiagarajan (1996), these unfoldings are exactly the trace-regular
  event structures. Thiagarajan (1996) conjectured that regular event
  structures correspond exactly to trace-regular event structures.  In
  a recent paper (Chalopin and Chepoi, 2017), we disproved this
  conjecture, based on the striking bijection between domains of event
  structures, median graphs, and CAT(0) cube complexes.  On the other
  hand, we proved that Thiagarajan's conjecture is true for regular
  event structures whose domains are principal filters of universal
  covers of finite special cube complexes.

  In the current paper, we prove the converse: to any finite 1-safe
  Petri net $N$ one can associate a finite special cube complex
  ${\bX}_N$ such that the domain of the event structure $\cE_N$
  (obtained as the unfolding of $N$) is a principal filter of the
  universal cover $\tbX_N$ of ${\bX_N}$. This establishes a bijection
  between 1-safe Petri nets and finite special cube complexes and
  provides a combinatorial characterization of trace-regular event
  structures.

  Using this bijection and techniques from graph theory and geometry
  (MSO theory of graphs, bounded treewidth, and bounded hyperbolicity)
  we disprove yet another conjecture by Thiagarajan (from the paper
  with S. Yang from 2014) that the monadic second order logic of a
  1-safe Petri net (i.e., of its event structure unfolding) is
  decidable if and only if its unfolding is grid-free. It was proven
  by Thiagarajan and Yang, 2014 that the MSO logic is undecidable if
  the unfolding is not grid-free. Our counterexample is the
  trace-regular event structure which arises from a virtually special
  square complex $\bZ$. The domain of this event structure $\dotcE_Z$
  is the principal filter of the universal cover $\tbZ$ of $\bZ$ in
  which to each vertex we added a pendant edge.  The graph of the
  domain of $\dotcE_Z$ has bounded hyperbolicity (and thus the event
  structure $\dotcE_Z$ is grid-free) but has infinite treewidth.
  Using results of Seese, Courcelle and M\"uller and Schupp, we show
  that this implies that the MSO theory of the event structure
  $\dotcE_Z$ is undecidable.
\end{abstract}

\begin{document}

\maketitle

\clearpage
\section{Introduction}

Finite 1-safe Petri nets, also called net systems, are natural models
of asynchronous concurrency. Nielsen, Plotkin, and
Winskel~\cite{NiPlWi} proved that every net system
$N=(S,\Sigma,F,m_{0})$ unfolds into an event structure
${\cE}_N=(E,\le, \#,\lambda)$ describing all possible executions of
$N$: the events of $\cE_N$ are all prime Mazurkiewicz traces on the
set of transitions of $N$, equipped with the causal dependency and
conflict relations. Later results of Nielsen, Rozenberg, and
Thiagarajan~\cite{NiRoThi} show in fact that 1-safe Petri nets and
event structures represent each other in a strong sense. An event
structure~\cite{NiPlWi,Winskel,WiNi} is a partially ordered set of the
occurrences of actions, called events, together with a conflict
relation.  The partial order captures the causal dependency of
events. The conflict relation models incompatibility of events so that
two events that are in conflict cannot simultaneously occur in any
state of the computation. Consequently, two events that are neither
ordered nor in conflict may occur concurrently. The domain of an event
structure consists of all computation states, called configurations.
Each computation state is a subset of events subject to the
constraints that no two conflicting events can occur together in the
same computation and if an event occurred in a computation then all
events on which it causally depends have occurred too. Therefore, the
domain of an event structure is the set of all finite configurations
ordered by inclusion. The future (or the principal filter) of a
configuration is the set of all finite configurations containing it.

In a series of papers~\cite{RoTh,Thi_regular,Thi_conjecture,ThiaYa},
Thiagarajan formulated (alone or with co-authors) three important
conjectures (1) about the local-to-global behavior of event structures
(the \emph{nice labeling conjecture}): \emph{any event structure of
  finite degree admits a finite nice labeling}, (2) on the
relationship between event structures and net systems: \emph{regular
  event structures are exactly the unfoldings of net systems} and (3)
about the decidability of the Monadic Second Order theory (MSO theory)
of net systems: \emph{grid-free net systems are exactly the net
  systems with decidable MSO theory}. The last two conjectures were
motivated by the fact that in each case, one of the two implications
holds and by evidences and important particular cases for which the
converse implication also holds. For example, it was proven
in~\cite{Thi_regular,Thi_conjecture} that unfoldings of net systems
are exactly the trace-regular event structures, and thus the second
conjecture asks whether a regular event structure is trace-regular.

In~\cite{CC-ICALP17,CC-thiag,Ch_nice}, we provided counterexamples to
the first two conjectures. In the current paper, we provide a
counterexample to the third conjecture about the decidability of the
MSO theory of grid-free net systems. The three counterexamples are
based on different ideas and techniques, however, they all use the
bijections between domains of event structures, median graphs, and
CAT(0) cube complexes. Median graphs is the most important class of
graphs in metric graph theory and CAT(0) cube complexes play an
essential role in geometric group theory and the topology of
3-manifolds.  Even if the three conjectures turned out to be false,
the work on them raised many important open questions and the current
paper establishes a surprising bijection between 1-safe Petri nets
(trace-regular event structures) and finite special cube complexes.
Notice that special cube complexes, introduced by Haglund and
Wise~\cite{HaWi1,HaWi2}, played an essential role in the recent
solution of the famous virtual Haken conjecture for hyperbolic
3-manifolds by Agol~\cite{Agol,Agol_ICM}.

\section{On Thiagarajan's Conjectures}

We continue with an informal description of Thiagarajan's conjectures,
of some related work on them, and of the results of this paper.

\subsection{The Nice Labeling Conjecture}

The \emph{nice labeling conjecture} was formulated by Rozoy and
Thiagarajan in~\cite{RoTh} and asserts that

\begin{thiconjecture}\label{conj-nice}
  Every event structure with finite degree admits a finite nice
  labeling.
\end{thiconjecture}

A (finite) nice labeling is a labeling of events with the letters from
some finite alphabet such that any two co-initial events (i.e., any
two events which are concurrent or in minimal conflict) have different
labels. The nice labelings of event structures arise when studying the
equivalence of three different models of distributed computation:
labeled event structures, net systems, and distributed monoids. The
nice labeling conjecture can be viewed as a question about a
local-to-global finite behavior of such models.

A counterexample to nice labeling conjecture was constructed
in~\cite{Ch_nice}. It is based on the bijection between domains of
event structures, pointed median graphs, and CAT(0) cube complexes and
on the Burling's construction~\cite{Bu} of 3-dimensional box
hypergraphs with clique number 2 and arbitrarily large chromatic
numbers.  Assous, Bouchitt\'e, Charretton, and Rozoy~\cite{AsBouChaRo}
proved that the event structures of degree 2 admit nice labelings with
2 labels and noticed that Dilworth's theorem implies the conjecture
for conflict-free event structures of bounded degree.
Santocanale~\cite{Santo} proved that all event structures of degree 3
with tree-like partial orders have nice labelings with 3
labels. Chepoi and Hagen~\cite{ChHa} proved that the nice labeling
conjecture holds for event structures with 2-dimensional domains,
i.e., for event structures not containing three pairwise concurrent
events.

For CAT(0) cube complexes, a related question was independently
formulated by F. Haglund, G. Niblo, M. Sageev, and the second author
of this paper: \emph{can the 1-skeleton of any CAT(0) cube complex of
  finite degree be isometrically embedded into the Cartesian product
  of a finite number of trees?}  A negative answer to this question
was obtained in~\cite{ChHa}, based on a modification of the
counterexample from~\cite{Ch_nice}. However, in~\cite{ChHa} it was
shown that the answer is positive for 2-dimensional CAT(0) cube
complexes. Haglund~\cite{Hag-HdR} proved that this embedding question
has a positive answer for hyperbolic CAT(0) cube complexes. Modifying
the argument of~\cite{Hag-HdR}, we can also show that the nice
labeling conjecture is true for event structures with hyperbolic
domains.

\subsection{The Conjecture on Regular Event Structures}

To deal with net systems,
Thiagarajan~\cite{Thi_regular,Thi_conjecture} introduced the notions
of regular event structure and trace-regular event structure. The main
difference is that the regularity of event structures is defined for
unlabeled event structures while trace-regularity is defined under the
much stronger assumption of a given trace-regular labeling.  These
definitions were motivated by the fact that the event structures
${\cE}_N$ arising from \emph{finite} 1-safe Petri nets $N$ are
regular: Thiagarajan~\cite{Thi_regular} in fact proved that event
structures of \emph{finite} 1-safe Petri nets correspond to
trace-regular event structures.  This lead Thiagarajan to conjecture
(see also Conjecture~\ref{conjecture_regular} below)
in~\cite{Thi_regular,Thi_conjecture} that

\begin{thiconjecture}\label{conj2}
  Regular event structures and trace-regular event structures are the
  same.
\end{thiconjecture}

Equivalently, this can be reformulated in the following way: \emph{An
  event structure $\cE$ is isomorphic to the event structure unfolding
  of a net system if and only if $\cE$ is regular.}

Nielsen and Thiagarajan~\cite{NiThi} established this conjecture for
regular conflict-free event structures and Badouel, Darondeau, and
Raoult~\cite{BaDaRa} proved it for context-free event domains, i.e.,
for domains whose underlying graph is a context-free graph sensu
M\"uller and Schupp~\cite{MullerSchupp}. Morin~\cite{Morin} showed
that any event structure admitting a regular nice labeling is
trace-regular.

In~\cite{CC-thiag}, we presented a counterexample to
Conjecture~\ref{conj2} based on a geometric and combinatorial view on
event structures.  To deal with regular event structures, we showed
in~\cite{CC-thiag} how to construct regular event domains from CAT(0)
cube complexes. By a result of Gromov~\cite{Gromov}, CAT(0) cube
complexes are exactly the universal covers of nonpositively curved
cube (NPC) complexes.  Of particular importance for us are the CAT(0)
cube complexes arising as universal covers of \emph{finite} NPC
complexes. We adapted the universal cover construction to directed NPC
complexes $(Y,o)$ and showed that every principal filter of the
directed universal cover $(\tY,\tildo)$ is the domain of an event
structure. Furthermore, if the NPC complex $Y$ is finite, then this
event structure is regular. Motivated by this result, we called an
event structure \emph{strongly regular} if its domain is the principal
filter of the directed universal cover ${\tbY}=(\tY,\tildo)$ of a
finite directed NPC complex ${\bY}=(Y,o)$. Our counterexample to
Conjecture~\ref{conj2} is a strongly regular event structure not
admitting a finite regular nice labeling. It is derived from a
particular NPC square complex $\bX$ defined by
Wise~\cite{Wi_thesis,Wi_csc}.

In view of this counterexample, one can ask the following two
important questions:

\begin{question}\label{question1}
  Are the event structure unfoldings of finite 1-safe Petri nets
  strongly regular?
\end{question}

\begin{question}\label{question2}
  Under which conditions a regular event structure is trace-regular?
\end{question}

Haglund and Wise~\cite{HaWi1,HaWi2} introduced four types of
pathologies which may occur in NPC complexes.  They called the NPC
complexes without such pathologies \emph{special}.  The main
motivation for studying special cube complexes was the profound idea
of Wise that the famous virtual Haken conjecture for hyperbolic
3-manifolds can be reduced to solving problems about special cube
complexes. In a breakthrough result, Agol~\cite{Agol,Agol_ICM}
completed this program and solved the virtual Haken conjecture using
the deep theory of special cube complexes of~\cite{HaWi1,HaWi2}. The main
ingredient in this proof is Agol's theorem that finite NPC complexes
whose universal covers are hyperbolic are virtually special (i.e.,
they admit finite covers which are special cube complexes).

In~\cite{CC-thiag} we proved that Thiagarajan's conjecture is true for
event structures whose domains arise as principal filters of universal
covers of finite special cube complexes. Using the result of Agol, we
specified this result and showed that Thiagarajan's conjecture is true
for strongly regular event structures whose domains occur as principal
filters of hyperbolic CAT(0) cube complexes that are universal covers
of finite directed NPC complexes. Since context-free domains are
hyperbolic, this result can be viewed in some sense as a partial
generalization of the result of Badouel et al.~\cite{BaDaRa}.

In the current paper, we establish the converse to the previous result
of~\cite{CC-thiag}: we prove that to any 1-safe Petri net $N$ one can
associate a finite directed labeled special cube complex ${\bX}_N$
such that the domain of the event structure unfolding ${\cE}_N$ is a
principal filter of the universal cover $\tbX_N$ of ${\bX_N}$.
This shows that all event structures arising as unfoldings of finite
1-safe Petri nets are strongly regular, answering in the positive
Question~\ref{question1}. This also shows that specialness must be
added to strong regularity to ensure a positive answer to
Thiagarajan's Conjecture~\ref{conj2}.  Therefore, the trace-regular
event structures can be characterized as the event structures whose
domains arise from \emph{finite} \emph{special} cube complexes.  This
establishes a surprising bijection between 1-safe Petri nets
(fundamental objects in concurrency) and special cube complexes
(fundamental objects in geometric group theory).

\subsection{The Conjecture on the Decidability of the MSO Logic of
  Trace-regular Event Structures}

Thiagarajan and Yang~\cite{ThiaYa} defined the monadic second order
(MSO) theory ${\MSO}({\cE}_N)$ of an event structure unfolding
$\cE_N=(E,\le, \#,\lambda)$ of a net system $N=(S,\Sigma,F,m_{0})$ as
the MSO theory of the relational structure
$(E,{(R_a)}_{a\in \Sigma},\le)$ (see Subsection~\ref{MSO-def} for
definitions).  This immediately leads to the following fundamental
question:

\begin{question}\label{question3}
  When ${\MSO}(\cE_N)$ is decidable?
\end{question}

It turns out that the MSO theory of trace event structures is not
always decidable:~\cite{ThiaYa} presented such an example suggested by
I. Walukiewicz. To circumvent this example, Thiagarajan and Yang
formulated the notion of a grid event structure and they showed that
the MSO theory of event structures containing grids is
undecidable. This leads Thiagarajan to conjecture in~\cite{ThiaYa}
that (see also Conjecture~\ref{MSO} below):

\begin{thiconjecture}
  The $\MSO$ theory of a trace-regular event structure $\cE_N$ is
  decidable if and only if $\cE_N$ is grid-free.
\end{thiconjecture}

Notice also that preceding~\cite{ThiaYa}, Madhusudan~\cite{Mad} proved
that the MSO theory of a trace event structure is decidable provided
quantifications over sets are restricted to conflict-free subsets of
events. This shows that the MSO theory of conflict-free trace-regular
event structures is decidable.

With the event structure $\cE_N$ one can associate two other MSO
logics: the MSO logic $\MSO(\oG(\cE_N))$ of the directed graph
$\oG(\cE_N)$ of the domain $\cD(\cE_N)$ of $\cE_N$ and the MSO logic
$\MSO(G(\cE_N))$ of the undirected graph $G(\cE_N)$ of the
domain. This leads to the next question:

\begin{question}\label{question4}
  When $\MSO(\oG(\cE_N))$ (respectively, $\MSO(G(\cE_N))$) is
  decidable?
\end{question}

We prove in this paper that the decidability of $\MSO(G(\cE_N))$ and
of $\MSO(\oG(\cE_N))$ is equivalent to the fact that $G(\cE_N)$ has
finite treewidth and to the fact that $\oG(\cE_N)$ is a context-free
graph. This completely answers Question~\ref{question4}.  We also
prove that if $\MSO(\oG(\cE_N))$ is decidable, then $\MSO(\cE_N)$ is
decidable (the converse is not true). We introduce the notion of
hairing of an event structure $\cE_N$, which is an event structure
$\dotcE_N$ obtained from $\cE_N$ by adding an event $e_c$ for each
configuration $c$ of the domain in a such a way that $e_c$ is in
conflict with all events except those from $c$ (those events precede
$e_c$).  We prove that $\MSO(\dotcE_N)$ is decidable if and only if
$\MSO(\oG(\cE_N))$ is decidable, i.e., if and only if $G(\cE_N)$ has
finite treewidth, providing provide a partial answer to
Question~\ref{question3}.

Using these results, we construct a counterexample to Thiagarajan's
Conjecture~\ref{MSO}.  Namely, we construct an NPC square complex
${\bZ}$.  We show that ${\bZ}$ is virtually special and thus any
principal filter of the universal cover of ${\bZ}$ is the domain of a
trace-regular event structure $\cE_Z$. The hairing $\dotcE_Z$ of
$\cE_Z$ is still trace-regular.  We show that the graphs $G(\cE_Z)$
and $G(\dotcE_Z)$ have infinite treewidth and bounded
hyperbolicity. The first result implies that $\MSO(\dotcE_Z)$ is
undecidable while the second result shows that $\dotcE_Z$ is
grid-free.

\section{Event Structures and Net Systems}

\subsection{Event Structures} 

An \emph{event structure} is a triple ${\cE}=(E,\le, \#)$, where
\begin{itemize}
\item $E$ is a set of \emph{events},
\item $\le\subseteq E\times E$ is a partial order of \emph{causal
    dependency},
\item $\#\subseteq E\times E$ is a binary, irreflexive, symmetric
  relation of \emph{conflict},
\item $\downarrow \!e:=\{ e'\in E: e'\le e\}$ is finite for any
  $e\in E$,
\item $e\# e'$ and $e'\le e''$ imply $e\# e''$.
\end{itemize}

Two events $e',e''$ are \emph{concurrent} (notation $e'\| e''$) if
they are order-incomparable and they are not in conflict.  The
conflict $e'\# e''$ between two elements $e'$ and $e''$ is
\emph{minimal} (notation $e'\#_{\mu} e''$) if there is no event
$e\ne e',e''$ such that either $e\le e'$ and $e\# e''$ or $e\le e''$
and $e\# e'$. We say that $e$ is an \emph{immediate predecessor} of
$e'$ (notation $e\lessdot e'$) if and only if $e\le e', e\ne e'$, and
for every $e''$ if $e\le e''\le e'$, then $e''=e$ or $e''=e'$.

Given two event structures ${\cE}_1=(E_1,\le_1, \#_1)$ and
${\cE}_2=(E_2,\le_2, \#_2)$, a bijective map $f:E_1\rightarrow E_2$ is
an isomorphism if $e\le_1 e'$ iff $f(e)\le_2 f(e')$ and $e\#_1 e'$ iff
$f(e)\#_2 f(e')$ for every $e,e'\in E_1$. Then ${\cE}_1$ and ${\cE}_2$
are said to be isomorphic; notation ${\cE}_1\equiv {\cE}_2$.

A \emph{configuration} of an event structure ${\cE}=(E,\le,\#)$ is any
finite subset $c\subset E$ of events which is \emph{conflict-free}
($e,e'\in c$ implies that $e,e'$ are not in conflict) and
\emph{downward-closed} ($e\in c$ and $e'\le e$ implies that
$e'\in c$)~\cite{WiNi}. Notice that $\varnothing$ is always a
configuration and that $\downarrow \!e$ and
$\downarrow \!e\setminus \{ e\}$ are configurations for any $e\in E$.
The \emph{domain} of $\cE$ is the set $\cD:=\cD({\cE})$ of all
configurations of ${\cE}$ ordered by inclusion; $(c',c)$ is a
(directed) edge of the Hasse diagram of the poset
$({\cD}({\cE}),\subseteq)$ if and only if $c=c'\cup \{ e\}$ for an
event $e\in E\setminus c$.  An event $e$ is \emph{enabled} by a
configuration $c$ if $e\notin c$ and $c\cup \{ e\}$ is a
configuration.  Denote by $en(c)$ the set of all events enabled at the
configuration $c$.  Two events are \emph{co-initial} if they are both
enabled at some configuration $c$. Note that if $e$ and $e'$ are
co-initial, then either $e \#_{\mu} e'$ or $e \| e'$. Note that two
events $e$ and $e'$ are in minimal conflict $e\#_{\mu}e'$ if and only
if $e \# e'$ and $e$ and $e'$ are co-initial.  The \emph{degree}
$\degree(\cE)$ of $\cE$ is the least positive integer $d$ such that
$|en(c)|\le d$ for any configuration $c$ of $\cE$; $\cE$ has
\emph{finite degree} if $\degree(\cE)$ is finite. The \emph{future}
(or the \emph{principal filter}) $\cF(c)$ of a configuration $c$ is
the set of all configurations $c'$ containing $c$:
$\cF(c) =~\uparrow\! c:=\{c'\in {\cD}({\cE}): c\subseteq c'\}$,
i.e., $\cF(c)$ is the principal filter of $c$ in the ordered set
$({\cD}({\cE}),\subseteq)$.

A \emph{labeled event structure} ${\cE}^{\lambda}=({\cE},\lambda)$ is
defined by an \emph{underlying event structure} ${\cE}=(E,\le, \#)$
and a \emph{labeling} $\lambda$ that is a map from $E$ to some
alphabet $\Sigma$. Two labeled event structures
${\cE}_1^{\lambda_1}=({\cE}_1,\lambda_1)$ and
${\cE}_2^{\lambda_1}=({\cE}_2,\lambda_2)$ are isomorphic (notation
${\cE}^{\lambda_1}_1\equiv {\cE}^{\lambda_2}_2$) if there exists an
isomorphism $f$ between the underlying event structures ${\cE}_1$ and
${\cE}_2$ such that $\lambda_2(f(e_1))=\lambda_1(e_1)$ for every
$e_1\in E_1$.  A labeling $\lambda: E\rightarrow \Sigma$ of an event
structure $\cE$ defines naturally a labeling of the directed edges of
the Hasse diagram of its domain $\cD(\cE)$ that we also denote by
$\lambda$. A labeling $\lambda: E\rightarrow \Sigma$ of $\cE$ is
called a \emph{nice labeling} if any two co-initial events have
different labels~\cite{RoTh}. A nice labeling of ${\cE}$ can be
reformulated as a labeling of the directed edges of the Hasse diagram
of its domain $\cD(\cE)$ subject to the following local conditions:

\begin{enumerate}[]
\item[]{\emph{{Determinism:}}} the edges outgoing from
the same vertex of ${\cD}({\cE})$ have different labels;
\item[]{\emph{{Concurrency:}}} the opposite edges of each square
of ${\cD}({\cE})$ are labeled with the same labels.
\end{enumerate}

In the following, we use interchangeably the labeling of an event
structure and the labeling of the edges of its domain.

\subsection{Mazurkiewicz Traces}

A \emph{(Mazurkiewicz) trace alphabet} is a pair $M=(\Sigma,I)$, where
$\Sigma$ is a finite non-empty alphabet set and
$I\subset \Sigma\times \Sigma$ is an irreflexive and symmetric
relation called the \emph{independence relation}.  The relation
$D := (\Sigma\times \Sigma)\setminus I$ is called the \emph{dependence
  relation}. As usual, $\Sigma^*$ is the set of finite words with
letters in $\Sigma$. For $\sigma\in \Sigma^*$, $\last(\sigma)$ denotes
the last letter of $\sigma$.  The independence relation $I$ induces
the equivalence relation $\sim_I$, which is the reflexive and
transitive closure of the binary relation $\leftrightarrow_I$: \emph{if
  $\sigma,\sigma'\in \Sigma^*$ and $(a,b)\in I$, then
  $\sigma ab\sigma'\leftrightarrow_I\sigma ba\sigma'$.} The
$\sim_I$-equivalence class containing $\sigma\in \Sigma^*$ is called a
\emph{(Mazurkiewicz) trace} and will be denoted by $\trsigma$. The
trace $\trsigma$ is \emph{prime} if $\sigma$ is non-null and for every
$\sigma'\in \trsigma$, $\last(\sigma)=\last(\sigma')$. The partial
ordering relation $\sqsubseteq$ between traces is defined by
$\trsigma\sqsubseteq \trace{\tau}$ (and $\trsigma$ is said to be a
\emph{prefix} of $\trace{\tau}$) if there exist $\sigma'\in \trsigma$
and $\tau'\in \trace{\tau}$ such that $\sigma'$ is a prefix of
$\tau'$.

\subsection{Trace-regular Event Structures}\label{regular}

In this subsection, we recall the definitions of regular event
structures, trace-regular event structures, and regular nice labelings
of event structures. We closely follow the definitions and notations
of~\cite{Thi_regular,Thi_conjecture,NiThi}.  Let ${\cE}=(E,\le, \#)$
be an event structure. Let $c$ be a configuration of $\cE$.  Set
$\#(c)=\{ e': \exists e\in c, e\# e'\}$. The \emph{event structure
  rooted at} $c$ is defined to be the triple
$\cE\setminus c=(E',\le',\#')$, where $E'=E\setminus (c\cup \# (c))$,
$\le'$ is $\le$ restricted to $E'\times E'$, and $\#'$ is $\#$
restricted to $E'\times E'$. It can be easily seen that the domain
${\cD}({\cE}\setminus c)$ of the event structure $\cE\setminus c$ is
isomorphic to the principal filter $\cF(c)$ of $c$ in ${\cD}({\cE})$
such that any configuration $c'$ of ${\cD}({\cE})$ corresponds to the
configuration $c'\setminus c$ of ${\cD}({\cE}\setminus c)$.

For an event structure ${\cE}=(E,\le, \#)$, define the equivalence
relation $R_{\cE}$ on its configurations as follows: for two
configurations $c$ and $c'$, set $cR_\cE c'$ if
$\cE\setminus c\equiv \cE\setminus c'$.  The \emph{index} of $\cE$ is
the number of equivalence classes of $R_\cE$, i.e., the number of
isomorphism types of futures of configurations of $\cE$. The event
structure $\cE$ is
\emph{regular}~\cite{Thi_regular,Thi_conjecture,NiThi} if $\cE$ has
finite index and finite degree.

Given a labeled event structure $\cE^{\lambda}=(\cE,\lambda)$, for any
configuration $c$ of $\cE$, if we restrict $\lambda$ to
$\cE\setminus c$, then we obtain a labeled event structure
$(\cE\setminus c,\lambda)$ denoted by $\cE^{\lambda}\setminus
c$. Analogously, define the equivalence relation $R_{\cE^{\lambda}}$
on its configurations by setting $c R_{\cE^{\lambda}} c'$ if
$\cE^{\lambda}\setminus c\equiv \cE^{\lambda}\setminus c'$. The
index of $\cE^{\lambda}$ is the number of equivalence classes of
$R_{\cE^{\lambda}}$.  We say that an event structure $\cE$ admits a
\emph{regular nice labeling} if there exists a nice labeling $\lambda$
of $\cE$ with a finite alphabet $\Sigma$ such that $\cE^{\lambda}$ has
finite index.

We continue with the definition of trace-regular event
structures~\cite{Thi_regular,Thi_conjecture}.  For a trace alphabet
$M=(\Sigma,I)$, an $M$-\emph{labeled event structure} is a labeled
event structure $\cE^{\phi}=(\cE, \lambda)$, where $\cE=(E,\le, \#)$
is an event structure and $\lambda: E\rightarrow \Sigma$ is a labeling
function which satisfies the following conditions:
\begin{enumerate}[{(LES}1)]
\item $e\#_{\mu} e'$ implies $\lambda(e)\ne \lambda (e')$,
\item if $e\lessdot e'$ or $e\#_{\mu} e'$, then
  $(\lambda(e),\lambda(e'))\in D$,
\item if $(\lambda(e),\lambda(e'))\in D$, then $e\le e'$ or $e'\le e$
  or $e\# e'$.
\end{enumerate}
We call $\lambda$ a \emph{trace labeling} of $\cE$ with the
trace alphabet $M=(\Sigma,I)$. The conditions (LES2) and (LES3) on the
labeling function ensures that the concurrency relation $\|$ of
$\cE$ respects the independence relation $I$ of $M$.  In
particular, since $I$ is irreflexive, from (LES3) it follows that any
two concurrent events are labeled differently. Since by (LES1) two
events in minimal conflict are also labeled differently, this implies
that $\lambda$ is a finite nice labeling of $\cE$.

An $M$-labeled event structure $\cE^{\lambda}=(\cE, \lambda)$ is
\emph{regular} if ${\cE^{\lambda}}$ has finite index. Finally, an
event structure $\cE$ is called a \emph{trace-regular event
  structure}~\cite{Thi_regular,Thi_conjecture} if there exists a trace
alphabet $M=(\Sigma,I)$ and a regular $M$-labeled event structure
$\cE^{\lambda}$ such that $\cE$ is isomorphic to the underlying event
structure of $\cE^{\lambda}$.  From the definition immediately follows
that every trace-regular event structure is also a regular event
structure.

\subsection{Net Systems and their Event Structure
  Unfoldings}\label{unfold}

In the following presentation of finite 1-safe Petri nets and their
unfoldings to event structures, we closely follow the paper by
Thiagarajan and Yang~\cite{ThiaYa}.  A \emph{net system} (or,
equivalently, a \emph{finite 1-safe Petri net}) is a quadruplet
$N=(S,\Sigma,F,m_{0})$ where $S$ and $\Sigma$ are disjoint finite sets
of \emph{places} and \emph{transitions} (called also \emph{actions} or
\emph{events}), $F\subseteq (S\times \Sigma)\cup (\Sigma\times S)$ is
the \emph{flow relation}, and $m_{0}\subseteq S$ is the \emph{initial
  marking}. For $v\in S\cup \Sigma$, set $\buv=\{ u: (u,v)\in F\}$ and
$\vbu=\{ u: (v,u)\in F\}$.  A \emph{marking} of $N$ is a subset of
$S$. The \emph{transition relation} (or the \emph{firing rule})
$\longrightarrow\subseteq 2^S\times \Sigma\times 2^S$ is defined by
$m\xlongrightarrow{a} m'$ if $\bua\subseteq m$,
$(\abu-\bua)\cap m=\varnothing$, and $m'=(m-\bua)\cup \abu$.  The
transition relation $\longrightarrow$ is extended to sequences of
transitions as follows (this new relation is also denoted by
$\longrightarrow$): (1) $m\xlongrightarrow{\varepsilon} m$ for any
marking $m$ and (2) if $m\xlongrightarrow{\sigma} m'$ for
$\sigma\in \Sigma^*$ and $m'\xlongrightarrow{a}m''$ for $a\in \Sigma$,
then $m\xlongrightarrow{\sigma a}m''$. $\sigma\in \Sigma^*$ is called
a \emph{firing sequence} at $m$ if there exists a marking $m'$ such
that $m\xlongrightarrow{\sigma} m'$. Denote by $FS$ the set of firing
sequences at $m_{0}$. A marking $m$ is \emph{reachable} if there exists
a firing sequence $\sigma$ such that
$m_{0}\xlongrightarrow{\sigma}m$. We denote by $\cR(N)$ the
set of reachable markings of $N$.

Given a net system $N=(S,\Sigma,F,m_{0})$, there is a canonical way to
associate a $\Sigma$-labeled event structure ${\cE}_N$ with $N$. The
trace alphabet associated with $N$ is the pair $(\Sigma,I)$, where
$(a,b)\in I$ iff $(\abu\cup \bua)\cap (\bbu \cup
\bub)=\varnothing$. Observe that the trace alphabet $(\Sigma,I)$ is
independent of the initial marking of $N$.  A \emph{firing trace } of
$N$ is a trace $\trsigma$ where $\sigma\in FS$.  Denote by $\FT(N)$
the set of all firing traces of $N$ and by $\PFT(N)$ the subset of
$\FT(N)$ consisting of prime firing traces.

\begin{figure}
  \centering
  \includegraphics[scale=0.5]{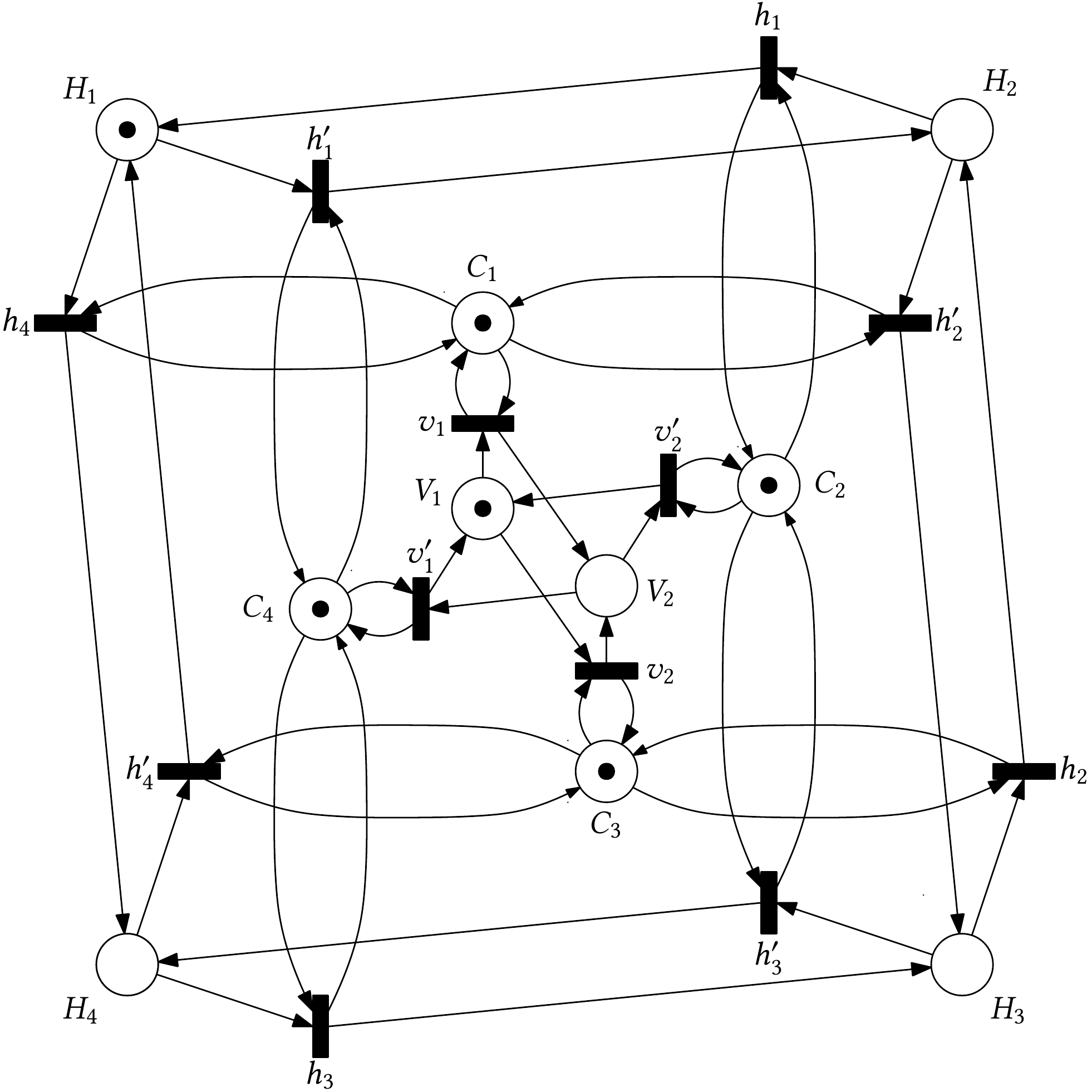}%
  \caption{The net system $N^*$ has 12 transitions (represented by
    rectangles)  and 10 places (represented by circles).}%
  \label{fig-petrinet}
\end{figure}

\begin{example}\label{ex-petrinet}
  In Fig.~\ref{fig-petrinet}, we present a net system $N^*$ with 12
  transitions $h_1$, $h_1'$, $h_2$, $h_2'$, $h_3$, $h_3'$, $h_4$,
  $h_4'$, $v_1$, $v_1'$, $v_2$, $v_2'$ and 10 places $H_1$, $H_2$,
  $H_3$, $H_4$, $V_1$, $V_2$, $C_1$, $C_2$, $C_3$, $C_4$. The initial
  marking is given by the places containing tokens in the figure.

  The trace alphabet $(\Sigma,I)$ associated with the net system $N^*$
  has 12 letters $h_1$, $h_1'$, $h_2$, $h_2'$, $h_3$, $h_3'$, $h_4$,
  $h_4'$, $v_1$, $v_1'$, $v_2$, $v_2'$. The letter $v_1$ is dependent
  from the letters $v_1'$, $v_2$, $v'_2$ (because of the place $V_1$
  and/or $V_2$), $h'_2$, and $h_4$ (because of $C_1$). The letter
  $h_1$ is dependent from the letters $h_1'$, $h_4$, $h'_4$ (because
  of the place $H_1$), $h'_2$, $h_2$ (because of $H_2$), $h'_3$, and
  $v_2'$ (because of $C_2$). For the remaining letters, the dependency
  relation is defined in a similar way.
  Observe that the letters $h_1$ and $h_3$ are independent, but there
  is no firing trace containing $h_1$ and $h_3$ as consecutive
  letters.
\end{example}

The \emph{event structure unfolding}~\cite{NiPlWi} of $N$ is the event
structure ${\cE}_N=(E,\le, \#,\lambda)$, where
\begin{itemize}
\item $E$ is the set of prime firing traces
  $\PFT(N)$,
\item $\le$ is $\sqsubseteq$, restricted to $E\times E$,
\item $e, e' \in E$ are in conflict iff there is no firing trace
  $\trsigma$ such that $e\sqsubseteq \trsigma$ and
  $e'\sqsubseteq \trsigma$,
\item $\lambda: E\rightarrow \Sigma$ is given by
  $\lambda(\trsigma)=\last(\sigma)$.
\end{itemize}

There exists an equivalence between unfoldings and trace-regular event
structures:

\begin{theorem}[{\cite[Theorem 1]{Thi_conjecture}}]\label{th:regular_trace}
  An event structure $\cE$ is a trace-regular event structure if and
  only if there exists a net system $N$ such that $\cE$ and ${\cE}_N$
  are isomorphic.
\end{theorem}

This lead Thiagarajan to conjecture
in~\cite{Thi_regular,Thi_conjecture} that

\begin{conjecture}\label{conjecture_regular}
  An event structure $\cE$ is isomorphic to the event structure
  ${\cE}_N$ arising from a finite 1-safe Petri net $N$ if and only if
  $\cE$ is regular.
\end{conjecture}

\subsection{The MSO Theory of Trace Event Structures}\label{MSO-def}

We start with the definition of monadic second-order logic
(MSO-logic). Let $A$ be a universe and $\cA=(A,{(R_i)}_{i\in I})$, where
$R_i\subseteq A^{n_i}$ for $i\in I$ be a \emph{relational structure}.
The MSO logic of $\cA$ has two types of variables: individual (or
first-order) variables and set (or second-order) variables. The
individual variables range over the elements of $A$ and are denoted by
$x,y,z$, etc. The set variables range over subsets of $A$ and are
denoted $X,Y,Z$, etc. \emph{MSO-formulas} over the signature of $\cA$
are constructed from the atomic formulas $R_i(x_1,\ldots,x_{n_i})$,
$x=y$, and $x\in X$ (where $i\in I$, $x_1,\ldots,x_{n_i},x,y$ are
individual variables and $X$ is a set variable) using the Boolean
connectives $\neg,\vee,\wedge$, and quantifications over first order
and second order variables.  The notions of free variables and bound
variables are defined as usual. A formula without free occurrences of
variables is called an \emph{MSO-sentence}.  If
$\varphi(x_1,\ldots,x_n,X_1,\ldots,X_m)$ is an MSO-formula where the
individual variables $x_1,\ldots,x_n$ and the set variables
$X_1,\ldots,X_m$ occur freely in $\varphi$, and $a_1,\ldots,a_n\in A$
and $A_1,\ldots,A_m\subseteq A$, then
${\cA}\models \varphi(a_1,\ldots,a_n,A_1,\ldots,A_m)$ means that
$\varphi$ evaluates to true in $\cA$ when $x_i$ evaluates to $a_i$ and
$X_j$ evaluates to $A_j$. The \emph{MSO theory} of $\cA$, denoted by
${\MSO}({\cA})$, is the set of all MSO-sentences $\varphi$ such that
${\cA}\models \varphi$.  The MSO theory of $\cA$ is \emph{decidable}
if there exists an algorithm deciding for each MSO-sentence $\varphi$
in $\MSO(\cA)$, whether ${\cA}\models \varphi$.

Let ${\cE}_N=(E,\le, \#,\lambda)$ be a trace-regular event structure,
which is the event structure unfolding of a net system
$N=(S,\Sigma,F,m_{0})$ (by Theorem~\ref{th:regular_trace}, any
trace-regular event structure admits such a
representation). Thiagarajan and Yang~\cite{ThiaYa} defined the MSO
theory ${\MSO}({\cE}_N)$ of ${\cE}_N$ as the MSO theory of the
relational structure $(E,{(R_a)}_{a\in \Sigma},\le)$, where $E$ is the
set of events, $R_a\subset E$ is the set of $a$-labeled events for
$a\in \Sigma$, and $\le\subseteq E\times E$ is the precedence
relation. The MSO theory of a net system $N$ is the MSO theory of its
event structure unfolding~\cite{ThiaYa}.

As shown in~\cite{ThiaYa}, the conflict relation $\#$, the concurrency
relation $\parallel$, and the notion of a configuration of $\cE$, as
well as other connectives of propositional logic such as
$\wedge,\Rightarrow$ (implies) and $\equiv$ (if and only if),
universal quantification over individual and set variables
($\forall x(\varphi), \forall X(\varphi)$), the set inclusion relation
$\subseteq $ ($X\subseteq Y$), can be defined as well. The conflict
and concurrency relations $\#$ and $\parallel$ of $\cE$ are defined
in~\cite{ThiaYa} in the following way:
\begin{itemize}
\item
  $x\widehat{\#} y:=\neg (x\le y)\wedge \neg (y\le x)\wedge
  \bigvee_{(a,b)\in D} (R_a(x)\wedge R_b(y))$.
\item
  $x\# y:= \exists x'\exists y' (x'\le x \wedge y'\le y \wedge
  x'\widehat{\#} y')$.
\item
  $x\parallel y:= \neg (x\le y)\wedge \neg (y\le x)\wedge \neg(x\#
  y)$.
\end{itemize}

An interpretation ${\cI}$ assigns to every individual variable an
event in $E$ and every set variable, a subset of $E$.  Then $\cE_N$
satisfies a formula $\varphi$ under an interpretation $\cI$ (i.e.,
$\cE_N$ is a model of the sentence $\varphi$), denoted by
${\cE_N}\models_{\cI} \varphi$, if the following holds~\cite{ThiaYa}:
\begin{itemize}
\item ${\cE_N}\models_{\cI} R_a(x)$ iff $\lambda({\cI}(x))=a$.
\item ${\cE_N}\models_{\cI} x\le y$ iff $\cI(x)\le \cI(y)$.
\item ${\cE_N}\models_{\cI} x\in X$ iff $\cI(x)\in \cI(X)$.
\item ${\cE_N}\models_{\cI} \exists x(\varphi)$ iff there exists
  $e\in E$ and an interpretation $\cI'$ such that
  $\cE\models_{\cI} \varphi$ where $\cI'$ satisfies the conditions:
  $\cI'(x)=e$, $\cI'(y)=\cI(y)$ for every individual variable $y$
  other than $x$, and $\cI'(X)=\cI(X)$ for every set variable $X$.
\item ${\cE_N}\models_{\cI} \exists X(\varphi)$ iff there exists
  $E'\subseteq E$ and an interpretation $\cI'$ such that
  $\cE\models_{\cI'} \varphi$ where $\cI'$ satisfies: $\cI'(X)=E'$,
  $\cI'(x)=\cI(x)$ for every individual variable $x$, and
  $\cI'(Y)=\cI(Y)$ for every set variable $Y$ other than $X$.
\item ${\cE_N}\models_{\cI} \neg \varphi$ and
  $\cE\models_{\cI} \varphi_1 \vee \varphi_2$ are defined in the
  standard way.
\end{itemize}

It turns out that the MSO theory of trace event structures is not
always decidable: Fig. 1 of~\cite{ThiaYa} presented an example of such
an event structure suggested by Igor Walukiewicz. To circumvent this
example, Thiagarajan and Yang formulated the following notion.

The event structure $\cE=(E,\le, \#)$ is \emph{grid-free}~\cite{ThiaYa}
if there does not exist three pairwise disjoint sets $X,Y,Z$ of $E$
satisfying the following conditions:
\begin{itemize}
\item $X=\{x_0,x_1,x_2,\ldots\}$ is an infinite set of events with
  $x_0<x_1<x_2<\cdots$.
\item $Y=\{y_0,y_1,y_2,\ldots\}$ is an infinite set of events with
  $y_0<y_1<y_2<\cdots$.
\item $X\times Y\subseteq \parallel$.
\item There exists an injective mapping $g: X\times Y\rightarrow Z$
  satisfying: if $g(x_i,y_j)=z$ then $x_i<z$ and $y_j<z$. Furthermore,
  if $i'>i$ then $x_{i'}\nless z$ and of $j'>j$ then $y_{j'}\nless z$.
\end{itemize}

The $\Sigma$-labelled event structure $(E,\le, \#, \lambda)$ is
\emph{grid-free} if the event structure $(E,\le, \#)$ is
grid-free. The net system $N$ is \emph{grid-free} if the event
structure $\cE_N$ is grid-free. As noticed in~\cite{ThiaYa},
Walukiewicz's net system is not grid-free. Thiagarajan and
Yang~\cite{ThiaYa} proved that if a net system $N$ is not grid-free,
then the MSO theory $\MSO(\cE_N)$ is not decidable. Thiagarajan
conjectured that the converse holds:

\begin{conjecture}\label{MSO}
  The $\MSO$ theory of a net system $N$ is decidable iff $N$ is
  grid-free.
\end{conjecture}

\section{Domains, Median Graphs, and CAT(0) Cube
  Complexes}\label{sec-median-CAT(0)}

In this section, we recall the bijections between domains of event
structures and median graphs/CAT(0) cube complexes established
in~\cite{ArOwSu} and~\cite{BaCo}, and between median graphs and
1-skeleta of CAT(0) cube complexes established in~\cite{Ch_CAT}
and~\cite{Ro}.

\subsection{Median Graphs}

Let $G=(V,E)$ be a simple, connected, not necessarily finite
graph. The \emph{distance} $d_G(u,v)$ between two vertices $u$ and $v$
is the length of a shortest $(u,v)$-path, and the \emph{interval}
$I(u,v)$ consists of all vertices (metrically) \emph{between} $u$ and
$v$: $I(u,v):=\{ x\in V: d_G(u,x)+d_G(x,v)=d_G(u,v)\}$. An induced
subgraph $H$ of $G$ (or the corresponding vertex set) is called
\emph{convex} if it includes the interval of $G$ between any of its
vertices. An induced subgraph $H$ of $G$ is called \emph{gated} if for
any vertex $v \in V(G) \setminus V(H)$ there exists a unique vertex
$v' \in V(H)$ such that $v' \in I(v,u)$ for any $u \in V(H)$ (the
vertex $v'$ is called the \emph{gate} of $v$ in $H$). Any gated
subgraph is convex.  A graph $G=(V,E)$ is \emph{isometrically
  embeddable} into a graph $H=(W,F)$ if there exists a mapping
$\varphi : V\rightarrow W$ such that
$d_H(\varphi (u),\varphi (v))=d_G(u,v)$ for all vertices $u,v\in V$.

A graph $G$ is called \emph{median} if $I(x,y)\cap I(y,z)\cap I(z,x)$
contains a unique vertex $m$ for each triplet $x,y,z$ of vertices ($m$
is the \emph{median} of $x,y,z$). Median graphs are bipartite.  Basic
examples of median graphs are trees, hypercubes, rectangular grids,
and Hasse diagrams of distributive lattices and of median
semilattices~\cite{BaCh_survey}.  With any vertex $v$ of a median
graph $G=(V,E)$ is associated a \emph{canonical partial order} $\le_v$
defined by setting $x\le_v y$ if and only if $x\in I(v,y)$; $v$ is
called the \emph{basepoint} of $\le_v$. Since $G$ is bipartite, the
Hasse diagram $G_v$ of the partial order $(V,\le_v)$ is the graph $G$
in which any edge $xy$ is directed from $x$ to $y$ if
$d_G(x,v)<d_G(y,v)$. We call $G_v$ a \emph{pointed median graph}.

Median graphs can be obtained from hypercubes by amalgams and median
graphs are themselves isometric subgraphs of
hypercubes~\cite{BaVdV,Mu}. The canonical isometric embedding of a
median graph $G$ into a (smallest) hypercube can be determined by the
so called \emph{Djokovi\'c-Winkler (``parallelism'')} relation
$\Theta$ on the edges of $G$~\cite{Dj,Wink}. For median graphs, the
equivalence relation $\Theta$ can be defined as follows. First say
that two edges $uv$ and $xy$ are in relation $\Theta'$ if they are
opposite edges of a $4$-cycle in $G$. Then let $\Theta$ be the
reflexive and transitive closure of $\Theta'$. Any equivalence class
of $\Theta$ constitutes a cutset of the median graph $G$, which
determines one factor of the canonical hypercube~\cite{Mu}. The cutset
(equivalence class) $\Theta(xy)$ containing an edge $xy$ defines a
convex split $\{ W(x,y),W(y,x)\}$ of $G$~\cite{Mu}, where
$W(x,y)=\{ z\in V: d_G(z,x)<d_G(z,y)\}$ and $W(y,x)=V\setminus W(x,y)$
(we call the complementary convex sets $W(x,y)$ and $W(y,x)$
\emph{halfspaces}). Conversely, for every convex split of a median
graph $G$ there exists at least one edge $xy$ such that
$\{ W(x,y),W(y,x)\}$ is the given split. We denote by
$\{ \Theta_i: i\in I\}$ the equivalence classes of the relation
$\Theta$ (in~\cite{BaCo}, they were called parallelism classes). For
an equivalence class $\Theta_i, i\in I$, we denote by $\{ A_i,B_i\}$
the associated convex split. We say that $\Theta_i$ \emph{separates}
the vertices $x$ and $y$ if $x\in A_i,y\in B_i$ or
$x\in B_i,y\in A_i$. The isometric embedding $\varphi$ of $G$ into a
hypercube is obtained by taking a basepoint $v$, setting
$\varphi(v)=\varnothing$ and for any other vertex $u$, letting
$\varphi(u)$ be all parallelism classes of $\Theta$ which separate $u$
from $v$.

From the definition it follows that any median graph $G$ satisfies the
following \emph{quadrangle condition}: for any four vertices $u,v,w,z$
with $d(v,z)=d(w,z)=1$ and $2=d(v,w)\leq d(u,v)=d(u,w)=d(u,z)-1$,
there exists a common neighbor $x$ of $v$ and $w$ such that
$d(u,x)=d(u,v)-1$. In fact, $x$ is the median of the triplet $u,v,w$
and is thus unique.

In median graphs, convex subgraphs are gated (this is not true for
general graphs). We now give a simple but useful local
characterization of convex sets of median graphs:

\begin{lemma}[\cite{Ch_metric}]\label{convex}
  A connected subgraph $S$ of a median graph $G$ is convex if and only
  if $S$ is locally-convex, i.e., $I(x,y)\subseteq S$ for any two
  vertices $x,y$ of $S$ having a common neighbor in $S$.
\end{lemma}

\subsection{Nonpositively Curved Cube Complexes}

A 0-cube is a single point. A 1-\emph{cube} is an isometric copy of
the segment $[-1,1]$ and has a cell structure consisting of 0-cells
$\{\pm 1\}$ and a single 1-cell. An $n$-\emph{cube} is an isometric
copy of ${[-1,1]}^n$, and has the product structure, so that each closed
cell of ${[-1,1]}^n$ is obtained by restricting some of the coordinates
to $+1$ and some to $-1$.  A \emph{cube complex} is obtained from a
collection of cubes of various dimensions by isometrically identifying
certain subcubes.  The \emph{dimension} $\dimension(X)$ of a cube
complex $X$ is the largest $d$ for which $X$ contains a $d$-cube. A
\emph{square complex} is a cube complex of dimension 2.  The 0-cubes
and the 1-cubes of a cube complex $X$ are called \emph{vertices} and
\emph{edges} of $X$ and define the graph $X^{(1)}$, the
\emph{$1$-skeleton} of $X$. We denote the vertices of $X^{(1)}$ by
$V(X)$ and the edges of $X^{(1)}$ by $E(X)$.  For $i\in\NN$, we denote
by $X^{(i)}$ the $i$-\emph{skeleton} of $X$, i.e., the cube complex
consisting of all $j$-dimensional cubes of $X$, where $j\le i$. A
\emph{square complex} $X$ is a combinatorial 2-complex whose 2-cells
are attached by closed combinatorial paths of length 4. Thus, one can
consider each 2-cell as a square attached to the 1-skeleton $X^{(1)}$
of $X$. All cube complexes occurring in this paper are
\emph{simple}~\cite{HaWi1} in the sense that two distinct squares
cannot meet along two consecutive edges.

The \emph{star} $\St(v,X)$ of a vertex $v$ of $X$ is the set of all
cubes containing $v$.  The \emph{link} of a vertex $x\in X$ is the
simplicial complex $\Link(x,X)$ with a $(d-1)$-simplex for each
$d$-cube containing $x$, with simplices attached according to the
attachments of the corresponding cubes.  More generally, the
\emph{link} of a $k$-dimensional cube $Q$ of $X$ is the simplicial
complex $\Link(Q,X)$ with a $(d-k-1)$-simplex for each $d$-cube
containing $Q$, with simplices attached according to the attachments
of the corresponding cubes.  Note that in the definition of the link,
the simplices are added with multiplicity: if $x$ (or $Q$) belongs to
a cube $Q'$ in multiple ways, then $Q'$ contributes to the link with
multiple (disjoint) simplices. For example, if $X$ is a 1-dimensional
complex with only one 0-cube $x$ and only one 1-cube $e$ (a loop
around $x$), then $\Link(x,X)$ consists of two disjoint $0$-simplices.

The link $\Link(x,X)$ is a \emph{flag (simplicial) complex} if each
$(d+1)$-clique in $\Link(x,X)$ spans an $d$-simplex. A cube complex
$X$ is \emph{flag} if $\Link(x,X)$ is a flag simplicial complex for
every vertex $x \in X$. This flagness condition of $\Link(x,X)$ can be
restated as follows: whenever three $(k + 2)$-cubes of ${X}$ share a
common $k$-cube containing $x$ and pairwise share common
$(k+1)$-cubes, then they are contained in a $(k+3)$--cube of $X$. A
cube complex $X$ is called \emph{simply connected} if it is connected
and if every continuous mapping of the 1-dimensional sphere $S^1$ into
$X$ can be extended to a continuous mapping of the disk $D^2$ with
boundary $S^1$ into $X$.  Note that $X$ is connected iff
$G(X)=X^{(1)}$ is connected, and $X$ is simply connected iff $X^{(2)}$
is simply connected. Equivalently, a cube complex $X$ is {simply
  connected} if $X$ is connected and every cycle $C$ of its
$1$-skeleton is null-homotopic, i.e., it can be contracted to a single
point by elementary homotopies.

Given two cube complexes $X$ and $Y$, a \emph{covering (map)} is a
surjection $\varphi\colon Y \to X$ mapping cubes to cubes, preserving
the inclusion of cubes, and such that $\varphi$ induces an isomorphism
between $\Link(v,Y)$ and $\Link(\varphi(v),X)$.  The condition on the
links is equivalent to the following condition on the stars:
$\varphi_{|\St(v,Y)}\colon \St(v,Y)\to \St(\varphi(v),X)$ is a
bijection for every vertex $v$ in $Y$.  The space $Y$ is then called a
\emph{covering space} of $X$. For any vertex $v$ of $X$, any vertex
$\tv$ of $Y$ such that $\varphi(\tv)=v$ is called a \emph{lift} of
$v$.  It is well-known that if $X$ and $Y$ are flag cube complexes,
$Y$ is a covering space of $X$ if and only if the $2$-skeleton
$Y^{(2)}$ of $Y$ is a covering space of $X^{(2)}$.  A \emph{universal
  cover} of $X$ is a simply connected covering space; it always exists
and it is unique up to isomorphism~\cite[Sections 1.3 and
4.1]{Hat}. The universal cover of a complex $X$ will be denoted by
$\tX$. In particular, if $X$ is simply connected, then its universal
cover $\tX$ is $X$ itself.

An important class of cube complexes studied in geometric group theory
and combinatorics is the class of CAT(0) cube complexes.  CAT(0)
spaces can be characterized in several different natural ways and have
many strong properties, see for
example~\cite{BrHa}. Gromov~\cite{Gromov} gave a beautiful
combinatorial characterization of CAT(0) cube complexes, which can be
also taken as their definition:

\begin{theorem}[\cite{Gromov}]\label{Gromov}
  A cube complex $X$ is CAT(0) if and only if $X$ is simply connected
  and the links of all vertices of $X$ are flag complexes.  If $Y$ is
  a cube complex in which the links of all vertices are flag
  complexes, then the universal cover $\tY$ of $Y$ is a CAT(0) cube
  complex.
\end{theorem}

In view of Theorem~\ref{Gromov}, the cube complexes in which the links
of vertices are flag complexes are called \emph{nonpositively curved
  cube complexes} or shortly \emph{NPC complexes}.  As a corollary of
Gromov's result, \emph{for any NPC complex $X$, its universal cover
  $\tX$ is CAT(0).}

A square complex $X$ is a \emph{VH-complex} (\emph{vertical-horizontal
  complex}) if the 1-cells (edges) of $X$ are partitioned into two
sets $V$ and $H$ called \emph{vertical} and \emph{horizontal} edges
respectively, and the edges in each square alternate between edges in
$V$ and $H$. Notice that if $X$ is a VH-complex, then $X$ satisfies
the Gromov's nonpositive curvature condition since no three squares
may pairwise intersect on three edges with a common vertex, i.e.,
VH-complexes are particular NPC square complexes.


We continue with the bijection between CAT(0) cube complexes and
median graphs:

\begin{theorem}[\cite{Ch_CAT,Ro}]\label{CAT(0)-median}
  Median graphs are exactly the 1-skeleta of CAT(0) cube complexes.
\end{theorem}

The proof of Theorem~\ref{CAT(0)-median} presented in~\cite{Ch_CAT} is
based on the following local-to-global characterization of median
graphs. A graph $G$ is median if and only if its cube complex is
simply connected and $G$ satisfies the 3-cube condition: if three
squares of $G$ pairwise intersect in an edge and all three intersect
in a vertex, then they belong to a 3-cube.

A \emph{midcube} of the $d$-cube $c$, with $d\geq 1$, is the isometric
subspace obtained by restricting exactly one of the coordinates of $d$
to 0.  Note that a midcube is a $(d-1)$-cube.  The midcubes $a$ and
$b$ of $X$ are \emph{adjacent} if they have a common face, and a
\emph{hyperplane} $H$ of $X$ is a subspace that is a maximal connected
union of midcubes such that, if $a,b\subset H$ are midcubes, either
$a$ and $b$ are disjoint or they are adjacent.  Equivalently, a
hyperplane $H$ is a maximal connected union of midcubes such that, for
each cube $c$, either $H\cap c=\varnothing$ or $H\cap c$ is a single
midcube of $c$. The \emph{carrier} $N(X)$ of a hyperplane $H$ of $X$ is
the union of all cubes intersecting $H$.

\begin{theorem}[\cite{Sa}]\label{Sageev}
  Each hyperplane $H$ of a CAT(0) cube complex $X$ is a CAT(0) cube
  complex of dimension $\le \dimension(X)-1$ and ${X}\setminus H$ has
  exactly two components, called \emph{halfspaces}.
\end{theorem}

A 1-cube $e$ (an edge) is \emph{dual} to the hyperplane $H$ if the
0-cubes of $e$ lie in distinct halfspaces of $X \setminus H$, i.e., if
the midpoint of $e$ is in a midcube contained in $H$.  The relation
``dual to the same hyperplane'' is an equivalence relation on the set
of edges of $X$; denote this relation by $\Theta$ and denote by
$\Theta(H)$ the equivalence class consisting of 1-cubes dual to the
hyperplane $H$ ($\Theta$ is precisely the parallelism relation on the
edges of the median graph $X^{(1)}$).
The following results summarize the well known convexity properties of
halfspaces and carriers of CAT(0) cube complexes.

\begin{theorem}[\cite{Mu,vdV}]\label{gated-carriers}
  If $H$ is a hyperplane of a CAT(0) cube complex $X$, then the
  carrier $N(H)$ of $H$ and the two halfspaces defined by $H$
  restricted to the vertices of $X$ induce convex and thus gated
  subgraphs of the 1-skeleton $G(X)$ of $X$. Any convex subgraph $H$
  of $G(X)$ is the intersection of the halfspaces of $G(X)$ containing
  $H$.
\end{theorem}

\begin{proposition}[\cite{vdV}]\label{Helly-CAT0}
  For any set $\cH$ of $d$ pairwise intersecting hyperplanes of a
  CAT(0) cube complex $X$, the carriers of the hyperplanes of $\cH$
  intersect in a $d$-cube of $X$.
\end{proposition}

\subsection{Domains versus Median Graphs/CAT(0) Cube Complexes}\label{sec-dom-med}

Theorems 2.2 and 2.3 of Barth\'elemy and Constantin~\cite{BaCo} (this
result was independently rediscovered by Ardilla et al.~\cite{ArOwSu}
in the language of CAT(0) cube complexes) establish the following
bijection between event structures and pointed median graphs
(in~\cite{BaCo}, event structures are called sites):

\begin{theorem}[\cite{BaCo}]\label{median_domain}
  The (undirected) Hasse diagram of the domain $(\cD(\cE),\subseteq)$
  of an event structure $\cE=(E,\le, \#)$ is
  median. Conversely, for any median graph $G$ and any basepoint $v$
  of $G$, the pointed median graph $G_v$ is the Hasse diagram of the
  domain of an event structure.
\end{theorem}

We briefly recall the canonical construction of an event structure
from a pointed median graph presented in~\cite{BaCo}.  Consider a
median graph $G$ and an arbitrary basepoint $v$. The events of the
event structure $\cE_v$ are the hyperplanes of $X$. Two
hyperplanes $H$ and $H'$ define concurrent events if and only if they
cross. The hyperplanes $H$ and $H'$ are in precedence relation
$H\leq H'$ if and only if $H=H'$ or $H$ separates $H'$ from $v$.
Finally, the events defined by $H$ and $H'$ are in conflict if and
only if $H$ and $H'$ do not cross and neither separates the other from
$v$.

\subsection{Special Cube Complexes}

Consider a cube complex $Y$.  Analogously to CAT(0) cube complexes,
one can define the parallelism relation $\Theta'$ on the edges of $Y$
by setting that two edges of $Y$ are in relation $\Theta'$ if they are
opposite edges of a common 2-cube of $Y$. Let $\Theta$ be the
reflexive and transitive closure of $\Theta'$ and let
$\{ \Theta_i: i\in I\}$ denote the equivalence classes of $\Theta$.
For an equivalence class $\Theta_i$, the hyperplane $H_i$ associated
to $\Theta_i$ is the cube complex consisting of the midcubes of all
cubes of $Y$ containing one edge of $\Theta_i$. The edges of
$\Theta_i$ are \emph{dual} to $H_i$. Let $\cH(Y)$ be the set of
hyperplanes of $Y$. The \emph{carrier} $N(H)$ of a hyperplane $H$ of
$Y$ is the union of all cubes intersecting $H$. The \emph{open carrier}
$\mathring{N}(H)$ is the union of all open cubes intersecting $H$.

The hyperplanes of a cube complex do not longer satisfy the nice
properties of the hyperplanes of CAT(0) cube complexes: they do not
partition the complex in two parts, they may self-intersect,
self-osculate, two hyperplanes may cross and osculate, etc. Haglund
and Wise~\cite{HaWi1} detected five types of pathologies which may
occur in a cube complex (see Fig.~\ref{fig-special}):
\begin{enumerate}[(a)]
\item self-intersecting hyperplane;
\item one-sided hyperplane;
\item directly self-osculating hyperplane;
\item indirectly self-osculating hyperplane;
\item a pair of hyperplanes, which both intersect and osculate.
\end{enumerate}

\begin{figure}
  \centering
  \includegraphics[scale=0.7]{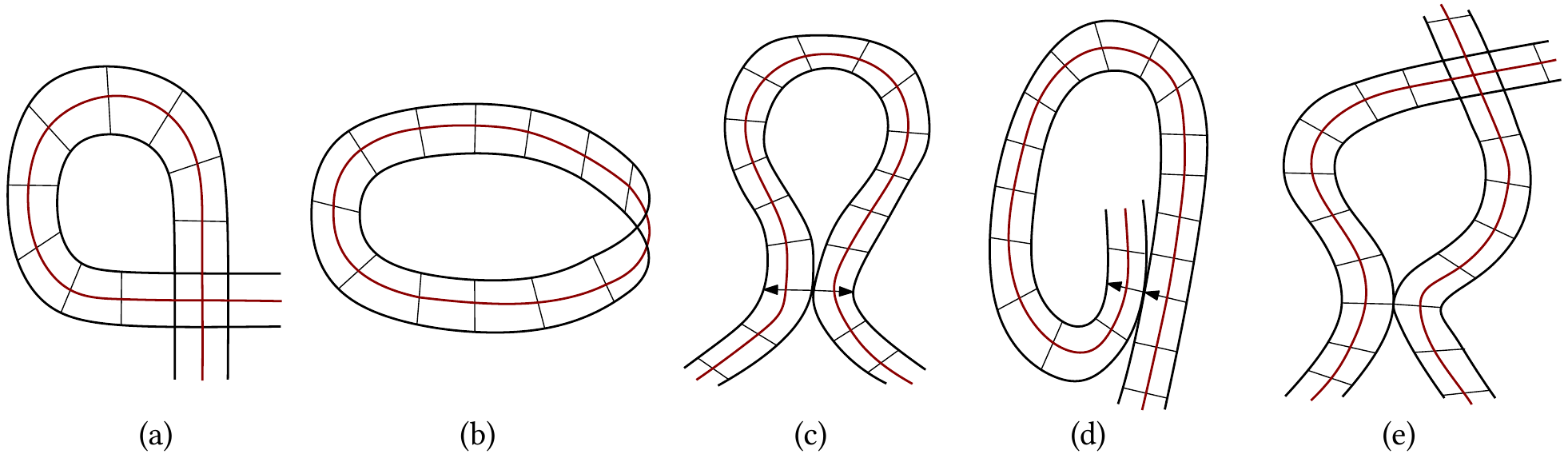}%
  \caption{A self-intersecting hyperplane~(a), a one-sided
    hyperplane~(b), a directly self-osculating hyperplane~(c), an
    indirectly self-osculating hyperplane~(d), and a pair of
    hyperplanes that inter-osculate~(e).}%
  \label{fig-special}
\end{figure}

A hyperplane $H$ is \emph{two-sided} if $\mathring{N}(H)$ is
homeomorphic to the product $H\times (-1,1)$, and there is a
combinatorial map $H\times[-1,1]\rightarrow X$ mapping $H\times \{0\}$
identically to $H$.  As noticed in~\cite[p.1562]{HaWi1}, requiring
that the hyperplanes of $Y$ are two-sided is equivalent to defining an
orientation on the dual edges of $H$ such that all sources of such
edges belong to one of the sets $H\times \{-1\}, H\times \{1\}$ and
all sinks belong to the other one. This orientation is obtained by
taking the equivalence relation generated by elementary parallelism
relation: declare two oriented edges $e_1$ and $e_2$ of $Y$
\emph{elementary parallel} if there is a square of $Y$ containing
$e_1$ and $e_2$ as opposite sides and oriented in the same
direction. Such an orientation $o$ of the edges of $Y$ is called an
\emph{admissible} orientation of $Y$. Observe that $Y$ admits an
admissible orientation if and only if every hyperplane $H$ of $Y$ is
two-sided (one can choose an admissible orientation for each
hyperplane independently). Given a cube complex $Y$ and an admissible
orientation $o$ of $Y$, $(Y,o)$ is called a \emph{directed} cube
complex.

We continue with the definition of each of the pathologies (in which
we closely follow~\cite[Section 3]{HaWi1}). The hyperplane is
\emph{one-sided} if it is not two-sided (see
Fig.~\ref{fig-special}(b)). Two hyperplanes $H_1$ and $H_2$
\emph{intersect} if there exists a cube $Q$ and two distinct midcubes
$Q_1$ and $Q_2$ of $Q$ such that $Q_1 \subseteq H_1$ and
$Q_2 \subseteq H_2$, i.e., there exists a square with two consecutive
edges $e_1,e_2$ such that $e_1$ is dual to $H_1$ and $e_2$ is dual to
$H_2$. A hyperplane $H$ of $Y$ \emph{self-intersects} if it contains
more than one midcube from the same cube, i.e., there exist two edges
$e_1,e_2$ dual to $H$ that are consecutive in some square of $Y$ (see
Fig.~\ref{fig-special}(a)). Let $v$ be a vertex of $Y$ and let
${e_1},{e_2}$ be two distinct edges incident to $v$ but such that
${e_1}$ and ${e_2}$ are not consecutive edges in some square
containing $v$. The hyperplanes $H_1$ and $H_2$ \emph{osculate} at
$(v,{e_1},{e_2})$ if $e_1$ is dual to $H_1$ and $e_2$ is dual to
$H_2$. The hyperplane $H$ \emph{self-osculate} at $(v,{e_1},{e_2})$ if
$e_1$ and $e_2$ are dual to $H$.  Consider a two-sided hyperplane $H$
and an admissible orientation $o$ of its dual edges. Suppose that $H$
self-osculate at $(v,{e_1},{e_2})$. If $v$ is the source of both $e_1$
and $e_2$ or the sink of both $e_1$ and $e_2$, then we say that $H$
\emph{directly self-osculate} at $(v,{e_1},{e_2})$ (see
Fig.~\ref{fig-special}(c)). If $v$ is the source of one of $e_1$,
$e_2$, and the sink of the other, then we say that $H$
\emph{indirectly self-osculate} at $(v,{e_1},{e_2})$ (see
Fig.~\ref{fig-special}(d)). Note that a self-osculation of a
hyperplane $H$ is either direct or indirect, and this is independent
of the orientation of the edges dual to $H$.  Two hyperplanes $H_1$
and $H_2$ \emph{inter-osculate} if they both intersect and osculate
(see Fig.~\ref{fig-special}(e)).

Haglund and Wise~\cite[Definition 3.2]{HaWi1} called a cube complex
$Y$ \emph{special} if its hyperplanes are two-sided, do not
self-intersect, do not directly self-osculate, and no two hyperplanes
inter-osculate. The definition of a special cube complex $Y$ depends
only of the 2-skeleton $Y^{(2)}$~\cite[Remark~3.4]{HaWi1}. Since no
two hyperplanes of $Y$ inter-osculate, any special cube complex and its
2-skeleton satisfy the 3-cube condition. In fact, Haglund and Wise
proved that special cube complexes can be seen as nonpositively curved
complexes:

\begin{lemma}[\cite{HaWi1}, {Lemma~3.13}]\label{cube-completable}
  If $X$ is a special cube complex, then $X$ is contained in a unique
  smallest nonpositively curved cube complex with the same
  $2$-skeleton as $X$.
\end{lemma}

In view of this lemma, we will consider only $2$-dimensional special
cube complexes, since they can always be canonically completed to NPC
complexes that are also special.

\section{Geodesic Traces and Prime Traces}\label{geometry-traces}

Let $M=(\Sigma, I)$ be a trace alphabet and let ${\cE}=(E,\le,\#)$ be
an $M$-labeled event structure. Then the concurrency relation $\|$ of
$\cE$ coincides with the independence relation $I$.  Let $\cD({\cE})$
be the domain of $\cE$, $G({\cE})$ be the covering graph of
$\cD({\cE})$, and $X({\cE})$ be the CAT(0) cube complex of $G({\cE})$
pointed at vertex $v_0$. Any vertex $v$ of the median graph $G({\cE})$
corresponds to a configuration $c(v)$ of $\cD({\cE})$; in particular,
$c(v_0)=\varnothing$.  In this section we characterize traces arising
from geodesics of the domain.

\subsection{Geodesic Traces}

Any shortest path $\pi=(v_0,v_1,\ldots,v_{k-1},v_k=v)$ from $v_0$ to a
vertex $v$ of $G({\cE})$ gives rise to a word $\sigma(\pi)$ of
$\Sigma^*$: the $i$th letter of $\sigma(\pi)$ is the label
$\lambda(v_{i-1}v_i)$ of the edge $v_{i-1}v_i$.  We say that a word
$\sigma\in \Sigma^*$ is \emph{geodesic} if $\sigma=\sigma(\pi)$ for a
shortest path $\pi$ between $v_0$ and a vertex $v$ of $G({\cE})$. The
trace $\trsigma$ of a geodesic word $\sigma$ is called a
\emph{geodesic trace}.

Two shortest $(v_0,v)$-paths $\pi$ and $\pi'$ of $G({\cE})$ are called
\emph{homotopic} if they can be transformed one into another by a
sequence of elementary homotopies, i.e., there exists a finite
sequence $\pi=:\pi_1,\pi_2,\ldots,\pi_{k-1},\pi_k:=\pi'$ of shortest
$(v_0,v)$-paths such that for any $i=1,\ldots,k-1$ the paths $\pi_i$
and $\pi_{i+1}$ differ only in a square
$Q_i=(v_{j-1},v_{j},v_{j+1},v'_{j})$ of $X({\cE})$.

The following result is well-known; we provide a simple proof using
median graphs.

\begin{lemma}\label{homotopic}
  Any two shortest $(v_0,v)$-paths $\pi_1$ and $\pi_2$ of $G({\cE})$
  are homotopic.
\end{lemma}

\begin{proof}
  We proceed by induction on the distance $k=d(v_0,v)$. If $k=2$, then
  the result is obvious because the paths $\pi_1$ and $\pi_2$ bound a
  square of $X({\cE})$. Let $w_1$ be the neighbor of $v$ in $\pi_1$
  and $w_2$ be the neighbor of $v$ in $\pi_2$. Observe that
  $d(v_0,w_1) = d(v_0,w_2) = d(v_0,v)-1 = k-1$. If $w_1 = w_2$, then
  the result holds by induction hypothesis. Otherwise, by the
  quadrangle condition, there exists a common neighbor $x$ of $w_1$
  and $w_2$ such that $d(v_0,x) = k-2$. Let $\pi_1'$ be the subpath of
  $\pi_1$ from $v_0$ to $w_1$, let $\pi_2'$ be the subpath of $\pi_2$
  from $v_0$ to $w_2$, and let $\pi_3$ be a shortest path from $v_0$
  to $x$. By induction hypothesis, the path $\pi_1'$ is homotopic to
  the path $\pi_1'' = \pi_3\cdot (x,w_1)$ and the path $\pi_2'$ is
  homotopic to the path $\pi_2'' = \pi_3\cdot (x,w_2)$.  Since
  $vw_1xw_2$ is a square of $X(\cE)$, the path $\pi_3 \cdot (x,w_1,v)$
  is homotopic to the path $\pi_3 \cdot (x,w_2,v)$. Consequently, the
  paths $\pi_1$ and $\pi_2$ are homotopic.
\end{proof}

\begin{lemma}\label{trace1}
  If $\pi$ and $\pi'$ are two shortest $(v_0,v)$-paths of $G({\cE})$,
  then $\sigma(\pi')$ belongs to $\trace{\sigma(\pi)}$.
\end{lemma}

\begin{proof}
  By Lemma~\ref{homotopic}, the paths $\pi$ and $\pi'$ are
  homotopic. Thus there exists a finite sequence
  $\pi=:\pi_1,\pi_2,\ldots,\pi_{k-1},\pi_k:=\pi'$ of shortest
  $(v_0,v)$-paths such that for any $i=1,\ldots,k-1$ the paths $\pi_i$
  and $\pi_{i+1}$ differ only in a square
  $Q_i=(v_{j-1},v_{j},v_{j+1},v'_{j})$ of $X({\cE})$.  Denote by
  $\sigma_1$ the sequence of labels of the edges of the path
  $(v_0,\ldots,v_j)$. Analogously, denote by $\sigma_2$ the sequence
  of labels of the edges of the path $(v_{j+1},\ldots,v_k=v)$.  The
  edges $v_{j-1}v_j$ and $v'_{j}v_{j+1}$ are dual to the same hyperplane
  $H_a$, thus
  $\lambda(v_{j-1}v_j)=\lambda(v'_{j}v_{j+1})=\lambda(H_a)=a$. Analogously,
  the edges $v_{j-1}v'_j$ and $v_{j}v_{j+1}$ are dual to the same
  hyperplane $H_b$, whence
  $\lambda(v_{j-1}v'_j)=\lambda(v_{j}v_{j+1})=\lambda(H_b)=b$. Since
  $H_a$ and $H_b$ intersect, the events corresponding to those
  hyperplanes are concurrent, therefore $(a,b)\in I$. Consequently,
  $\sigma(\pi_i)=\sigma_1 ab\sigma_2\leftrightarrow_I\sigma_1
  ba\sigma_2=\sigma(\pi_{i+1})$, establishing that for any two
  consecutive paths $\pi_i,\pi_{i+1}$ the words $\sigma(\pi_i)$ and
  $\sigma(\pi_{i+1})$ belong to the same trace, proving that
  $\sigma(\pi')$ belongs to $\trace{\sigma(\pi)}$.
\end{proof}

\begin{lemma}\label{trace2}
  For any shortest $(v_0,v)$-path $\pi$ of $G({\cE})$, the geodesic
  trace $\trace{\sigma(\pi)}$ consists exactly of all $\sigma(\pi')$
  such that $\pi'$ is a shortest $(v_0,v)$-path.
\end{lemma}

\begin{proof}
  From Lemma~\ref{trace1},
  $\{ \sigma(\pi'): \pi' \text{ is a shortest } (v_0,v)-{\text
    path}\}\subseteq \trace{\sigma(\pi)}$.  To prove the converse
  inclusion we show that if $\sigma_1,\sigma_2\in \Sigma^*$,
  $(a,b)\in I$ and $\sigma_1 ab\sigma_2=\sigma(\pi')$ for a shortest
  $(v_0,v)$-path $\pi'$, then $\sigma_1 ba\sigma_2=\sigma(\pi'')$ for
  a shortest $(v_0,v)$-path $\pi''$. Indeed, since $(a,b)\in I$, the
  hyperplanes $H_a$ and $H_b$ dual to the incident $a$- and $b$-edges
  of $\pi'$ intersect in an $ab$-square $Q$. Moreover, the carriers of
  $H_a$ and $H_b$ intersect in $Q$.  Since those carriers also contain
  the incident $a$- and $b$-edges of $\pi'$, $Q$ contains the $a$- and
  $b$-edges of $\pi'$. Let $\pi''$ be obtained from $\pi'$ by
  replacing the $ab$-path by the $ba$-path of $Q$. Then obviously
  $\pi''$ is a shortest $(v_0,v)$-path and that
  $\sigma(\pi'')=\sigma_1 ba\sigma_2$. 
\end{proof}

For a vertex $v$ of $G({\cE})$, we will denote by $\trace{\sigma_v}$
the Mazurkiewicz geodesic trace of all shortest $(v_0,v)$-paths, i.e.,
the trace of the interval $I(v_0,v)$.  Denote by
$\GT(\cE)$ the set of all geodesic traces of
$\cE$.  From Lemma~\ref{trace2}, we immediately obtain the following
corollary:

\begin{corollary}\label{geodesictrace}
  There exists a natural bijection $v\mapsto \trace{\sigma_v}$ between
  the set of vertices of $G(\cE)$ (i.e., the configurations of $\cE$)
  and the set $\GT(\cE)$ of geodesic traces of $\cE$.
\end{corollary}

Now we describe the precedence and the conflict relations between
geodesic traces.

\begin{lemma}\label{trace3}
  For two geodesic traces $\trace{\sigma_u}$ and $\trace{\sigma_v}$,
  we have $\trace{\sigma_u}\sqsubseteq \trace{\sigma_v}$ iff
  $u\in I(v_0,v)$.
\end{lemma}

\begin{proof}
  By definition of $\sqsubseteq$ and Lemma~\ref{trace2},
  $\trace{\sigma_u}\sqsubseteq \trace{\sigma_v}$ iff there exists a
  shortest $(v_0,u)$-path $\pi'$ and a shortest $(v_0,v)$-path $\pi$
  such that $\sigma(\pi')$ is a prefix of $\sigma(\pi)$. But this is
  equivalent to the fact that $\pi'$ is a subpath of $\pi$ which is
  equivalent to the fact that $u$ belongs to the interval $I(v_0,v)$.
\end{proof}

We say that two geodesic traces $\trace{\sigma_u}$ and
$\trace{\sigma_v}$ are in \emph{conflict} if there is no geodesic
trace $\trace{\sigma_w}$ such that
$\trace{\sigma_u}\sqsubseteq \trace{\sigma_w}$ and
$\trace{\sigma_v}\sqsubseteq \trace{\sigma_w}$. By~\ref{trace3}, this
definition can be rephrased as follows:

\begin{lemma}\label{trace4}
  Two geodesic traces $\trace{\sigma_u}$ and $\trace{\sigma_v}$ are in
  conflict iff there is no $w$ with $u,v\in I(v_0,w)$.
\end{lemma}

\subsection{Prime Geodesic Traces}

Recall that a trace $\trsigma$ is \emph{prime} if $\sigma$ is non-null
and for every $\sigma'\in \trsigma$,
$\last(\sigma)=\last(\sigma')$. We characterize now prime geodesic
traces of $\cE$, in particular we prove that they are in bijection
with the hyperplanes (events) of $\cE$.

We call an interval $I(v_0,v)$ \emph{prime} if the vertex $v$ has
degree 1 in the subgraph induced by $I(v_0,v)$.

\begin{lemma}\label{geoprimetrace1}
  A geodesic trace $\trace{\sigma_v}$ is prime iff the interval
  $I(v_0,v)$ is prime.
\end{lemma}

\begin{proof}
  If $I(v_0,v)$ is prime and $v'$ is the unique neighbor of $v$ in
  $I(v_0,v)$, then for any shortest $(v_0,v)$-path $\pi$, we have
  $\last(\sigma(\pi))=\lambda(v'v)$. Thus the geodesic trace
  $\trace{\sigma_v}$ is prime.  Conversely, if $\trace{\sigma_v}$ is
  prime, then $\last(\sigma(\pi))=\last(\sigma(\pi'))$ for any two
  shortest $(v_0,v)$-paths $\pi$ and $\pi'$. Since $v$ has a unique
  incoming edge labeled $\last(\sigma(\pi))$, necessarily $v$ has
  degree 1 in $I(v_0,v)$.
\end{proof}

\begin{lemma}\label{geoprimetrace2}
  Each hyperplane $H$ of $X({\cE})$ (i.e., each event of $\cE$) gives
  a unique prime geodesic trace $\trace{\sigma_H}:=\trace{\sigma_v}$
  defined by the prime interval $I(v_0,v)$, where $v'$ is the gate of
  $v_0$ in the carrier $N(H)$ of the hyperplane $H$ and $v$ is the
  neighbor of $v'$ such that the edge $v'v$ is dual to
  $H$. Conversely, for each prime geodesic trace $\trace{\sigma_u}$
  there exists a unique hyperplane $H$ such that
  $\trace{\sigma_u}=\trace{\sigma_H}$.
\end{lemma}

\begin{proof}
  For a hyperplane $H$, let $v'$ and $v$ be defined as in the
  formulation of the lemma. Let $A$ and $B$ be the two complementary
  halfspaces defined by $H$ and suppose that $v'\in A$ and $v\in B$.
  We assert that the interval $I(v_0,v)$ is prime, i.e., that $v'$ is
  the unique neighbor of $v$ in $I(v_0,v)$. Suppose by way of
  contradiction that $v$ is adjacent in $I(v_0,v)$ to another vertex
  $v''$. Since $v\in I(v',v'')$ and the halfspace $A$ is convex, $v''$
  cannot belong to this halfspace. Thus $v''$ belongs to $B$.  By
  quadrangle condition, there exists a vertex $w$ adjacent to $v',v''$
  and one step closer to $v_0$ than $v'$ and $v''$. From the convexity
  of $B$ we conclude that $w$ belongs to $A$. Since $w\in A$ is
  adjacent to $v''\in B$, $w$ belongs to the carrier $N(H)$ of
  $H$. Since $d(v_0,w)<d(v_0,v')$, this contradicts the assumption
  that $v'$ is the gate of $v_0$ in $N(X)$. This establishes that the
  interval $I(v_0,v)$ is prime.
  Conversely, let $\trace{\sigma_{u}}$ be a prime geodesic trace and
  let $u'$ be the unique neighbor of $u$ in $I(v_0,u)$. Let $H$ be the
  hyperplane dual to the edge $u'u$ and $A$ and $B$ be the halfspaces
  defined by $H$ with $u'\in A, u\in B$. We assert that $u'$ is the
  gate of $v_0$ in the carrier $N(H)$ of $H$.  Suppose that this is
  not true and let $v'$ be the gate of $v_0$ in $N(H)$. Let $v'v$ be
  the edge incident to $v'$ and dual to $H$. Then $v'\in I(v_0,u)$ and
  $v\in I(v',u)\subseteq I(v_0,u)$. Since $B$ is convex,
  $I(v,u)\subseteq B$, thus $u$ has a second neighbor in $I(v_0,u)$,
  contrary to the assumption that the interval $I(v_0,u)$ is
  prime. Hence $u'$ is the gate of $v_0$ in $N(H)$, establishing that
  $\trace{\sigma_u}=\trace{\sigma_H}$.
\end{proof}

Notice that the vertices $v$ such that the interval $I(v_0,v)$ is
prime are exactly the join irreducible elements of the poset
$(\cD(\cE),\subseteq)$ (i.e., the nonminimal elements which cannot be
written as the supremum of finitely many other elements). The
bijection between the set $\JJ(X({\cE}))$ of join irreducibles and the
set $\cH$ of hyperplanes was also established in~\cite{ArOwSu}.

Let $\PGT(\cE)$ be the set of geodesic prime traces of $\cE$.
From Lemma~\ref{geoprimetrace2}, we get the following:

\begin{corollary}\label{geodesicprimetrace}
  There exist natural bijections $H\mapsto \trace{\sigma_v}$ between
  the set of hyperplanes of $X(\cE)$ (i.e., events of $\cE$), the set
  $\PGT(\cE)$ of prime geodesic traces
  of $\cE$, and the set $\JJ(X({\cE}))$ of join irreducible elements
  of $(\cD,\subseteq)$.
\end{corollary}

\begin{lemma}\label{geoprimetrace3}
  For two hyperplanes $H',H$ of $X({\cE})$ with respective prime
  geodesic traces $\trace{\sigma_u}$ and $\trace{\sigma_v}$, $H'\le H$
  iff $\trace{\sigma_{u}}\sqsubseteq\trace{\sigma_{v}}$.
\end{lemma}

\begin{proof}
  By Lemma~\ref{trace3} it suffices to show that $H'\le H$ iff
  $u\in I(v_0,v)$. Let $A',B'$ and $A,B$ be the complementary
  halfspaces defined by $H'$ and $H$, respectively, and suppose that
  $v_0\in A'\cap A$. Let $u'$ be the gate of $v_0$ in $N(H')$ and $v'$
  be the gate of $v_0$ in $N(H)$. Then $u'\in A'$ and $u$ is the
  neighbor of $u'$ in $B'$.  Analogously, $v'\in A$ and $v$ is the
  neighbor of $v'$ in $B$. Notice also that $u$ is the gate of $v_0$
  in $B'$ and that $v$ is the gate of $v_0$ in $B$.
  First suppose that $H'\le H$, i.e., $H'$ separates $v_0$ from
  $H$. This is equivalent to the inclusion $B\subseteq B'$. Since $u$
  is the gate of $v_0$ in $B'$ and $v\in B$, this implies that
  $u\in I(v_0,v)$. Conversely, let $u\in I(v_0,v)$. Suppose by way of
  contradiction that $B\nsubseteq B'$, i.e., there exists
  $x\in B\cap A'$. Since $u'\in I(v_0,u)\subset I(v_0,v)$, $v$ belongs
  to the halfspace $B'$.  On the other hand, since $v$ is the gate of
  $v_0$ in $B$ and $x\in B$, we conclude that $v\in I(v_0,x)$. Since
  $v_0,x\in A'$ and $v\in B'$, this contradicts the convexity of
  $A'$. This establishes that $u\in I(v_0,v)$ implies $H'\le H$.
\end{proof}

\section{Directed NPC Complexes}\label{sec-directed-npc}

Since we can define event structures from their domains, universal
covers of NPC complexes represent a rich source of event structures.
To obtain regular event structures, it is natural to consider
universal covers of finite NPC complexes. Since domains of event
structures are directed, it is natural to consider finite NPC
complexes with directed edges.  However, the resulting directed
universal covers are not in general domains of event structures. In
particular, the domains corresponding to pointed median graphs given
by Theorem~\ref{median_domain} cannot be obtained in this way. In
order to overcome this difficulty, we introduced in~\cite{CC-thiag}
directed median graphs and directed NPC complexes. Using them, we
constructed regular event structures starting from finite directed NPC
complexes. In this section, we recall and extend these definitions and
constructions.

\subsection{Directed Median Graphs}

A \emph{directed median graph} is a pair $\oG=(G,o)$, where $G$ is a
median graph and $o$ is an orientation of the edges of $G$ in a such a
way that opposite edges of squares of $G$ have the same direction. By
transitivity of $\Theta$, all edges from the same parallelism class
$\Theta_i$ of $G$ have the same direction. Since each $\Theta_i$
partitions $G$ into two parts, $o$ defines a partial order $\prec_o$
on the vertex-set of $G$.  For a vertex $v$ of $G$, let
$\cF_{o}(v,G)=\{ x\in V: v \prec_o x\}$ be the principal filter of $v$
in the partial order $(V(G),\prec_o)$.  For any canonical basepoint
order $\le_v$ of $G$, $(G,\le_v)$ is a directed median graph. The
converse is obviously not true: the 4-regular tree $F_4$ directed so
that each vertex has two incoming and two outgoing arcs is a directed
median graph which is not induced by a basepoint order.

\begin{lemma}[\cite{CC-thiag}]\label{directed-median}
  For any vertex $v$ of a directed median graph $\oG=(G,o)$, we have:
  \begin{enumerate}
  \item $\cF_{o}(v,G)$ induces a convex subgraph of $G$;
  \item the restriction of the partial order $\prec_o$ on
    ${\cF}_{o}(v,G)$ coincides with the restriction of the
    canonical basepoint order $\le_v$ on $\cF_{o}(v,G)$;
  \item $\cF_{o}(v,G)$ together with $\prec_o$ is the domain of an
    event structure;
  \item for any vertex $u\in \cF_{o}(v,G)$, the principal filter
    $\cF_{o}(u,G)$ is included in $\cF_{o}(v,G)$ and coincides with
    the principal filter of $u$ with respect to the canonical
    basepoint order $\le_v$ on $\cF_{o}(v,G)$.
  \end{enumerate}
\end{lemma}

A \emph{directed $(x,y)$-path} of a directed median graph $\oG=(G,o)$
is a $(x,y)$-path $\pi(x,y)=(x=x_1,x_2,\ldots,x_{k-1},x_k=y)$ of $G$
in which any edge $x_{i}x_{i+1}$ is directed in $\oG$ from $x_i$ to
$x_{i+1}$.

\begin{lemma}\label{directedpath}
  Any directed path of a directed median graph $\oG$ is a shortest
  path of $G$.
\end{lemma}

\begin{proof}
  Since halfspaces of $G$ are convex, a path $\pi(x,y)$ of $G$ is a
  shortest path if and only if any hyperplane $H$ of $G$ intersects
  $\pi(x,y)$ in at most one edge. Since all edges of $G$ dual to the
  same hyperplane $H$ are directed in $\oG$ in the same way, $H$
  intersects a directed path $\pi(x,y)$ of $\oG$ in at most one
  edge. Hence the support of $\pi(x,y)$ is a shortest $(x,y)$-path in
  $G$.
\end{proof}

\subsection{Directed NPC Complexes}

A \emph{directed NPC complex} is a directed cube complex $(Y,o)$,
where $Y$ is a NPC complex and $o$ is an admissible orientation of
$Y$.  Recall that this means that $o$ is an orientation of the edges
of $Y$ in a such a way that the opposite edges of the same square of
$Y$ have the same direction.  For an edge $xy$, we denote $o(xy)$ by
$\ovr{xy}$ if the edge is directed from $x$ to $y$.  Note that there
exist NPC complexes that do not admit any admissible orientation:
consider a M{\"o}bius band of squares, for example.  An admissible
orientation $o$ of $Y$ induces in a natural way an orientation
$\tildo$ of the edges of its universal cover $\tY$, so that
$(\tY,\tildo)$ is a directed CAT(0) cube complex and
$(\tY^{(1)},\tildo)$ is a directed median graph.

In the following, we need to consider directed colored NPC complexes
and directed colored median graphs. A coloring $\nu$ of a directed NPC
complex $(Y,o)$ is an arbitrary map $\nu: E(Y) \to \Sigma$ where
$\Sigma$ is a set of colors. Note that a labeling is a coloring, but
not the converse: \emph{labelings} are precisely the colorings in
which opposite edges of any square have the same color.  In the
following, we denote a directed colored NPC complexes by bold letters
like $\bY = (Y,o,\nu)$.  Sometimes, we need to forget the colors and
the orientations of the edges of these complexes. For a complex $\bY$,
we denote by $Y$ the complex obtained by forgetting the colors and the
orientations of the edges of $\bY$ ($Y$ is called the \emph{support}
of $\bY$), and we denote by $(Y,o)$ the directed complex obtained by
forgetting the colors of $\bY$. We also consider directed colored
median graphs that are the $1$-skeletons of directed colored CAT(0)
cube complexes. Again we denote such directed colored median graphs by
bold letters like $\bG = (G,o,\nu)$. Note that (uncolored) directed
NPC complexes can be viewed as directed colored NPC complexes where
all edges have the same color.
When dealing with directed colored NPC complexes, we consider only
homomorphisms that preserve the colors and the directions of edges.
More precisely, $\bY' = (Y',o',\nu')$ is a covering of
$\bY = (Y,o,\nu)$ via a covering map $\varphi$ if $Y'$ is a covering
of $Y$ via $\varphi$ and for any edge $e \in E(Y')$ directed from $s$
to $t$, $\nu(\varphi(e)) = \nu'(e)$ and $\varphi(e)$ is directed from
$\varphi(s)$ to $\varphi(t)$. Since any coloring $\nu$ of a directed
colored NPC complex $Y$ leads to a coloring of its universal cover
$\tY$, one can consider the colored universal cover
$\tbY = (\tY,\tildo,\tnu)$ of $\bY$.
When we consider principal filters in directed colored median graphs
$\bG=(G,o,\nu)$ (in particular, when $G$ is the $1$-skeleton of the
universal cover $\tbY$ of a directed colored NPC complex $\bY$), we
say that two filters are isomorphic if there is an isomorphism between
them that preserves the directions and the colors of the edges.

We now formulate the crucial regularity property of
$(\tY^{(1)},\tildo,\tnu)$ when $(Y,o,\nu)$ is finite.

\begin{lemma}[\cite{CC-thiag}]\label{regular-cube}
  If $\bY = (Y,o,\nu)$ is a finite directed colored NPC complex, then
  $\tbY^{(1)} = (\tY^{(1)},\tildo,\tnu)$ is a directed median graph
  with at most $|V(Y)|$ isomorphism types of colored principal
  filters.  In particular, if $(Y,o)$ is a finite directed NPC
  complex, then $(\tY^{(1)},\tildo)$ is a directed median graph with
  at most $|V(Y)|$ isomorphism types of principal filters.
\end{lemma}

\begin{proposition}[\cite{CC-thiag}]\label{regular-finite-npc}
  Consider a finite (uncolored) directed NPC complex $(Y,o)$. Then for
  any vertex $\tv$ of the universal cover $\tY$ of $Y$, the principal
  filter $\cF_{\tildo}(\tv,\tY^{(1)})$ with the partial order
  $\prec_{\tildo}$ is the domain of a regular event structure with at
  most $|V(Y)|$ different isomorphism types of principal filters.
\end{proposition}

We call an event structure $\cE=(E,\le, \#)$ and its domain $\cD(\cE)$
\emph{strongly regular} if $\cD(\cE)$ is isomorphic to a principal
filter of the universal cover of some finite directed NPC complex. In
view of Proposition~\ref{regular-finite-npc}, any strongly regular
event structure is regular.

\section{Directed Special Cube Complexes}\label{sec-dspec}

\subsection{The Results}

Consider a finite NPC complex $Y$ and let $\cH = \cH(Y)$ be the set of
hyperplanes of $Y$. We define a canonical labeling
$\lambda_\cH: E(Y) \to \cH$ by setting $\lambda_\cH(e) = H$ if the
edge $e$ is dual to $H$. For any covering map $\varphi: \tY \to Y$,
$\lambda_\cH$ is naturally extended to a labeling $\tlambda_{\cH}$ of
$E(\tY)$ by setting $\tlambda_{\cH}(e) =
\lambda_\cH(\varphi(e))$. In~\cite{CC-thiag} we proved that the
strongly regular event structures obtained from finite special cube
complexes are trace-regular event structures and that this
characterizes special cube complexes:

\begin{proposition}[\cite{CC-thiag}]\label{special-trace}
  A finite NPC complex $Y$ with two-sided hyperplanes is special if
  and only if there exists an independence relation $I$ on
  $\cH = \cH(Y)$ such that for any admissible orientation $o$ of $Y$,
  for any covering map $\varphi: \tY \to Y$, and for any principal
  filter $\cD = (\cF_{\tildo}(\tv,\tY^{(1)}),\prec_{\tildo})$ of
  $(\tY,\tildo)$, the canonical labeling $\tlambda_{\cH}$ is a
  trace-regular labeling of $\cD$ with the trace alphabet $(\cH,I)$.
\end{proposition}

A finite NPC complex $X$ is called \emph{virtually
  special}~\cite{HaWi1,HaWi2} if $X$ admits a finite special cover,
i.e., there exists a finite special NPC complex $Y$ and a covering map
$\varphi: Y\rightarrow X$.  We call a strongly regular event structure
$\cE=(E,\le, \#)$ and its domain $\cD(\cE)$ \emph{cover-special} if
$\cD(\cE)$ is isomorphic to a principal filter of the universal cover
of a virtually special complex with an admissible orientation.

\begin{theorem}[\cite{CC-thiag}]\label{virtuallyspecial}
  Any cover-special event structure $\cE$ admits a trace-regular
  labeling, i.e., Thiagarajan's Conjecture~\ref{conjecture_regular} is
  true for cover-special event structures.
\end{theorem}

In the following, we need an extension of
Proposition~\ref{special-trace}.  Let $Y$ be a finite cube complex
with two-sided hyperplanes with an admissible orientation $o$.  Since
the hyperplanes of $Y$ are two-sided, there exists a bijection between
the labelings of the edges of $Y$ (i.e., colorings in which opposite
edges of each square have equal colors) and the labelings of the
hyperplanes of $Y$.  Let $M=(\Sigma,I)$ be a trace alphabet.
Extending the definition of trace labelings of domains of event
structures, we call a labeling $\lambda: E(Y) \to \Sigma$ of $(Y,o)$ a
\emph{trace labeling} if the following conditions hold:
\begin{enumerate}[{(TL}1)]
\item if there exists a square of $Y$ in which two opposite edges are
  labeled $a$ and two other opposite edges are labeled $b$, then
  $(a,b) \in I$;
\item for any vertex $v$ of $Y$, any two distinct outgoing edges
  $\ovr{vx},\ovr{vy}$ have different labels and
  $(\lambda(\ovr{vx}),\lambda(\ovr{vy}))\in I$ iff $\ovr{vx}$ and
  $\ovr{vy}$ belong to a common square of $Y$;
\item $(\lambda(\ovr{xv}),\lambda(\ovr{vy}))\in I$ iff $\ovr{xv}$ and
  $\ovr{vy}$ belong to a common square of $Y$;
\item for any vertex $v$ of $Y$, any two distinct incoming edges
  $\ovr{xv},\ovr{yv}$ have different labels and
  $(\lambda(\ovr{xv}),\lambda(\ovr{yv}))\in I$ iff $\ovr{xv}$ and
  $\ovr{yv}$ belong to a common square of $Y$.
\end{enumerate}

Since for a trace labeling $\lambda$ all edges dual to a hyperplane of
$Y$ have the same label, $\lambda$ defines in a canonical way a
labeling $\lambda: \cH\to \Sigma$ of the hyperplanes $\cH$ of $Y$: for
a hyperplane $H$, $\lambda(H)=\lambda(e)$ for any edge $e$ dual to
$H$.  Notationally, for an edge $xy$ of $Y$ directed from $x$ to $y$
and its dual hyperplane $H$, we  write
$\lambda(xy)=\lambda(\ovr{xy})=\lambda(H)$ to denote the (same) label
of $xy$, $\ovr{xy}$, and $H$.

\begin{remark}\label{rem-TL}
  Note that (TL1) is a consequence of the axioms (TL2)-(TL4).  Observe
  that (TL2)-(TL4) are equivalent to the condition that for any two
  incident edges $e_1,e_2 \in Y$, $(\lambda(e_1),\lambda(e_2)) \in I$
  iff $e_1$ and $e_2$ belong to a common square of $Y$.
\end{remark}

\begin{remark}
  Even if formulated differently, the axioms (TL1)-(TL3) can be viewed
  as a ``local'' reformulation of the axioms (LES1)-(LES3).  On the
  other hand, for domains of event structures the axiom (TL4) is
  implied by the axioms (LES1)-(LES3) because such domains are pointed
  median graphs and thus any two directed edges with the same sink
  belong to a square.
\end{remark}

The existence of trace labelings characterizes the special cube
complexes:

\begin{theorem}\label{special-trace-bis}
  For a finite cube complex $Y$ with two-sided hyperplanes, the
  following conditions are equivalent:
  \begin{enumerate}
  \item $Y$ is special;
  \item for any admissible orientation $o$ of $Y$ there exists a trace
    labeling $\lambda$ of $(Y,o)$;
  \item there exists an admissible orientation $o$ of $Y$ such that
    $(Y,o)$ admits a trace labeling.
  \end{enumerate}
\end{theorem}

In view of Theorem~\ref{special-trace-bis}, if a finite directed cube
complex $(Y,o)$ is given with a trace labeling, $Y$ is supposed to be
special. Given a directed special cube complex $(Y,o)$ with a trace
labeling $\lambda: E(Y) \to \Sigma$ for a trace alphabet
$M=(\Sigma, I)$, let $\tbY=(\tY,\tildo,\tlambda)$ denotes the directed
labeled universal cover of $Y$ and let $\varphi: \tY \to Y$ denotes
the covering map.

\begin{proposition}\label{special-trace1}
  Let $(Y,o)$ be a finite directed special cube complex with two-sided
  hyperplanes, $M=(\Sigma,I)$ be a trace alphabet, and
  $\lambda: E(Y)\to \Sigma$ be a trace labeling of $Y$.  Then for any
  principal filter $\cD=(\cF_{\tildo}(\tv,\tY^{(1)}),\prec_{\tildo})$
  of $\tbY=(\tY,\tildo)$, the labeling $\tlambda$ is a trace-regular
  labeling of $\cD$.
\end{proposition}

\begin{remark}
  Since a trace labeling of a directed special cube complex $(Y,o)$
  satisfies (TL4), $\tlambda$ is a trace-regular labeling not only of
  $(\cF_{\tildo}(\tv,\tY^{(1)}),\prec_{\tildo})$ but also of the whole
  complex $\tbY=(\tY,\tildo)$.
\end{remark}

\begin{remark}\label{remark-trace-lab-special}
  The canonical labeling $\lambda_\cH$ in
  Proposition~\ref{special-trace} is a trace labeling because all
  hyperplanes are labeled differently. Thus,
  Theorem~\ref{special-trace-bis} and Proposition~\ref{special-trace1}
  can be viewed as extensions of
  Proposition~\ref{special-trace}. However,
  Proposition~\ref{special-trace} cannot be directly deduced from
  Theorem~\ref{special-trace-bis} and Proposition~\ref{special-trace1}
  because in those  results we assume that trace labelings of
  complexes satisfy (TL4).
\end{remark}

\subsection{Proof of Theorem~\ref{special-trace-bis}}

The implication (1)$\Rightarrow$(2) follows from
Proposition~\ref{special-trace} while (2)$\Rightarrow$(3) is trivial.
To prove (3)$\Rightarrow$(1) suppose that $o$ is an admissible
orientation of $Y$ such that $(Y,o)$ admits a trace labeling
$\lambda: E(Y)\to \Sigma$ with the trace alphabet $M=(\Sigma,I)$.  We
show that $Y$ is special.
First, if $Y$ contains a self-intersecting hyperplane $H$, then there
exists a square $Q$ such that the four edges of $Q$ are dual to
$H$. Consequently, the four directed edges of $Q$ are labeled
$\lambda(H)$. Since $o$ is an admissible orientation, there exists a
vertex $v \in Q$ that has two outgoing edges with the same label,
contradicting (TL2).
Now suppose that $Y$ contains a hyperplane $H$ that directly
self-osculate at $(v,{e_1},{e_2})$. Let $e_1=xv$ and $e_2=yv$, and
observe that with respect to the orientation $o$, either
$e_1 = \ovr{vx}$ and $e_2 = \ovr{vy}$, or $e_1 = \ovr{xv}$ and
$e_2 = \ovr{yv}$. This contradicts (TL2) in the first case and (TL4)
in the second case, since $\lambda(e_1) = \lambda(e_2)$.
Finally, if $Y$ contains two hyperplanes $H_1$ and $H_2$ that
inter-osculate, then they osculate at $(v,e'_1,e'_2)$ and they
intersect on a square $Q$.  Let $Q=(e_1,e_2,e_3,e_4)$. Suppose that
the edges $e'_1,e_1,e_3$ are dual to $H_1$ and $e'_2,e_2,e_4$ are dual
to $H_2$.  Hence
$\lambda(e'_1)=\lambda(e_1)=\lambda(e_3)=\lambda(H_1)$ and
$\lambda(e'_2)=\lambda(e_2)=\lambda(e_4)=\lambda(H_2)$.  Since $o$ is
an admissible orientation, $Q$ has a source $s$. By (TL2) applied at
$s$, $(\lambda(e_1),\lambda(e_2)) \in I$. But if we consider the edges
$e_1', e_2'$ incident to $v$, by Remark~\ref{rem-TL},
$(\lambda(e_1),\lambda(e_2)) \notin I$, a contradiction.

\subsection{Proof of Proposition~\ref{special-trace1}}

By Proposition~\ref{regular-finite-npc},
$\cD=(\cF_{\tildo}(\tv,\tY^{(1)}),\prec_{\tildo})$ is the domain of a
regular event structure $\cE$. As explained in
Subsection~\ref{sec-dom-med}, the events of $\cE$ are the hyperplanes
of $\cD$.  Hyperplanes $\tH$ and $\tH'$ are concurrent if and only if
they cross, and $\tH \leq \tH'$ if and only if $\tH=\tH'$ or $\tH$
separates $\tH'$ from $\tv$.  The events $\tH$ and $\tH'$ are in
conflict iff $\tH$ and $\tH'$ do not cross and neither separates the
other from $v$. Note that this implies that $\tH \lessdot \tH'$ iff
$\tH$ separate $\tH'$ from $v$ and $\tH$ and $\tH'$ osculate, and
$\tH \#_\mu \tH'$ iff $\tH$ and $\tH'$ osculate and neither of $\tH$
and $\tH'$ separates the other from $v$. Notice also that each
hyperplane $\tH'$ of $\cD$ is the intersection of a hyperplane $\tH$
of $\tY$ with $\cD$.

\begin{claim}
  $\tlambda$ is a trace-regular labeling of $\cD$ with the trace
  alphabet $(\Sigma,I)$.
\end{claim}

First note that if $\te_1,\te_2$ are opposite edges of a square of
$\cD$, then $e_1 = \varphi(\te_1)$ and $e_2 = \varphi(\te_2)$ are
opposite edges of a square of $Y$ and thus
$\tlambda(\te_1) = \lambda(e_1) = \lambda(e_2) =
\tlambda(\te_2)$. Consequently, $\tlambda$ is a labeling of the edges
of $\cD$. Since each labeling is a coloring, from
Lemma~\ref{regular-cube}, $\cD$ has at most $|V(Y)|$ isomorphism types
of labeled principal filters. Therefore, in order to show that
$\tlambda$ is a trace-regular labeling of $\cD$, we just need to show
that $\tlambda$ satisfies the conditions (LES1),(LES2), and (LES3).

For any two hyperplanes $\tH_1, \tH_2$ in minimal conflict in $\cD$,
there exist edges $\te_1$ dual to $\tH_1$ and $\te_2$ dual to $\tH_2$
such that $\te_1$ and $\te_2$ have the same source $\tu$. Note that
since $\tH_1$ and $\tH_2$ are in conflict, $\te_1$ and $\te_2$ do not
belong to a common square of $\cD$. Moreover, if $\te_1$ and $\te_2$
are in a square $\tQ$ in $\tY$, then since there is a directed path
from $\tv$ to $\tu$, and since $\tu$ is the source of $\tQ$, all
vertices of $\tQ$ are in
$(\cF_{\tildo}(\tv,\tY^{(1)}),\prec_{\tildo}) = \cD$. Consequently,
the hyperplanes $\tH_1$ and $\tH_2$ osculate at $(\tu,\te_1,\te_2)$ in
$\tY$. Let $u = \varphi(\tu)$, $e_1 = \varphi(\te_1)$, and
$e_2 = \varphi(\te_2)$, and note that $u$ is the source of $e_1$ and
$e_2$. Let $H_1$ and $H_2$ be the hyperplanes of $Y$ that are
respectively dual to $e_1$ and $e_2$. Since $\varphi$ is a covering
map, $e_1$ and $e_2$ do not belong to a common square. Therefore
$\lambda(e_1)\ne \lambda(e_2)$ and
$(\lambda(e_1),\lambda(e_2))\notin I$.  Since
$\tlambda(\tH_1)=\lambda(e_1)$ and $\tlambda(\tH_2)=\lambda(e_2)$,
this establishes (LES1).  This also establishes (LES2) when
$\tH_1 \#_\mu \tH_2$.

Suppose now that $\tH_1 \lessdot \tH_2$ in $\cD$. There exist edges
$\te_1$ dual to $\tH_1$ and $\te_2$ dual to $\tH_2$ such that the sink
$\tu$ of $\te_1$ is the source of $\te_2$. Since $\tH_1$ separates
$\tH_2$ from $\tv$ in $\cD$, $\tH_1$ also separates $\tH_2$ from $\tv$
in $\tY$. Consequently, $\te_1$ and $\te_2$ do not belong to a common
square of $\tY$ and the hyperplanes $\tH_1$ and $\tH_2$ osculate at
$(\tu,\te_1,\te_2)$. Let $u = \varphi(\tu)$, $e_1 = \varphi(\te_1)$,
and $e_2 = \varphi(\te_2)$, and note that $u$ is the sink of $e_1$ and
the source of $e_2$.  Since $\varphi$ is a covering map, $e_1$ and
$e_2$ do not belong to a common square and thus
$(\lambda(e_1),\lambda(e_2))\notin I$. Since
$\tlambda(\tH_1)=\lambda(e_1)$ and $\tlambda(\tH_2)=\lambda(e_2)$,
this establishes (LES2) when $\tH_1 \lessdot \tH_2$.

We prove (LES3) by contraposition. Consider two hyperplanes
$\tH_1, \tH_2$ that are concurrent, i.e., they intersect in
$\cD$. Since $\tH_1$ and $\tH_2$ intersect in $\tY$, there exists a
square $\tQ$ containing two consecutive edges $\te_1$ dual to $\tH_1$
and $\te_2$ dual to $\tH_2$.  Let $H_1$ and $H_2$ be the hyperplanes
of $Y$ that are dual to $e_1 =\varphi(\te_1)$ and
$e_2 = \varphi(\te_2)$. Note that $\tlambda(\te_1) = \lambda(e_1)$ and
$\tlambda(\te_2) = \lambda(e_2)$. Since $\varphi$ is a covering map,
$e_1$ and $e_2$ belong to a square in $Y$.  Therefore
$(\tlambda(\tH_1), \tlambda(\tH_2))=(\lambda(e_1),\lambda(e_2)) \in
I$, establishing (LES3).

\section{1-Safe Petri Nets and Special Cube Complexes}

\subsection{The Results}

In this section we present the first main result of the paper, namely
we show that to any net system $N=(S,\Sigma,F,m_{0})$ one can
associate a finite directed special cube complex ${\bX}_N=(X_N,o)$
with a trace labeling $\lambda_N: E(X_N)\to \Sigma$ such that the
domain $\cD(\cE_N)$ of the event structure unfolding
$\cE_N$ of $N$ is a principal filter of the universal cover
$\tbX_N$ of ${\bX}_N$.

Let $N=(S,\Sigma,F,m_{0})$ be a net system.  The transition relation
$\longrightarrow\subseteq 2^S\times \Sigma\times 2^S$ defines a
directed graph whose vertices are all markings of $N$ and there is an
arc from a marking $m$ to a marking $m'$ iff there exists a transition
$a\in \Sigma$ such that $m\xlongrightarrow{a} m'$ (i.e.,
$\bua\subseteq m$, $(\abu-\bua)\cap m=\varnothing$, and
$m'=(m-\bua)\cup \abu$).  Denote by $G_N$ the connected component of
the support of this graph that contains the initial marking $m_{0}$
and call the undirected graph $G_N$ the \emph{marking graph} of $N$.
Let $\oG_N=(G_N, o)$ denote $G_N$ whose edges are oriented according
to $\longrightarrow$ (for notational conveniences we use $o$ instead
of $\longrightarrow$) and call $\oG_N$ the \emph{directed marking
  graph}. The marking graph $G_N$ contains all markings reachable from
$m_{0}$ but it may also contain other markings. Notice also that the
directed marking graph $\oG_N$ is deterministic and codeterministic,
i.e., for any vertex $m$ and any transition $a\in \Sigma$ there exists
at most one arc $m\xlongrightarrow{a} m'$ and at most one arc
$m''\xlongrightarrow{a} m$. We say that two distinct transitions
$a,b\in \Sigma$ are \emph{independent} if
$(\bua\cup \abu)\cap (\bub\cup \bbu)=\varnothing$.  Consider the trace
alphabet $(\Sigma,I)$ where $(a,b) \in I$ if and only if the
transitions $a$ and $b$ are independent.

\begin{definition}
  The $2$-dimensional \emph{cube complex} $X_N$ of $N$ is defined in
  the following way. The 0-cubes and the 1-cubes of $X_N$ are the
  vertices and the edges of the marking graph $G_N$.  A 4-cycle
  $(m,m_1,m',m_2)$ of $G_N$ defines a square of $X_N$ iff there exist
  two (necessarily distinct) independent transitions $a,b\in \Sigma$
  such that
  $m\xlongrightarrow{a} m_1, m\xlongrightarrow{b} m_2,
  m_1\xlongrightarrow{b} m'$, and $m_2\xlongrightarrow{a} m'$.
\end{definition}

The cube complex $X_N$ can be transformed into a directed and colored
cube complex $\bX_N=(X_N,o,\lambda_N)$: an edge $mm'$ of $G_N$ is
oriented from $m$ to $m'$ and $\lambda_N(mm')=a$ iff
$m\xlongrightarrow{a} m'$ holds.

\begin{theorem}\label{Petri-to-special}
  $(X_N,o)$ is a finite directed special cube complex with two-sided
  hyperplanes and $\lambda_N$ is a trace labeling of $X_N$ with the
  trace alphabet $(\Sigma,I)$.
\end{theorem}

By Lemma~\ref{cube-completable}, $X_N$ can be completed in a canonical
way to a NPC complex that is also special. In the following, we also
denote this completion by $X_N$.

\begin{figure}
  \includegraphics[scale=0.85]{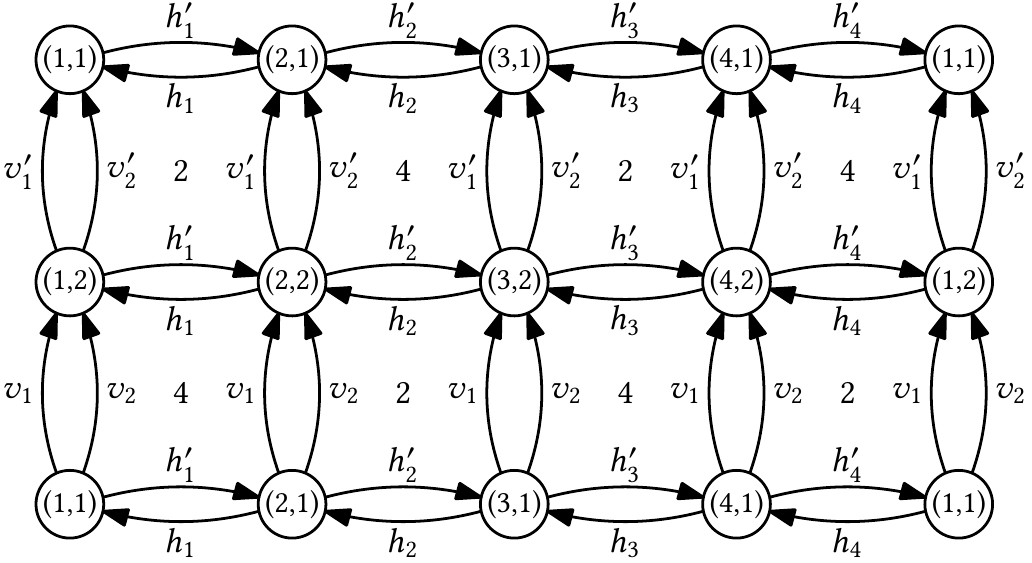}%
  \caption{The special cube complex of the net system $N^*$. A vertex
    labeled $(i,j)$ corresponds to the marking
    $\{H_i,V_j,C_1,C_2,C_3,C_4\}$ of $N$.}%
  \label{fig-petrinet-XN}
\end{figure}

\begin{example}\label{ex-petrinet-XN}
  The special cube complex $X_{N^*}$ of the net system $N^*$ from
  Example~\ref{ex-petrinet} is representend in
  Fig.~\ref{fig-petrinet-XN}. In the figure, the leftmost vertices
  (respectively, edges) should be identified with the rightmost
  vertices (respectively, edges) that have the same label.
  Similarly, the lower vertices and edges should be identified with
  the uppper vertices and edges.
  The complex $X_{N^*}$ has $8$ vertices, $32$ edges, and $24$
  squares. A 4-cycle in the figure is a square of $X_{N^*}$ if
  opposite edges have the same label and if the labels appearing on
  the edges of the square correspond to independent transitions of
  $N^*$. For example, on the right bottom corner, the directed
  $4$-cycle labeled by $h_4$ and $v_1$ is not a square of $X_{N^*}$
  because the transitions $h_4$ and $v_1$ are not independent (as
  explained in Example~\ref{ex-petrinet}).
  The number ($2$ or $4$) in the middle of each $4$-cycle represent
  the number of squares of $X_{N^*}$ on the vertices of this
  $4$-cycle.
\end{example}

Let $\tX_N$ denotes the universal cover of the special cube complex
$X_N$ and let $\varphi: \tX_N \to X_N$ denotes a covering map.  Let
$\tbX_N=(\tX_N,\tildo,\tlambda_N)$ be the directed colored CAT(0) cube
complex, in which the orientation and the coloring are defined as in
Section~\ref{sec-dspec}. For any lift $\tm_0$ of $m_{0}$, denote by
${\cE}_{X_N}=(E',\le',\#',\tlambda_N)$ the $\Sigma$-labeled event
structure whose domain is the principal filter
$\cF_{\tildo}(\tm_0,\tX_N^{(1)})$. Finally, let
${\cE}_N=(E,\le,\#,\lambda)$ be the event structure unfolding of $N$
as defined in Subsection~\ref{unfold} and denote by $\cD(\cE_N)$ the
domain of ${\cE}_N$. The main result of this section is the following
theorem:

\begin{theorem}\label{Petri-to-special-isomorphism}
  The event structures ${\cE}_N=(E,\le,\#,\lambda)$ and
  ${\cE}_{X_N}=(E',\le',\#',\tlambda_N)$ are isomorphic.
\end{theorem}

Using Thiagarajan's characterization of trace-regular event structures
(Theorem~\ref{th:regular_trace}), we establish the converse of
Theorem~\ref{Petri-to-special}.

\begin{proposition}\label{Petri-to-special-isomoprhism2}
  For any finite (virtually) special cube complex $X$, any admissible
  orientation $o$ of $X$, and any vertex $\tv$ in the universal cover
  $\tX$ of $X$, there exists a finite net system $N$ such that the
  domain of the event structure ${\cE}_{N}$ is isomorphic to the
  principal filter $(\cF_{\tildo}(\tv_0,\tX^{(1)}),\leq_{\tildo})$.
\end{proposition}

By Theorem~\ref{Petri-to-special} and
Proposition~\ref{Petri-to-special-isomoprhism2}, we obtain a
correspondence between trace-regular event structures and special cube
complexes, leading to the following corollary:

\begin{corollary}\label{cor-tracereg-implies-stronglyreg}
  Any trace-regular event structure is cover-special, and thus
  strongly regular.
\end{corollary}

\begin{remark}
  In~\cite{CC-thiag}, the following question was formulated: Is it
  true that any regular event structure is strongly regular?
  In view of Corollary~\ref{cor-tracereg-implies-stronglyreg}, if the
  answer to this question is negative, this would provide
  automatically other counterexamples to Thiagarajan's
  Conjecture~\ref{conjecture_regular}, as the counterexamples provided
  in~\cite{CC-thiag} are strongly regular event structures that are
  not trace-regular.
\end{remark}

\subsection{Proof of Theorem~\ref{Petri-to-special}}

We say that a square $Q$ of $X_N$ is an $\{a,b\}$-\emph{square} if two
opposite edges of $Q$ are labeled $a$ and two other opposite edges are
labeled $b$. By the definition of squares of $X_N$, if $Q$ is an
$\{a,b\}$-{square}, then necessarily $(a,b) \in I$.
Each square $Q$ of $X_N$ has a unique source (a vertex $m$ of $Q$
whose two incident edges are directed from $m$) and a unique sink (a
vertex $m'$ of $Q$ whose two incident edges are directed to $m'$). We
restate the definition of squares of $X_N$ in the following way:

\begin{claim}\label{squares}
  A vertex $m$ of $X_N$ is the source of an $\{a,b\}$-square
  $Q=(m,m_1,m',m_2)$ of $X_N$ iff $\bua\cup \bub\subseteq m$,
  $((\abu-\bua)\cup (\bbu-\bub))\cap m=\varnothing$, and
  $(\bua\cup \abu)\cap (\bub\cup \bbu)=\varnothing$. In this case,
  $m_1=m-\bua+\abu, m_2=m-\bub+\bbu, m'=m_1-\bub+\bbu=m_2-\bua+\abu$
  and
  $m\xlongrightarrow{a} m_1, m\xlongrightarrow{b} m_2,
  m_1\xlongrightarrow{b} m', m_2\xlongrightarrow{a} m'$.
\end{claim}

From Claim~\ref{squares}, for any square $Q=(m,m_1,m',m_2)$ of $X_N$,
$\lambda_N(mm_1)=\lambda_N(m_2m')$ and
$\lambda_N(mm_2)=\lambda_N(m_1m')$. Therefore the coloring $\lambda_N$
is a labeling of the directed complex $(X_N,o)$.

A transition $a\in \Sigma$ is called \emph{degenerated} if $\bua=\abu$.

\begin{claim}\label{degenerate-transition}
  $a\in \Sigma$ is degenerated iff for any arc
  $m\xlongrightarrow{a} m'$ of $\oG_N$, we have $m=m'$,
  i.e., $m\xlongrightarrow{a} m'$ is a loop. Moreover, if $a$ is
  degenerated, then $m\xlongrightarrow{a} m$ iff $\bua\subseteq m$.
\end{claim}

\begin{proof}
  If $a$ is degenerated and $m\xlongrightarrow{a} m'$, since
  $\abu=\bua \subseteq m$, we have $m'=(m-\bua)\cup
  \abu=m$. Conversely, if $m\xlongrightarrow{a} m'$ is a loop, since
  $\bua\subseteq m, (\abu-\bua)\cap m=\varnothing$, and
  $m=(m-\bua)\cup \abu$, we have $\bua=\abu$. The second assertion
  follows from the definition of $m\xlongrightarrow{a} m'$.
\end{proof}

Notice that an $\{ a,b\}$-square $(m,m_1,m',m_2)$ of $X_N$ is either
non-degenerated (its four vertices are pairwise distinct) or one of
the transitions $a$ or $b$ is degenerated and another one not (if $a$
is degenerated, then $m=m_1$ and $m_2=m')$ or both transitions $a$ and
$b$ are degenerated (in this case, $m=m_1=m'=m_2$ and $\bua=\abu$,
$\bub=\bbu$, and $\bua \cap \bub = \varnothing$).

For $a\in \Sigma$, a hyperplane $H$ of $X_N$ is called an
$a$-\emph{hyperplane} if all edges $mm'$ dual to $H$ are labeled
$a$. An $a$-hyperplane $H$ such that $a$ is a degenerated transition
is called a \emph{degenerated hyperplane}. By
Claim~\ref{degenerate-transition}, a degenerated $a$-hyperplane can be
viewed as a connected component of the subgraph of the marking graph
$G_N$ induced by all markings $m$ having a loop labeled $a$.

\begin{claim}\label{XN2-sided}
  The hyperplanes of $X_N$ are two-sided.
\end{claim}

\begin{proof}
  Any degenerated hyperplane can always be viewed as a two-sided
  hyperplane. Now suppose that $H$ is a non-degenerated
  $a$-hyperplane. If $H$ is not two-sided, then there exists an edge
  $mm'$ dual to $H$ such that $m\xlongrightarrow{a} m'$ and
  $m'\xlongrightarrow{a} m$. Consequently,
  $\bua\subseteq m, (\abu-\bua)\cap m=\varnothing, m'=(m-\bua)\cup
  \abu$ and
  $\bua\subseteq m', (\abu-\bua)\cap m'=\varnothing, m=(m'-\bua)\cup
  \abu$.  Since $\abu\subseteq m$ and $(\abu-\bua)\cap m=\varnothing$,
  we have $\abu\subseteq\bua$. If there exists $e\in \bua-\abu$, then
  $e\notin m'=(m-\bua)\cup \abu$, contradicting that
  $\bua\subseteq m'$. As a result, we have $\bua=\abu$, i.e., $a$
  is a degenerated transition, contradicting the choice of $H$.
\end{proof}

\begin{claim}\label{XNlambda}
  $\lambda_N$ is a trace labeling of $(X_N,o)$.
\end{claim}

\begin{proof}
  From the definition of the squares of $X_N$ it follows that
  $\lambda_N$ satisfies (TL1).

  To prove (TL2), let $m$ be a vertex of $X_N$ with two outgoing edges
  $\ovr{mm_1},\ovr{mm_2}$. Since $\oG_N$ is deterministic,
  $\lambda_N(mm_1)\ne \lambda_N(mm_2)$, say $\lambda_N(mm_1)=a$ and
  $\lambda_N(mm_2)=b$.  We assert that $(a,b)\in I$ iff $\ovr{mm_1}$
  and $\ovr{mm_2}$ belong to a common square of $X_N$. If $\ovr{mm_1}$
  and $\ovr{mm_2}$ belong to a square of $X_N$, by the definition of
  the squares of $X_N$, we have
  $(a,b)=(\lambda_N(mm_1),\lambda_N(mm_2))\in I$. Conversely, suppose
  that $(a,b)\in I$, i.e.,
  $(\bua\cup \abu)\cap (\bub\cup \bbu)=\varnothing$. Since
  $m\xlongrightarrow{a} m_1$ and $m\xlongrightarrow{b} m_2$, we also
  have $\bua\cup \bub\subseteq m$ and
  $((\abu-\bua)\cup (\bbu-\bub))\cap m=\varnothing$ hold. From
  Claim~\ref{squares} we deduce that $m$ is the source of an
  $\{a,b\}$-square in which the two neighbors of $m$ are $m_1$ and
  $m_2$. This establishes (TL2).

  To prove (TL3), let $m\xlongrightarrow{a} m_1$ and
  $m_1\xlongrightarrow{b} m'$ be two distinct arcs. Then we have to
  prove that $(a,b)\in I$ iff $\ovr{mm_1}$ and $\ovr{m_1m'}$ belong to
  a square of $X_N$.  If $\ovr{mm_1}$ and $\ovr{m_1m'}$ belong to a
  square of $X_N$, then $\ovr{mm_1}$ and $\ovr{m_1m'}$ are consecutive
  edges of that square, and by the definition of the squares of $X_N$,
  this implies that $(a,b)\in I$.  Conversely, let $(a,b)\in I$, i.e.,
  $(\bua\cup \abu)\cap (\bub\cup \bbu)=\varnothing$. Set
  $m_2=m-\bub+\bbu$. Suppose by way of contradiction that
  $(m,m_1,m',m_2)$ is not a square of $X_N$. By Claim~\ref{squares}
  and since $(\bua\cup \abu)\cap (\bub\cup \bbu)=\varnothing$, either
  $\bua\cup \bub \nsubseteq m$ or
  $((\abu-\bua)\cup (\bbu-\bub))\cap m\ne \varnothing$. Since
  $\bua\subseteq m$ and $(\abu-\bua)\cap m=\varnothing$ because
  $m\xlongrightarrow{a} m_1$, either $\bub \nsubseteq m$ or
  $(\bbu-\bub)\cap m\ne \varnothing$. If $\bub \nsubseteq m$, since
  $(\bua\cup \abu)\cap \bub=\varnothing$ and $m_1=m-\bua+\abu$, we
  deduce that $\bub \nsubseteq m_1$, contradicting that
  $m_1\xlongrightarrow{b} m'$. On the other hand, if there exists
  $e\in (\bbu-\bub)\cap m$, since $(\bbu-\bub)\cap m_1=\varnothing$
  and $m_1=m-\bua+\abu$, we have $e\in \bua-\abu$. Hence
  $e\in \bua\cap \bbu$, contradicting that $(a,b) \in I$. This proves
  that $m_2$ is an admissible marking and that $(m,m_1,m',m_2)$ is a
  square of $X_N$, establishing (TL3).

  Finally, we  establish (TL4). Let $m'$ be a vertex of $X_N$ and
  let $m_1\xlongrightarrow{a} m'$ and $m_2\xlongrightarrow{b} m'$ be
  two distinct arcs. Since $\oG_N$ is codeterministic,
  $a \neq b$.  To prove (TL4) we have to show that $(a,b)\in I$ iff
  $\ovr{m_1m'}$ and $\ovr{m_2m'}$ belong to a common square of
  $X_N$. Again, one direction directly follows from the definition of
  the squares of $X_N$. Conversely, suppose that $(a,b)\in I$, i.e.,
  $(\bua\cup \abu)\cap (\bub\cup \bbu)=\varnothing$. Hence
  $(\abu-\bua)\cap (\bbu-\bub)=\varnothing$. Set
  $m:=m'-(\abu \cup \bbu)+(\bua\cup \bub)$.  We assert that
  $(m,m_1,m',m_2)$ is a square of $X_N$ with source $m$. Since the
  transitions $a$ and $b$ are independent, by Claim~\ref{squares} it
  suffices to show that $\bua\cup \bub\subseteq m$ and
  $((\abu-\bua)\cup (\bbu-\bub))\cap m=\varnothing$. Both these
  properties directly follow from the definition of $m$. This
  establishes (TL4).
\end{proof}

Theorem~\ref{Petri-to-special} now follows from Claims~\ref{XN2-sided}
and~\ref{XNlambda} and Theorem~\ref{special-trace-bis}.

\subsection{Proof of Theorem~\ref{Petri-to-special-isomorphism} and Proposition~\ref{Petri-to-special-isomoprhism2}}

Let $N=(S,\Sigma,F,m_{0})$ be a net system. As above, $FS$ denotes the
set of all firing sequences at $m_0$, i.e., all words
$\sigma\in \Sigma^*$ for which there exists a marking $m$ such that
$m_0\xlongrightarrow{\sigma}m$. The trace alphabet associated to $N$
is the pair $M=(\Sigma,I)$ where $(a,b)\in I$ iff
$(\abu\cup \bua)\cap (\bbu \cup \bub)=\varnothing$. Recall that
$\FT(N)$ is the set of \emph{firing traces} of $N$, i.e., the set of
traces of the form $\trsigma$ for $\sigma\in FS$.  Let
${\cE}_N=(E,\le,\#,\lambda)$ be the $M$-labeled event structure
unfolding of a net system $N$. Recall that the set events of ${\cE}_N$
is the set $\PFT(N)$ of prime firing traces of $N$ and the label of an
event $\trsigma$ is $\lambda(\trsigma)=\last(\sigma)$.

Let also ${\cE}_{X_N}=(E',\le',\#',\tlambda_N)$ be the
$\Sigma$-labeled event structure whose domain is the principal filter
$\cF_{\tildo}(\tm_0,\tX_N^{(1)})$ of the universal cover
$\tbX_N=(\tX_N,\tildo,\tlambda_N)$ of the special cube complex
$(X_N,o,\lambda_N)$ of $N$. Let $\varphi: \tX_N\mapsto X_N$ denote a
covering map.  Let $G({\cE}_{X_N})$ denotes the median graph of
${\cE}_{X_N}$. From Theorem~\ref{Petri-to-special} and
Proposition~\ref{special-trace1} it follows that $\tlambda_N$ is a
trace labeling of the event structure ${\cE}_{X_N}$. By
Corollary~\ref{geodesictrace} there exists a bijection between the set
of configurations of $\cE_{X_N}$ and the set $\GT(\cE_{X_N})$ of
geodesic traces of $\cE_{X_N}$. By Corollary~\ref{geodesicprimetrace}
there exists a bijection $H\mapsto \trace{\sigma_v}$ between the set
of hyperplanes (events) of $\cE_{X_N}$ and the set $\PGT(\cE_{X_N})$
of prime geodesic traces of $\cE_{X_N}$.

The next claim establishes a bijection between geodesic traces and
firing traces.

\begin{claim}\label{faringtrace}
  Any geodesic trace $\trace{\sigma_{\tm}}$ of ${\cE}_{X_N}$ is a
  firing trace of $N$. Conversely, for any firing trace $\trsigma$
  there exists a geodesic trace $\trace{\sigma_{\tm}}$ such that
  $\trsigma=\trace{\sigma_{\tm}}$. In particular, there is a bijection
  between prime geodesic traces of $\cE_{X_N}$ and the prime firing
  traces of $N$.
\end{claim}

\begin{proof}
  Each firing sequence $\sigma$ of $N$ corresponds to a directed path
  in the directed marking graph $\oG_N$: if
  $\sigma=a_1\ldots a_k\in FS$ is a firing sequence, then there exists
  reachable markings $m_1,\ldots,m_{k+1}$ such that
  $\pi(\sigma):=m_0\xlongrightarrow{a_1}m_1\xlongrightarrow{a_2}\ldots
  \xlongrightarrow{a_{k-1}}m_k\xlongrightarrow{a_k}m_{k+1}$ is a
  directed $(m_0,m_{k+1})$-path of $\oG_N$. Since $G_N$ is the
  1-skeleton of the special cube complex $X_N$, the directed universal
  cover $\tbX_N$ of $X_N$ contains a directed path $\tpi(\sigma)$ from
  $\tm_0$ to $\tm_{k+1}$ whose image under the covering map $\varphi$
  is $\pi(\sigma)$. By Lemma~\ref{directedpath}, $\tpi(\sigma)$ is a
  shortest $(\tm_0,\tm_{k+1})$-path in the 1-skeleton of $\tbX_N$. Let
  $\sigma'$ be another firing sequence such that there exists
  $(a_{i},a_{i+1})\in I$ and
  $\sigma'=a_1\ldots a_{i-1}a_{i+1}a_{i}a_{i+2}\ldots a_k$. Since the
  transitions $a_i$ and $a_{i+1}$ are independent, there exists
  $m'_{i+1}$ such that
  $m_i\xlongrightarrow{a_{i+1}}m'_{i+1}\xlongrightarrow{a_i}m_{i+2}$
  and $(m_i,m_{i+1},m_{i+2},m'_{i+1})$ is a square of $X_N$.  Then
  $\sigma'$ corresponds to the directed path
  $\pi(\sigma')=m_0\xlongrightarrow{a_1} \ldots
  \xlongrightarrow{a_{i-1}}m_i\xlongrightarrow{a_{i+1}}m'_{i+1}
  \xlongrightarrow{a_i}m_{i+2}\xlongrightarrow{a_{i+2}}\ldots
  \xlongrightarrow{a_k}m_{k+1}$ of $\oG_N$. Analogously to
  $\pi(\sigma)$, $\tbX_N$ also contains a directed path
  $\tpi(\sigma')$ with origin $\tm_0$ whose $\varphi$-image is
  $\pi(\sigma')$. Moreover, $(\tm_i,\tm_{i+1},\tm_{i+2},\tm'_{i+1})$
  is a square of $\tbX_N$, thus $\tpi(\sigma')$ can be obtained from
  $\tpi(\sigma)$ by an elementary homotopy with respect to this
  square. Consequently, the firing trace $\trsigma$ of $N$ is
  contained in the geodesic trace $\trace{\sigma_{\tm_{k+1}}}$ of
  ${\cE}_{X_N}$.

  It remains to prove the converse inclusion, i.e., that any geodesic
  trace $\trace{\sigma_{\tm}}$ of ${\cE}_{X_N}$ is contained in a
  firing trace. Let $\tpi$ be a shortest $(\tm_0,\tm)$-path in the
  graph $G({\cE}_{X_N})$. Then the edges of $\tpi$ are directed from
  $\tm_0$ to $\tm$. The image of $\tpi$ by the covering map $\varphi$
  is a directed $(m_0,m)$-path $\pi$ in the graph $\oG_N$. This path
  is not necessarily shortest or simple, however the words defined by
  the labels of edges of $\tpi$ and $\pi$ coincide:
  $\sigma(\tpi)=\sigma(\pi)$. Since $\pi$ is a directed $(m_0,m)$-path
  in the marking graph, necessarily $\sigma(\pi)$ is a firing
  sequence, yielding that $\sigma(\tpi)\in FS$. By Lemma~\ref{trace2},
  the geodesic trace $\trace{\sigma(\tpi)}$ consists exactly of all
  $\sigma(\tpi')$ such that $\tpi'$ is a shortest
  $(\tm_0,\tm)$-path. By Lemma~\ref{homotopic}, the paths $\tpi$ and
  $\tpi'$ are homotopic, i.e., there exists a finite sequence
  $\tpi=:\pi_1,\tpi_2,\ldots,\tpi_{k-1},\tpi_k:=\tpi'$ of shortest
  $(\tm_0,\tm)$-paths such that for any $i=1,\ldots,k-1$ the paths
  $\tpi_i$ and $\tpi_{i+1}$ differ only in a square
  $\widetilde{Q}_i=(\tm_{j-1},\tm_{j},\tm_{j+1},\tm'_{j})$ of
  $\tbX_N$. Let $\pi_i$ denote the image of the path $\tpi_i$ under
  the covering map $\varphi$. Let also
  $Q_i:=(m_{j-1},m_{j},m_{j+1},m'_{j})=\varphi(\widetilde{Q}_i)$. Each
  $\pi_i$ is a directed $(m_0,m)$-path of $\oG_N$ and each $Q_i$ is a
  square of $X_N$. Moreover, $\sigma(\tpi_i)=\sigma(\pi_i)$ for each
  $i$ and the edges of the squares $\widetilde{Q}_i$ and $Q_i$ are
  labeled in the same way. Each $\pi_{i+1}$ is obtained from $\pi_i$
  by an elementary homotopy with respect to the square $Q_i$. From the
  definition of the squares of $X_N$ it follows that there exists
  $(a_j,a_{j+1})\in I$ such that
  $m_{j-1} \xlongrightarrow{a_{j}} m_j \xlongrightarrow{a_{j+1}}
  m_{j+1}$ and
  $m_{j-1} \xlongrightarrow{a_{j+1}} m'_j \xlongrightarrow{a_{j}}
  m_{j+1}$. But this implies that $\sigma(\pi_{i+1})$ is obtained from
  $\sigma(\pi_i)$ by exchanging $a_{j}$ with $a_{j+1}$, yielding that
  $\sigma(\pi_{i+1})$ belongs to the trace of $\sigma(\pi_{i})$. Since
  all $\sigma(\pi_i)$ are firing sequences of $FS$, they all belong to
  the firing trace of $\sigma(\pi)$. Since $\sigma(\tpi)=\sigma(\pi)$,
  we conclude that the geodesic trace of $\sigma(\tpi)$ is included in
  the firing trace of $\sigma(\pi)$. This concludes the proof of the
  equality between geodesic traces and firing traces.

  Observe that if $\trace{\sigma_{\tm}}$ is a prime geodesic trace,
  then the interval $I(\tm_0,\tm)$ is prime by
  Lemma~\ref{geoprimetrace1}. Since each
  $\sigma'_{\tm} \in \trace{\sigma_{\tm}}$ corresponds to a shortest
  $(\tm_0,\tm_{k})$-path, all such paths share the same last
  edge. Consequently, all the words in $\trace{\sigma_{\tm}}$ have the
  same last letter, and thus the corresponding firing trace is prime.
  Conversely, for any prime firing trace $\trsigma$, let $\tm$ be the
  vertex of $\tbX_N$ such that $\trsigma =
  \trace{\sigma_{\tm}}$. Since $\trsigma$ is prime, all words in
  $\trsigma = \trace{\sigma_{\tm}}$ have the same last letter. Since
  two incoming arcs of $\tm$ have different labels, this implies that
  $\tm$ has only one incoming arc, i.e., the interval $I(\tm_0,\tm)$
  is prime.  By Lemma~\ref{geoprimetrace1}, the geodesic trace
  $\trace{\sigma_{\tm}}$ is prime.
\end{proof}

Claim~\ref{faringtrace} establishes a bijection between prime geodesic
traces of $\cE_{X_N}$ and prime firing traces of $N$. Consequently,
there is a bijection between the hyperplanes (events) of ${\cE}_{X_N}$
and the hyperplanes (events) of $\cE_N$ and this bijection preserves
labels.  Therefore, to establish that the event structures ${\cE}_N$
and ${\cE}_{X_N}$ are isomorphic it remains to show that this
bijection preserves the precedence and the conflict relations.  By
Lemma~\ref{geoprimetrace3}, for two hyperplanes $H',H$ of
${\cE}_{X_N}$ with prime geodesic traces $\trace{\sigma_{\tu}}$ and
$\trace{\sigma_{\tv}}$, respectively, we have $H'\le H$ iff
$\trace{\sigma_{\tu}}\sqsubseteq\trace{\sigma_{\tv}}$. On the other
hand, for ${\cE}_N$ the precedence relation $\le$ is the prefix
relation $\sqsubseteq$. Therefore the precedence relation is
preserved.

Finally, we show that the conflict relation is also preserved.  Taking
into account the bijection between firing traces and geodesic traces
from Claim~\ref{faringtrace}, the definition of the conflict relation
in $\cE_N$ can be rephrased in the following way: two prime firing
traces $\trace{\sigma_{\tu}}$ and $\trace{\sigma_{\tv}}$ are in
conflict iff there does not exist a firing trace
$\trace{\sigma_{\tw}}$ such that $\trace{\sigma_{\tu}}$ and
$\trace{\sigma_{\tv}}$ are prefixes of $\trace{\sigma_{\tw}}$. By
Lemma~\ref{trace3},
$\trace{\sigma_{\tu}}\sqsubseteq \trace{\sigma_{\tw}}$ iff
$\tu\in I(\tm_0,\tw)$.  Consequently, two prime firing traces
$\trace{\sigma_{\tu}}$ and $\trace{\sigma_{\tv}}$ are in conflict iff
there does not exist a vertex $\tw$ such that
$\tu,\tv\in I(\tm_0,\tw)$.  By Lemma~\ref{trace4}, there does not
exists a vertex $\tw$ such that $\tu,\tv\in I(\tm_0,\tw)$ iff the
prime geodesic traces $\trace{\sigma_{\tu}}$ and
$\trace{\sigma_{\tv}}$ are in conflict in ${\cE}_{X_N}$.  This proves
that the conflict relation is preserved and finishes the proof of
Theorem~\ref{Petri-to-special-isomorphism}.

\medskip

We conclude with the proof of
Proposition~\ref{Petri-to-special-isomoprhism2}.  Consider a finite
virtually special cube complex $X$ with an admissible orientation $o'$
and let $Y$ be a finite special cover of $X$ and $o$ be the
orientation of $Y$ lifted from $o'$. Since the directed universal
cover $(\tX,\tildo')$ of $(X,o')$ is isomorphic to the directed
universal cover $(\tY,\tildo)$ of $(Y,o)$, it is enough to prove
Proposition~\ref{Petri-to-special-isomoprhism2} for the finite
directed special complex $(Y,o)$.  By Theorem~\ref{special-trace-bis},
there exists a trace labeling $\lambda$ of $(Y,o)$. By
Proposition~\ref{special-trace1}, for any vertex $\tv_0$ of $\tY$, the
lift $\tlambda$ of $\lambda$ is a trace-regular labeling of the
principal filter
$(\cF_{\tildo}(\tv,\tY^{(1)}),\prec_{\tildo})$. Hence,
$(\cF_{\tildo}(\tv,\tY^{(1)}),\prec_{\tildo})$ is the domain of a
trace-regular event structure $\cE$. By Thiagarajan's
Theorem~\ref{th:regular_trace}, there exists a net system $N$ such
that $\cE$ is isomorphic to $\cE_N$. This ends the proof of
Proposition~\ref{Petri-to-special-isomoprhism2}.

\section{The MSO Theory of Net Systems and of their
  Domains}\label{MSO-domain}

\subsection{The Results}

Let ${\cE}=(E,\le, \#,\lambda)$ be a trace-regular event structure and
let $\cD(\cE)$ denote the domain of $\cE$. Let $G(\cE)$ denote the
undirected covering median graph of $\cD(\cE)$ and
$\oG(\cE)=(G(\cE),o)$ denote the directed graph of $\cD(\cE)$. First
we characterize the trace event structures for which the MSO theories
of graphs $G(\cE)$ and $\oG(\cE)$ are decidable. In the following
theorem, we consider two MSO theories of graphs; informally, $\MSO_1$
allows only to quantify on vertices of the graph while $\MSO_2$ allows
to quantify on vertices and edges (see~Subsection~\ref{ssec:MSO} for
precise definitions).

\begin{theorem}\label{mso-graph}
  For a trace-regular event structure ${\cE}=(E,\le, \#,\lambda)$,
  conditions (1)-(6) are equivalent:
  \begin{enumerate}
  \item $\MSO(\oG(\cE))$ is decidable;
  \item $\MSO_1(G(\cE))$ is decidable;
  \item $\MSO_2(G(\cE))$ is decidable:
  \item $G(\cE)$ has finite treewidth;
  \item the clusters of $G(\cE)$ have bounded diameter;
  \item $\oG(\cE)$ is context-free.
  \end{enumerate}
\end{theorem}

If instead of the domain of an event structure unfolding of a net
system $N$, we consider the 1-skeleton of the universal cover of the
special cube complex $X_N$, we obtain the following result:

\begin{proposition}\label{mso-graph-complex}
  Let $N=(S,\Sigma,F,m_{0})$ be a net system, $X_N$ be the special
  cube complex of $N$, and $\oG(\tbX_N)$ be the 1-skeleton of the
  directed labeled universal cover of $X_N$.  Then conditions (1)-(4)
  are equivalent:
  \begin{enumerate}
  \item $\MSO(\oG(\tbX_N))$ is decidable;
  \item $\MSO_2(G(\tbX_N))$ is decidable:
  \item $G(\tbX_N)$ has finite treewidth;
  \item $\oG(\tbX_N)$ is context-free.
  \end{enumerate}
\end{proposition}

The proof of this result essentially follows from the result by Kuske
and Lohrey~\cite{KuLo} (see Theorem~\ref{th-kuskelohrey} below) that
the decidability of the MSO theory of a directed graph $\oG$ of
bounded degree whose automorphism group $\Aut(\oG)$ has only finitely
many orbits on $\oG$ is equivalent to the fact that $\oG$ is
context-free and to the fact that its undirected support has finite
treewidth. This result cannot be applied to prove
Theorem~\ref{mso-graph} because $\Aut(\oG(\cE))$ may have an infinite
number of orbits (however this is true for $\oG(\tbX_N)$).

To relate the MSO theory of a trace-regular event structure graph with
the MSO theory of its domain, we introduce the notion of the
\emph{hairing} $\dotcE = (\dotE,\dotleq,\dotdiese)$ of an event structure
$\cE =(E,\leq,\#)$. To obtain $\dotcE$, we add a \emph{hair event} $e_c$
for each configuration $c$ of $\cE$, i.e., $\dotE = E \cup E_C$ where
$E_C = \{e_c: c \in \cD(\cE)\}$. For any hair event $e_c$ and any
event $e \in \dotE$, we set $e \dotleq e_c$ if $e \in c$ and
$e \dotdiese e_c$ otherwise.
Suppose additionally that $\cE$ is trace-regular and let $\lambda$ be
a trace labeling of $\cE$ with a trace alphabet $M=(\Sigma,I)$.  Let
$h$ be a letter that does not belong to $\Sigma$ and consider the
trace alphabet $\dotM = (\Sigma \cup \{h\},I)$ (note that since $I$ is
not modified, $(h,a) \notin I$ for every $a \in \Sigma$). Let
$\dotlambda$ be the labeling of $\dotcE$ extending $\lambda$ by
setting $\dotlambda(e_c) = h$ for any $e_c \in E_C$.  The labeled
event structure obtained in this way is trace-regular:

\begin{proposition}\label{barycentric-subdivision}
  For a trace-regular event structure ${\cE}=(E,\le, \#,\lambda)$, the
  hairing $\dotcE =(\dotE,\dotleq, \dotdiese,\dotlambda)$ is also a
  trace-regular event structure.
\end{proposition}

By the definition of $\dotcE$, the directed graph $\oG(\dotcE)$ of its
domain $\cD(\dotcE)$ is obtained from the directed graph $\oG(\cE)$ of
the domain $\cD(\cE)$ of $\cE$ by adding an outgoing arc $\ovr{vw_v}$
to each vertex $v$ of $\oG(\cE)$.  In a similar way, we can define the
\emph{hairing} $\dotG$ (respectively $\dotX$) of any directed graph
$G$ (respectively, any directed NPC complex $X$) by adding for each
vertex $v$, a new vertex $v'$ and an arc $\ovr{vv'}$. Observe that
each new vertex $v'$ has in-degree $1$ and out-degree $0$. With this
definition, the hairing $\dotoG(\cE)$ of the directed graph
$\oG(\dotcE)$ of the domain $\cD(\dotcE)$ of $\cE$ coincides with the
directed graph $\oG(\dotcE)$ of the domain $\cD(\dotcE)$ of the
hairing $\dotcE$ of $\cE$.
Given a poset $\cD=(D,\prec)$, we define the \emph{hairing} of $\cD$
as the poset $\dotcD=(\dotD,\dotprec)$ such that the Hasse diagram of
$\dotcD$ is the hairing of the Hasse diagram of $\cD$. With this
definition, the domain $\cD(\dotcE)$ of $\cE$ coincides with the
domain $\cD(\dotcE)$ of the hairing $\dotcE$ of $\cE$.

Note that the hairing $\dottX$ of the universal cover $\tX$ of a
directed NPC complex $X$ coincides with the universal cover $\tdotX$
of the hairing $\dotX$ of $X$.  When an event structure $\cE$ is
strongly regular, there exists a finite directed NPC complex $X$ such
that $\cD(\cE)=\cF_{\tildo}(\tv,\tX^{(1)})$. In this case, we have:
\begin{equation}\label{eq-DcE}
  \cD(\dotcE) = \dotcD(\cE) = \dotcF_{\tildo}(\tv,\tX^{(1)}) =
  \cF_{\tildo}(\tv,\dottX{^{(1)}}) = \cF_{\tildo}(\tv,\tdotX{^{(1)}}).
\end{equation}

We can also define the \emph{hairing} $\dotN = (\dotS,\dotSigma,\dotF,\dotm_0)$
of a net system $N = (S,\Sigma,F,m_0)$ as follows. First, for each
transition $a \in \Sigma$, we add a place $p_a$ such that
$\bup_a = \pbu_a = \{a\}$ and such that $p_a$ contains a token in the
initial configuration. Then, we add a transition $h$ such that
$\buh = \{p_a: a \in \Sigma\}$ and $\hbu = \varnothing$. In other
words, $\dotS = S \cup \{p_a: a \in \Sigma\}$,
$\dotSigma = \Sigma \cup \{h\}$,
$\dotF = F \cup \{(p_a,a), (a,p_a), (p_a, h):a \in \Sigma\}$, and
$\dotm_0 = m_0 \cup \{p_a: a \in \Sigma\}$.

\begin{proposition}\label{prop-hairing-N}
  For a net system $N$, there is an isomorphism between the special
  cube complex $X_{\dotN}$ of the hairing $\dotN$ of $N$ and the hairing
  $\dotX_N$ of the special cube complex $X_N$ of $N$ that maps the
  initial marking $m_0$ of $N$ to the initial marking $\dotm_0$ of
  $\dotN$.
  Consequently, the event structure unfolding $\cE_{\dotN}$ of the
  hairing $\dotN$ of $N$ is isomorphic to the hairing $\dotcE_N$ of the
  event structure unfolding $\cE_N$ of $N$.
\end{proposition}

Note that Proposition~\ref{barycentric-subdivision} follows from
Proposition~\ref{prop-hairing-N} and Thiagarajan's
Theorem~\ref{th:regular_trace}. However, we provide a simple proof of
this result that does not rely on the involved construction of a net
system from a trace-regular event structure used in the proof of
Theorem~\ref{th:regular_trace}.

Notice that the hair events of $\dotcE$ introduce a lot of conflicting
events in $\dotcE$, and we use them to encode vertex variables as event
variables in order to prove the following result:

\begin{theorem}\label{barycentric-subdivision-MSO}
  For a trace-regular event structure ${\cE}=(E,\le, \#,\lambda)$,
  $\MSO(\dotcE)$ is decidable if and only if $\MSO(\oG(\cE))$ is
  decidable.  In particular, $\MSO(\dotcE)$ is decidable if and only if
  $G(\cE)$ has finite treewidth.
\end{theorem}

\begin{remark}
  The condition on the treewidth of $G(\cE)$ in
  Theorem~\ref{barycentric-subdivision-MSO} is independent of the
  choice of a particular trace labeling of $\cE$. Therefore, for a
  trace-regular event structure ${\cE}=(E,\le, \#)$ such that $G(\cE)$
  has bounded (respectively, unbounded) treewidth, $\MSO(\dotcE)$ is
  decidable (respectively, undecidable) for any trace-regular labeling
  $\dotlambda$ of $\dotcE$.  Consequently, if there exists a
  trace-regular labeling of $\dotcE$ such that $\MSO(\dotcE)$ is
  decidable (respectively, undecidable), then $\MSO(\dotcE)$ is
  decidable (respectively, undecidable) for all trace-regular
  labelings of $\dotcE$.
\end{remark}

Since $\MSO(\cE)$ is a fragment of $\MSO(\dotcE)$, we obtain the
following corollary of Theorem~\ref{barycentric-subdivision-MSO}:
\begin{corollary}
  For a trace-regular event structure $\cE =(E,\le, \#,\lambda)$, if
  $G(\cE)$ has finite treewidth, then $\MSO(\cE)$ is decidable.
\end{corollary}

\subsection{Treewidth}

Let $G=(V,E)$ be a simple graph, not necessarily finite. A \emph{tree
  decomposition}~\cite{RoSey} of $G$ is a pair $(T,f)$, where
$T=(V(T),E(T))$ is a tree and
$f: V(T)\rightarrow 2^V\setminus \{ \varnothing\}$ is a function such
that \emph{(i)} $\bigcup_{t\in V(T)} f(t)=V$, \emph{(ii)} for every
edge $uv\in E$ of $G$ there exists $t\in V(T)$ such that
$u,v\in f(t)$, and \emph{(iii)} if $t',t''\in V(T)$ and $t$ lies on
the unique path of $T$ from $t'$ to $t''$, then
$f(t')\cap f(t'')\subseteq f(t)$.  The \emph{width} of the tree
decomposition $(T,f)$ is $\sup \{|f(t)|-1: t\in V(T)\}$ and belongs to
${\mathbb N}\cup \{ \infty\}$.  The graph $G$ has \emph{treewidth}
$\le b$ if there exists a tree decomposition of $G$ of width $\le
b$. A graph $G$ has \emph{bounded} (or \emph{finite}) \emph{treewidth}
if it has treewidth $\le b$ for some $b\in {\mathbb N}$.  The
treewidth represents how close a graph is to a tree from a
\emph{combinatorial} point of view.

A graph $H$ is a \emph{minor} of a graph $G$ if $H$ can be obtained
from a subgraph $G'$ of $G$ by contracting edges. Equivalently, $H$ is
a minor of a connected graph $G$ if $G$ contains a subgraph $G'$ such
that there exists a partition of vertices of $G'$ into connected
subgraphs $\cP=\{ P_1,\ldots,P_t\}$ and a bijection
$f: V(H)\rightarrow \cP$ such that if $uv\in E(H)$ then there exists
an edge of $G'$ between the subgraphs $f(u)$ and $f(v)$ of $\cP$
(i.e., after contracting each subgraph $P_i\in \cP$ into a single
vertex we obtain a graph containing $H$ as a spanning
subgraph). Treewidth does not increase when taking a minor.

Since the treewidth of an $n\times n$ square grid is $n$, the
treewidth of a graph $G$ is always greater than or equal to the size
of the largest square grid minor of $G$. In the other direction, the
grid minor theorem by Robertson and Seymour~\cite{RoSey1} shows that
there exists a function $f$ such that the treewidth is at most $f(r)$
where $r$ is the size of the largest square grid minor of $G$:

\begin{theorem}[\cite{RoSey1}]\label{th-robseymour}
  A graph $G$ has bounded treewidth if and only if the square grid
  minors of $G$ have bounded size.
\end{theorem}

\subsection{Hyperbolicity}

Similarly to nonpositive curvature, Gromov hyperbolicity is defined in
metric terms. However, as for the CAT(0) property, the hyperbolicity
of a CAT(0) cube complex can be expressed in a purely combinatorial
way.  A metric space $(X,d)$ is
$\delta$-\emph{hyperbolic}~\cite{BrHa,Gromov} if for any four points
$v,w,x,y$ of $X$,
$d(v,w)+d(x,y) \leq \max \{d(v,x)+d(w,y), d(v,y)+d(w,x)\} + 2 \delta$.
A graph $G=(X,E)$ endowed with its standard graph-distance $d_G$ is
$\delta$-\emph{hyperbolic} if the metric space $(X,d_G)$ is
$\delta$-{hyperbolic}.  At the difference of treewidth, Gromov
hyperbolicity represents how close \emph{metrically} a graph is to a
tree.  A metric space $(X,d)$ is \emph{hyperbolic} if there exists
$\delta<\infty$ such that $(X,d)$ is $\delta$-hyperbolic.  In case of
median graphs, i.e., of 1-skeletons of CAT(0) cube complexes, the
hyperbolicity can be characterized in the following way:

\begin{lemma}[\cite{ChDrEsHaVa,Ha}]\label{hyp_median}
  Let $X$ be a CAT(0) cube complex.  Then its $1$-skeleton $X^{(1)}$
  is hyperbolic if and only if all isometrically embedded square grids
  are uniformly bounded.
\end{lemma}

In Hagen's paper~\cite[Theorem 7.6]{Ha}, previous lemma is a
consequence of another combinatorial characterization of the
hyperbolicity of CAT(0) cube complexes of bounded degree. The
\emph{crossing graph} $\Gamma(X)$ of a CAT(0) cube complex $X$ has the
hyperplanes of $X$ as vertices and pairs of intersecting hyperplanes
as edges. We  say that a graph $\Gamma$ has \emph{thin bicliques}
if there exists a natural number $n$ such that any complete bipartite
subgraph $K_{p,q}$ of $\Gamma$ satisfies $p\le n$ or $q\le n$.

\begin{theorem}[\cite{Ha}]\label{hyp_median_hagen}
  A CAT(0) cube complex $X$ with bounded degree is hyperbolic if and
  only if its crossing graph $\Gamma(X)$ has thin bicliques.
\end{theorem}

We call an event structure $\cE=(E,\le, \#)$ and its domain $\cD(\cE)$
\emph{hyperbolic} if $\cD(\cE)$ is isomorphic to a principal filter of
a directed CAT(0) cube complex, whose 1-skeleton is hyperbolic.  We
call an event structure $\cE=(E,\le, \#)$ and its domain $\cD(\cE)$
\emph{strongly hyperbolic-regular} if there exists a finite directed
NPC complex $(X,o)$ such that $\tX$ is hyperbolic and $\cD$ is a
principal filter of $(\tX^{(1)},\tildo)$. Note that an event structure
can be strongly regular and hyperbolic without being strongly
hyperbolic-regular (see Remark~\ref{rem-tore-bazar}).

\subsection{Context-free Graphs}

Let $G$ be an edge-labeled graph of uniformly bounded degree and $v_0$
be an arbitrary root (basepoint) of $G$. Let
$S(v_0,k)=\{ x\in V: d_G(v_0,x)=k\}$ denote the sphere of radius $k$
centered at $v_0$. A connected component $\Upsilon$ of the subgraph of
$G$ induced by $V\setminus S(v_0,k)$ is called an \emph{end} of $G$.
The vertices of $\Upsilon\cap S(v_0,k+1)$ are called \emph{frontier
  points} and this set is denoted by $C(\Upsilon)$~\cite{MullerSchupp}
and called a \emph{cluster}.  There exists a bijection between the ends
and the clusters: each end contains a unique cluster and conversely,
for a cluster $C$, the unique end $\Upsilon(C)$ containing $C$
consists of the union of all principal filters of the vertices
$v\in C$ (with respect to the basepoint order).

Let ${\Phi}(G)$ and $\ccC(G)$ denote the set of all ends and all
clusters of $G$, respectively.  An \emph{end-isomorphism} between two
ends $\Upsilon$ and $\Upsilon'$ of $G$ is a label-preserving mapping
$f$ between $\Upsilon$ and $\Upsilon'$ such that $f$ is a graph
isomorphism and $f$ maps $C(\Upsilon)$ to $C(\Upsilon')$.  Then $G$ is
called a \emph{context-free graph}~\cite{MullerSchupp} if ${\Phi}(G)$
has only finitely many isomorphism classes under end-isomorphisms.
Since $G$ has uniformly bounded degree, each cluster $C(\Upsilon)$ is
finite. Moreover, a context-free graph $G$ has only finitely many
isomorphism classes of clusters. Thus, there exists $\delta<\infty$
such that the diameters of the clusters of $G$ are bounded by
$\delta$. By~\cite[Proposition 12]{ChDrEsHaVa} any graph $G$ whose
diameters of clusters is uniformly bounded by $\delta$ is
$\delta$-hyperbolic (in fact, $G$ is quasi-isometric to a tree). Note
that the converse is not true (see the 1-skeleton of the square
complex $\tZ$ described in Section~\ref{th-counterexample}).

\subsection{Some Results from MSO Theory}\label{ssec:MSO}

In this subsection, we recall some results from MSO theory of
undirected graphs, directed labeled graphs, latices and posets, and
event structures. 
Among the MSO theories of various discrete structures, the MSO theory
of undirected graphs is probably the most complete, with various and
deep applications (see the book of Courcelle and
Engelfriet~\cite{CouEn}). Let $G=(V,E)$ be an undirected and unlabeled
graph.  The MSO logic as introduced in Subsection~\ref{MSO-def} only
allow quantifications over (subsets of) vertices of $G$. This theory
is usually denoted by $\MSO_1(G)$. In order to allow also
quantifications over (subsets of) edges, an extended representation of
a graph is used. This is the relational structure
$G^{e}=(V\cup E, \inc)$, where
$\inc=\{ (e,v)\in E\times V: \exists u\in V \text{ such that } e\in \{
uv, vu\}\}$.  The MSO theory of this relational structure $G^{e}$ is
usually denoted by $\MSO_2(G)$. Seese~\cite{See} proved the following
fundamental result about $\MSO_2$ decidability:

\begin{theorem}[\cite{See}]\label{th-seese}
  If $\MSO_2(G)$ is decidable, then $G$ has finite treewidth.
\end{theorem}

The converse of Seese's theorem is not true: one can construct trees
with undecidable $\MSO_2$ theory. On the other hand,
Courcelle~\cite{Cou2} proved that for any integer $k$ the class of all
graphs of treewidth at most $k$ has a decidable $\MSO_2$ theory.  If
$\MSO_2(G)$ is decidable, then $\MSO_1(G)$ is also decidable. Again,
the reverse implication is not true. However, Courcelle~\cite{Cou}
proved that the converse holds for graphs of bounded degree:

\begin{theorem}[\cite{Cou}]\label{th-courcelle}
  If $G$ is a graph with uniformly bounded degree and $\MSO_1(G)$ is
  decidable, then $\MSO_2(G)$ is also decidable.
\end{theorem}

Now, consider labeled directed graphs. Let $\Sigma$ be a finite
alphabet. A $\Sigma$-\emph{labeled directed graph} is a relational
structure $\oG=(V,{(E_a)}_{a\in \Sigma})$, where $V$ is the set of
vertices and $E_a\subseteq V\times V$ is the set of $a$-labeled
directed edges.  Denote by $\MSO(\oG)$ the MSO theory of this
relational structure.  We associate to $\oG$ the unlabeled graph
$G=(V,\bigcup_{a\in \Sigma} \{ uv: u\ne v, (u,v)\in E_a \text{ or }
(v,u)\in E_a\})$.

M\"uller and Schupp~\cite{MullerSchupp} proved the following
fundamental theorem about $\Sigma$-labeled pointed context-free graphs
of bounded degree (and directed according to the basepoint order):

\begin{theorem}[\cite{MullerSchupp}]\label{th-mullerschupp}
  If $\oG$ is a context-free graph, then $\MSO(\oG)$ is decidable.
\end{theorem}

For a directed graph $\oG$, denote by $\Aut(\oG)$ its group of
automorphisms. Two vertices $u,v$ of $\oG$ belongs to the same orbit
of $\Aut(\oG)$ if there exists $f \in \Aut(\oG)$ such that $f(u) = v$.
Kuske and Lohrey~\cite{KuLo} established a kind of converse to
Theorem~\ref{th-mullerschupp} (the formulation of
Theorem~\ref{mso-graph} is inspired by this theorem, but the proofs
are different):

\begin{theorem}[\cite{KuLo}]\label{th-kuskelohrey}
  Let $\oG$ be a $\Sigma$-labeled connected graph of bounded degree
  such that $\Aut(\oG)$ has only finitely many orbits on $\oG$. Then
  conditions (1)-(3) are equivalent:
  \begin{enumerate}
  \item $\MSO(\oG)$ is decidable;
  \item $G$ has finite treewidth;
  \item $\oG$ is context-free.
  \end{enumerate}
\end{theorem}

Kuske~\cite{Ku} characterized the decidability of the MSO logic of
distributive lattices. Let $\cL=(L,\le)$ be a distributive lattice.
Let $\JJ(\cL)$ denote the set of join irreducible elements of $\cL$;
$(\JJ(\cL),\le)$ can be viewed as a poset.  An \emph{antichain} is a
set of pairwise incomparable elements. Denote by $w(\cL)$ the
\emph{width} of $\cL$, i.e., the supremum of the cardinalities of its
antichains.  For a poset $(L,\le)$, $\MSO(L,\le)$ is the MSO theory of
the relational structure $(L,\le)$.

\begin{theorem}[\cite{Ku}]\label{th-kuske}
  Let $\cL$ be a distributive lattice. Then $\MSO(\cL)$ is decidable
  if and only if $\MSO(\JJ(\cL))$ is decidable and the width $w(\cL)$
  is bounded.
\end{theorem}

Since distributive lattices are exactly the domains of conflict-free
event structures and there exists a bijection between join
irreducibles and the events of that event structure
(Corollary~\ref{geodesicprimetrace}), Theorem~\ref{th-kuske} can be
viewed as a result about decidability of MSO theory of conflict-free
event structures (graphs and event structures). That the MSO theory of
trace conflict-free event structures is decidable follows from a more
general result of Madhusudan~\cite{Mad}:

\begin{theorem}[\cite{Mad}]\label{th-madhusudan}
  The MSO theory of a trace event structure $\cE$ is decidable
  provided quatifications over sets are restricted to conflict-free
  subsets of events. In particular, if $\cE$ is conflict-free, then
  $\MSO(\cE)$ is decidable.
\end{theorem}

\subsection{Grids}

In this section we consider several types of square grids, which
characterize different properties of event structures and their
graphs.  The infinite \emph{square grid} $\Gamma$ is the graph whose
vertices correspond to the points in the plane with nonnegative
integer coordinates and two vertices are connected by an edge whenever
the corresponding points are at distance 1. The $n\times n$ square
grid $\Gamma_n$ is the subgraph of $\Gamma$ whose vertices are all
vertices of $\Gamma$ with $x$- and $y$-coordinates in the range
$0,\ldots,n$. $\Gamma$ and $\Gamma_n$ can be viewed as directed graphs
with respect to the basepoint order with respect to the corner
$(0,0)$.  Note that $\Gamma$ is the domain of the event structure
consisting of two pairwise disjoint sets
$X=\{x_0,x_1,x_2,\ldots\},Y=\{y_0,y_1,y_2,\ldots\}$ of events, such
that $x_0<x_1<x_2<\cdots$ and $y_0<y_1<y_2<\cdots$, and all events of
$X$ are concurrent with all events of $Y$. This event structure is
conflict-free and trace-regular.  Below, if not specified, by
$\Lambda$ we denote either of the grids $\Gamma$ or $\Gamma_n$. A
\emph{directed grid} $\oLambda$ is a grid $\Lambda$ with basepoint
orientation with respect to the origin $(0,0)$.

By Theorem~\ref{th-robseymour}, the treewidth of a graph is
characterized by square grid minors. We say that a square grid
$\Lambda$ is a \emph{grid minor} of a graph $G$ if $\Lambda$ is a
minor of $G$.

By Lemma~\ref{hyp_median}, the hyperbolicity of a median graph (event
domain or 1-skeleton of a CAT(0) cube complex) is characterized by
isometrically embedded square grids. We say that a square grid
$\Lambda$ is an \emph{isometric grid} of a median graph $G=(V,E)$ if
there exists an isometric embedding of $\Lambda$ in $G$. An event
structure characterization of isometric grids is provided below.

A stronger version of isometric grid is the notion of a flat grid. We
say that an isometric grid $\Lambda$ is a \emph{flat grid} of a median
graph $G$ if for any two vertices $x,y$ of $\Lambda$ at distance 2,
any common neighbor $z$ of $x$ and $y$ in $G$ belongs to $\Lambda$.
Since any locally-convex connected subgraph of $G$ is convex
(Lemma~\ref{convex}), any flat grid is convex. If $G$ is the
$1$-skeleton of a 2-dimensional cube complex, then any isometric grid
is flat.  If $\Lambda$ is a flat grid of the graph $G(\cE)$ of an
event domain $\cD(\cE)$, then there are two disjoint subsets
$X=\{x_0,x_1,x_2,\ldots\},Y=\{y_0,y_1,y_2,\ldots\}$ of events of $\cE$
such that $x_0\lessdot x_1\lessdot x_2\lessdot\cdots$ and
$y_0\lessdot y_1\lessdot y_2\lessdot\cdots$, and all events of $X$ are
concurrent with all events of $Y$.

The minor of a graph is defined by contracting edges. Minors are also
implicitly used in the theory of event structures, namely, to define
the event structure $\cE\setminus c$ rooted at a configuration
$c$. The domain of $\cE\setminus c$ is the principal filter $\cF(c)$
of $c$, $\cF(c)$ is a convex subgraph of $G(\cE)$, and thus $\cF(c)$
is the intersection of all halfspaces containing $\cF(c)$. Therefore,
$\cF(c)$ can be obtained from the median graph $G(\cE)$ of $\cE$ by
contracting all hyperplanes which do not intersect $\cF(c)$.

Given a median graph $G$ and a hyperplane $H$ of its CAT(0) cube
complex, the median graph $G'$ is obtained by
\emph{hyperplane-contraction} of $G$ with respect to $H$ if $G'$ is
obtained from $G$ by contracting all edges of $G$ dual to $H$.  We say
that a median graph $G'$ is a \emph{strong-minor} of a median graph
$G$ if $G'$ can be obtained from $G$ by hyperplane-contraction of a
set of hyperplanes of $G$.

Finally recall the event structure $\cE_{TY}=(E,\le, \#)$ occurring in
the definition of grid-free event structures.  Recall that $E$
consists of three pairwise disjoint sets $X,Y,Z$ such that:
\begin{itemize}
\item $X=\{x_0,x_1,x_2,\ldots\}$ is an infinite set of events with
  $x_0<x_1<x_2<\cdots$.
\item $Y=\{y_0,y_1,y_2,\ldots\}$ is an infinite set of events with
  $y_0<y_1<y_2<\cdots$.
\item $X\times Y\subseteq \parallel$.
\item There exists an injective mapping $g: X\times Y\rightarrow Z$
  satisfying: if $g(x_i,y_j)=z$ then $x_i<z$ and $y_j<z$. Furthermore,
  if $i'>i$ then $x_{i'}\nless z$ and if $j'>j$ then $y_{j'}\nless z$.
\end{itemize}
The domain of $\cE_{TY}$ contains the infinite square grid $\Lambda$
as a strong-minor. This grid corresponds to the events defined by the
sets $X$ and $Y$ and is obtained by contracting all hyperplanes
corresponding to the events in $E\setminus (X\cup Y)$. On the
other hand, the events from $Z$ correspond to the hairs attached to
the grid $\Lambda$ in the definition of the hairing of an event
structure.  However, the relationship between the events of $Z$ or the
events of $Z$ and a part of events of $X\cup Y$ is not specified, thus
one cannot say more about the structure of the domain of $\cE_{TY}$.

We continue with relationships between isometric grids and
hyperbolicity.

\begin{lemma}\label{lem-isometric-sq}
  Let $\cE=(E,\le, \#)$ be an event structure of bounded degree. If
  the directed median graph $\oG(\cE)$ contains an isometric
  $n\times n$ directed square grid, then $E$ contains two disjoint
  conflict-free sets of events
  $A=\{ x_0,x_1,\ldots, x_{n-1}\},B=\{ y_0,y_1,\ldots,y_{n-1}\}$ such
  that $x_i \parallel y_j$ for any two events $x_i\in A, y_j\in B$.
  Conversely, if for any $n \in \NN$, $E$ contains two disjoint sets
  of events
  $A=\{ x_0,x_1,\ldots, x_{n-1}\},B=\{ y_0,y_1,\ldots,y_{n-1}\}$ such
  that $x_i \parallel y_j$ for any two events $x_i\in A, y_j\in B$,
  then the median graph $G(\cE)$ is not hyperbolic, and thus contains
  arbitrarily large isometric square grids.
\end{lemma}

\begin{proof}
  If $G(\cE)$ contains an isometric $n\times n$ directed grid
  $\oLambda$, then let $X=\{ x_0,\ldots,x_{n-1}\}$ denote the events
  defining the edges of one side of $\oLambda$ and let
  $Y=\{y_0,\ldots,y_{n-1}\}$ denote the events defining the edges of
  another incident side of $\oLambda$. Since each hyperplane $H_{x_i}$
  intersects each hyperplane $H_{y_j}$ we conclude that the events of
  $X$ are concurrent with the events of $Y$. It remains to show that
  two events of $X$ or two events of $Y$ cannot be in conflict. Pick
  any $x_i,x_{j}\in X$ with $i<j$ and suppose that $X$ define
  horizontal edges of $\Lambda$.  Then in $\oLambda$ the hyperplane
  $H_{x_i}$ separates the origin of the grid from the carrier of
  $H_{x_{j}}$. This implies that $x_i$ and $x_j$ cannot be in
  conflict.

  To prove the converse, we use Theorem~\ref{hyp_median_hagen}
  of~\cite{Ha}: we show that for any $n$, the crossing graph
  $\Gamma(X(\cE))$ contains a complete bipartite subgraph $K_{n,n}$.
  Suppose that the maximum degree of $\cE$ is $d$. Recall that the
  \emph{Ramsey theorem} asserts that for any two integers $r$ and $s$
  there exists a least positive integer $R(r,s)$ such that any graph
  with at least $R(r,s)$ vertices either contains a stable set of size
  $r$ or a clique of size $s$. Let $m\ge R(n,d+1)$. Then $E$ contains
  two disjoint sets of events
  $A=\{ x_0,x_1,\ldots, x_{m-1}\},B=\{ y_0,y_1,\ldots,y_{m-1}\}$ such
  that $x_i \parallel y_j$ for any two events $x_i\in A, y_j\in
  B$. Recall that two events $e$ and $e'$ of $E$ are concurrent iff
  their hyperplanes $H_{e}$ and $H_{e'}$ intersect, i.e., $H_e$ and
  $H_{e'}$ are adjacent in $\Gamma(X(\cE))$. Consequently, $H_{x_i}$
  and $H_{y_j}$ are adjacent in $\Gamma(X(\cE))$ for any $x_i\in A$
  and $y_j\in B$. Let $\Gamma'$ (respectively, $\Gamma''$) be the
  subgraph of $\Gamma(X(\cE))$ induced by the hyperplanes defined by
  the events of $A$ (respectively, of $B$). Since $\Gamma'$ contain
  $m\ge R(n,d+1)$ vertices, by Ramsey's theorem, $\Gamma'$ either
  contains a stable set $A'$ of size $n$ or a clique $C'$ of size
  $d+1$. In the second case we conclude that $X(\cE)$ contains $d+1$
  pairwise intersecting hyperplanes. By Proposition~\ref{Helly-CAT0},
  this implies that $X(\cE)$ contains a $(d+1)$-cube $Q$. Since the
  orientation of the edges of $X(\cE)$ is admissible, $Q$ contains a
  source of degree $d+1$, contrary to the assumption that the maximum
  degree of $\cE$ is $d$.  Consequently, $\Gamma'$ contains a stable
  set $A'$ of size $n$. Similarly, $\Gamma''$ contains a stable set
  $B'$ of size $n$. But then $A'\cup B'$ induce the complete bipartite
  graph $K_{n,n}$ in the crossing graph $\Gamma(X(\cE))$.
\end{proof}

\begin{proposition}\label{hyperbolic-gridfree}
  If the graph $G(\cE)$ of an event structure $\cE$ of bounded degree
  is hyperbolic, then $\cE$ is grid-free.
\end{proposition}

\begin{proof}
  Suppose by way of contradiction that $\cE$ contains three disjoint
  infinite sets of events $X,Y,Z$ defining the event structure
  $\cE_{TY}$. Since every $x \in X$ is concurrent with every
  $y \in Y$, applying Lemma~\ref{lem-isometric-sq} with $A = X$ and
  $B = Y$, we deduce that $G(\cE)$ is not hyperbolic, a contradiction.
\end{proof}

\subsection{Proof of Theorem~\ref{mso-graph}}

Since for a $\Sigma$-labeled directed graph $\oG$, the decidability of
$\MSO(\oG)$ implies the decidability of $\MSO_1(G)$,
(1)$\Rightarrow$(2).  Since the degrees of the vertices of $G(\cE)$
are uniformly bounded, the implication (2)$\Rightarrow$(3) follows
from Courcelle's Theorem~\ref{th-courcelle}~\cite{Cou}.  The
implication (3)$\Rightarrow$(4) is a particular case of Seese's
Theorem~\ref{th-seese}~\cite{See}. Finally, the implication
(6)$\Rightarrow$(1) follows from the M\"uller and Schupp
Theorem~\ref{th-mullerschupp}~\cite{MullerSchupp} that the MSO theory
of context-free graphs is decidable. It remains to establish the
implications (4)$\Rightarrow$(5) and (5)$\Rightarrow$(6).

\paragraph{(4)$\Rightarrow$(5)}
Suppose by way of contradiction that $G(\cE)$ has clusters of
arbitrarily large diameters. In this case, for any $n$ we construct in
$G(\cE)$ a half of the square $n\times n$ grid as a minor (denote this
half-grid by $\frac{1}{2}\Gamma_n$). Since $\frac{1}{2}\Gamma_n$
contains the $\frac{n}{2}\times \frac{n}{2}$ square grid, we deduce
that $G(\cE)$ contains arbitrarily large square grids as minors,
contradicting that the treewidth of $G(\cE)$ is finite.  Let
$\{ z_{i,j}: 0\le i\le n, 0\le j\le n, \text{ and } i+j\le n\}$ be the
set of vertices of $\frac{1}{2}\Gamma_n$.

Let $v_0$ be the basepoint. Recall that $S(v_0,k)$ is the sphere of
radius $k$ centered at $v_0$. We need the following properties of
clusters of $G(\cE)$ (which hold for all median graphs):

\begin{claim}\label{zigzag1}
  Let $u,v$ be two vertices in a common cluster $C$ of $G(\cE)$ at
  distance $k$ from $v_0$. Then there exists a $(u,v)$-path
  $P'(u,v)=(u,p_1,q_1,p_2,q_2,\ldots,p_{m-1},q_{m-1},p_{m},v)$ such
  that $Q_1(u,v)=\{ p_1,\ldots,p_{m}\}\subseteq S(v_0,k+1)$ and
  $Q_2(u,v)=\{ q_1,\ldots, q_{m-1}\}\subseteq C\subseteq S(v_0,k)$.
\end{claim}

\begin{proof} 
  Since $u,v$ belong to a common cluster $C\subseteq S(v_0,k)$, there
  exists a $(u,v)$-path $P(u,v)$ in $G(\cE) \setminus S(v_0,k-1)$.
  Among all such paths, let $P'(u,v)$ be a path minimizing the sum
  $\sum_{w\in P'(u,v)} d(v_0,w)$. We assert that all vertices of
  $P'(u,v)$ have distance $k$ or $k+1$ to $v_0$. Suppose that $x$ is a
  furthest from $v_0$ vertex of $P'(u,v)$ and that
  $k':=d(v_0,x)\ge k+2$. Let $y$ and $z$ are the neighbors of $x$ in
  $P'(u,v)$. From the choice of $x$ and since $G(\cE)$ is bipartite it
  follows that $d(v_0,y)=d(v_0,z)=k'-1$. By quadrangle condition,
  there exists a vertex $x'$ adjacent to $y$ and $z$ such that
  $d(v_0,x') = k'-2\ge k$. Replacing $x$ by $x'$ in $P'(u,v)$, we
  obtain a path $P'_0(u,v)$ in $G(\cE) \setminus S(v_0,k-1)$ such that
  $\sum_{w\in P'_0(u,v)} d(v_0,w) < \sum_{w\in P'(u,v)} d(v_0,w)$,
  contradicting the choice of $P'(u,v)$. Therefore all vertices of
  $P'(u,v)$ have distance $k$ or $k+1$ from $v_0$. Since the ends
  $u,v$ of $P'(u,v)$ have distance $k$ to $v_0$ and $G(\cE)$ is
  bipartite, the path $P'(u,v)$ is zigzagging, i.e.,
  $P'(u,v) = (u,p_1,q_1,p_2,q_2,\ldots,p_{m-1},q_{m-1},p_{m},v)$ and
  we can set $Q_1(u,v):=\{ p_1,\ldots,p_{m}\}\subseteq S(v_0,k+1)$ and
  $Q_2(u,v):=\{ q_1,\ldots, q_{m-1}\}\subseteq C\subseteq S(v_0,k)$.
\end{proof}

\begin{claim}\label{zigzag2}
  Let $u,v$ be two vertices in a common cluster $C$ of $G(\cE)$ at
  distance $k$ from $v_0$. Then for any $(u,v)$-path
  $P_1(u,v)=(q_0=u,p_1,q_1,p_2,q_2,\ldots,p_{m-1},q_{m-1},p_{m},v=q_m)$
  such that $Q_1(u,v)=\{ p_1,\ldots,p_{m}\}\subseteq S(v_0,k+1)$ and
  $Q_2(u,v)=\{ q_1,\ldots, q_{m-1}\}\subseteq C\subseteq S(v_0,k)$,
  there exists a sequence of vertices
  $Q_3(u,v)=\{ r_1,\ldots,r_{m'}\}\subseteq S(v_0,k-1)$, such that
  $P_2(u,v)=(u,r_1,q_{i_1},r_2,q_{i_2},\ldots,q_{i_{m'-1}},r_{m'},v)$
  is a $(u,v)$-path of $G(\cE)$ and $r_1$ is adjacent to
  $q_0=u,q_1,\ldots,q_{i_1}$, $r_2$ is adjacent to
  $q_{i_1},q_{i_1+1},\ldots,q_{i_2}$, etc, and $r_{m'}$ is adjacent to
  $q_{i_{m'-1}},\ldots, q_{m-1},u=q_m$.
\end{claim}

\begin{proof}
  Since $d(v_0,q_0)=d(v_0,q_1)=k$ and $d(v_0,p_1)=k+1$, by quadrangle
  condition there exists a vertex $r_1$ adjacent to $q_0$ and $q_1$
  such that $d(v_0,r_1) = k-1$. Let $q_{i_1}$ be the last vertex of
  $Q_2(u,v)$ such that $r_1$ is adjacent to
  $q_0,q_1,\ldots,q_{i_1}$. Again, since
  $d(v_0,q_{i_1})=d(v_0,q_{{i_1}+1})=k$ and $d(v_0,p_{{i_1}+1})=k+1$,
  by quadrangle condition there exists a vertex $r_2$ adjacent to
  $q_{i_1}$ and $q_{{i_1}+1}$ such that $d(v_0,r_2) = k-1$. Since
  $r_1$ is not adjacent to $q_{{i_1}+1}$, we have $r_1\ne r_2$. Let
  $q_{i_2}$ be the last vertex of $Q_2(u,v)$ such that $r_2$ is
  adjacent to $q_{i_1},q_{{i_1}+1},\ldots, q_{i_2}$. Continuing this
  way, we define all vertices of $Q_3(u,v)=\{
  r_1,\ldots,r_{m'}\}$. Then
  $P_2(u,v)=(u,r_1,q_{i_1},r_2,q_{i_2},\ldots,q_{i_{m'-1}},r_{m'},v)$
  is a $(u,v)$-path.
\end{proof}

We call the union of paths $P_1(u,v)$ and $P_2(u,v)$ a \emph{fence}
and denote it by $F(u,v)$. Call $P_1(u,v)$ the \emph{upper path} and
$P_2(u,v)$ the \emph{lower path} of $F(u,v)$.  Notice that
$P_1(u,v)\cap P_2(u,v)\subseteq C$.  From the definition of clusters,
all vertices of $P_2(u,v)\setminus P_1(u,v)$ also belong to a common
cluster $C'$.  If $d(u,v)=n'$, then both paths $P_1(u,v)$ and
$P_2(u,v)$ have length at least $n'$. Thus, setting
$u':=r_1, v':=r_{m'}$ and denoting by $P_1(u',v')$ the subpath of
$P_2(u,v)$ between $u'$ and $v'$, we conclude that its length is at
least $n'-2$. On the other hand, the length of $P_1(u',v')$ is at most
$n''-2$, where $n''$ is the length of $P_1(u,v)$.  Applying
Claim~\ref{zigzag2} to $P_1(u',v')$ we define the path $P_2(u',v')$
and the fence $F(u',v')$. Note that $P_1(u,v')=F(u,v)\cap
F(u',v')$. Continuing this way, after $\frac{n'}{2}\le n\le n''$
steps, we find two sequences of vertices
$Q_u=(u=u_{n},u_{n-1}=u',u_{n-2},\ldots,u_{1},u_0=w)$ and
$Q_v=(v=v_n,v_{n-1}=v',v_{n-2},\ldots,v_{1},v_0=w)$ (constituting
shortest $(u,w)$- and $(v,w)$-paths) and for each pair $u_i,v_i$ a
fence $F(u_i,v_i)=P_1(u_i,v_i)\cup P_2(u_i,v_i)$ such that any two
consecutive fences $F(u_{i+1},v_{i+1})$ and $F(u_i,v_i)$ intersect in
the path $P_1(u_{i},v_{i})$. We denote the union of all fences
$F(u_i,v_i)$, $i=n,\ldots,0$, by $F^*$. We also denote by $C_i$ the
cluster containing $u_i$ and $v_i$ (in particular, $C_n=C$).

We assert that $F^*$ contains the half-grid $\frac{1}{2}\Gamma_n$ as a
minor. Since $\frac{n'}{2}\le n\le n''$ and $n''\ge n'$, we will be
done.  For this, for each vertex $z_{i,j}$ of $\frac{1}{2}\Gamma_n$ we
define a connected subgraph $Z_{i,j}$ of $F^*$ satisfying the
following properties:
\begin{enumerate}
\item $Z_{0,i}=\{ u_i\}$ and $Z_{i,0}=\{ v_i\}$ for each
  $i=0,\ldots,n$;
\item for each $k=0,\ldots, 2n$, if $i+j=k$, then $Z_{i,j}$ is a
  subpath of the lower path $P_2(u_k,v_k)\subseteq C_k\cup C_{k-1}$ of
  the fence $F(u_k,v_k)$ and $Z_{i,j}$ starts and ends at cluster
  $C_i$;
\item for each $k=0,\ldots, 2n$ the paths $Z_{i,j}$ with $i+j=k$ are
  pairwise disjoint and are lexicographically ordered along
  $P_2(u_k,v_k)$ from $u_k$ to $v_k$ (i.e., for two pairs $(i,j)$ and
  $(i',j')$ with $i+j=i'+j'=k$, the path $Z_{i,j}$ appears before the
  path $Z_{i',j'}$ in $P_2(u_k,v_k)$ iff $i<i'$);
\item for each pair $(i,j)$, the first vertex of the path $Z_{i,j}$ is
  adjacent to the last vertex of the path $Z_{i-1,j}$ and the last
  vertex of $Z_{i,j}$ is adjacent to the first vertex of the path
  $Z_{i,j-1}$.
\end{enumerate}

From conditions (3) and (4) we deduce that the paths $Z_{i,j}$ are
pairwise disjoint. Therefore, contracting them, we get
$\frac{1}{2}\Gamma_n$ as a minor.

We construct the paths $Z_{i,j}$ recursively. Suppose that the paths
$Z_{i,j}$ satisfying the previous conditions have been defined for all
pairs $(i,j)$ such that $i+j\le k$ and we have to define the paths
$Z_{i,j}$ with $i+j=k+1$. We proceed lexicographically on all such
pairs. Consider a current pair $(i,j)$ with $i+j=k+1$. By induction
assumption, the paths $Z_{i-1,j}$ and $Z_{i,j-1}$ have been
defined. Let $x$ denote the last vertex of the path $Z_{i-1,j}$ and
$y$ denote the first vertex of the path $Z_{i,j-1}$.  By definition
and induction hypothesis, the paths $Z_{i-1,j}$ and $Z_{i,j-1}$ are
contained in the clusters $C_{k-1}\cup C_{k-2}$, are disjoint, and
start and end at vertices of $C_{k-1}$. Consequently, $x$ and $y$ are
vertices of $C_{k-1}$ and $x$ appears before $y$ in the path
$P_2(u_{k-1},v_{k-1})$. In particular, $x$ and $y$ belong to the path
$P_2(u_k,v_k)$. Traverse $P_2(u_k,v_k)$ from $u_k$ to $v_k$. Denote by
$x'$ the vertex appearing after $x$ in $P_2(u_k,v_k)$ and by $y'$ the
vertex of $P_2(u_k,v_k)$ appearing before $y$. Denote by $Z_{i,j}$ the
subpath of $P_2(u_k,v_k)$ comprised between $x'$ and $y'$. Then $x'$
and $y'$ are respectively the first and last vertices of $Z_{i,j}$.

We show that $Z_{i,j}$ satisfies the conditions (2)-(4). Condition (4)
follows from the definition of the vertices $x,y,x',y'$ and of the
path $Z_{i,j}$.  Since $x,y\in C_{k-1}$, from the definition of
$P_2(u_k,v_k)$ it follows that $x',y'\in C_k$. Hence $Z_{i,j}$
satisfies (2). Since the paths $Z_{i',j'}$ with $i'+j'=k-1$ are
lexicographically ordered, from the definition of the paths $Z_{i,j}$
with $i+j=k$, it follows that such paths are also lexicographically
ordered and pairwise disjoint. This concludes the proof of the
implication.

\paragraph{(5)$\Rightarrow$(6)}
The implication follows from~\cite[Proposition 4.4]{BaDaRa} and the
fact that trace event structures are recognizable by trace automata.
Here we present a different (and hopefully simpler) proof. Let
${\Phi}(G(\cE))$ be the set of ends of $G(\cE)$.  We have to prove
that ${\Phi}(G(\cE))$ has only finitely many isomorphism classes under
end-isomorphisms. Recall that there exists a bijection between the
ends of ${\Phi}(G(\cE))$ and the clusters of $\ccC(G(\cE))$ and that
for a cluster $C$ we denote by $\Upsilon(C)$ the end containing $C$.
Let $M$ be the size of the alphabet $\Sigma$. Since $\cE$ is a
trace-regular event structure, the degrees of vertices of $G(\cE)$ are
uniformly bounded, say by some constant $\Delta$. Suppose that the
diameters of clusters are uniformly bounded by $D$.  We say that the
sets of a set family $\cS$ have \emph{constant size} if the sizes of
all sets of $\cS$ depend only of $M,\Delta$, and $D$. First notice
that all clusters $C$ have constant size.  Indeed, pick a vertex $v$
of $C$. Then $C$ is included in the ball $B(v,D)$ of radius $D$
centered at $v$. This ball has at most
$K:=\sum_{i=0}^D \Delta^i=O(\Delta^D)$ vertices.

Pick any cluster $C$. Since $(\cD({\cE}),\subseteq)$ is a median meet
semilattice, there exists the smallest median meet sub-semilattice
$M(C)$ containing the set $C$, the \emph{median closure} of $C$. Let
$m_C\in M(C)$ denote the meet of $C$ in
$(\cD({\cE}),\subseteq)$. Denote by $H(C)$ the subgraph of $G(\cE)$
induced by the set
$\{ v\in M(C): m_C\le v \text{ and } \exists x\in C, v\le x\}$.

\begin{claim}\label{immediate-past}
  For any $C\in \ccC(G(\cE))$, the graph $H(C)$ has constant size.
\end{claim}

\begin{proof}
  Any $H(C)$ is included in the median closure $M(C)$ and any $M(C)$
  is included in the convex hull $\conv(C)$ of $C$. Therefore it
  suffices to prove that $\conv(C)$ has constant size, namely that if
  has constant diameter. Pick $x,y\in \conv(C)$. The distance $d(x,y)$
  in $G(\cE)$ is the number of hyperplanes separating $x$ and $y$.
  Since $\conv(C)$ is the intersection of all halfspaces of $G(\cE)$
  including $C$, the hyperplanes defining such halfspaces do not
  separate $x$ and $y$.  Therefore $x$ and $y$ can be separated only
  by hyperplanes separating vertices of $C$. There are at most $D$
  hyperplanes separating two given vertices of $C$, thus there are at
  most $DK^2$ hyperplanes separating vertices of $C$.  Consequently,
  $d(x,y)\le DK^2$, establishing that the diameter of $\conv(C)$ is
  constant.
\end{proof}

Suppose that the edges of each $H(C)$ are directed and labeled by
$\lambda$ as in $G(\cE)$. For each vertex $v\in H(C)$, let $r(v)=i$ if
the principal filter $\cF(v)$ of $v$ belongs to the isomorphism class
$i$. Call the edge- and vertex-labeled graph $(H(C),\lambda,r)$ the
\emph{recent past} of the cluster $C$.  Since by
Claim~\ref{immediate-past} all $H(C)$ have constant size, $\lambda$ is
finite, and there exists only a finite number of types of principal
filters, we conclude that there exists only a finite number of types
of recent pasts $\cP_1,\ldots,\cP_n$.

Pick any isomorphism class $\cP_i$ and pick any two graphs $H(C)$ and
$H(C')$ belonging to $\cP_i$. Notice that since any isomorphism $g$
between $H(C)$ and $H(C')$, preserves the orientation of edges, $g$
maps the unique source $m_C$ of $H(C)$ to the unique source $m_{C'}$
of $H(C')$. The set $C$ of sinks of $H(C)$ is mapped to the set $C'$
of sinks of $H(C')$.  Since $r(m_C)=r(m_{C'})$, there exists an
isomorphism $f$ between the labeled principal filters $\cF(m_C)$ and
$\cF(m_{C'})$.

\begin{claim}\label{recent-past-isomorphism}
  Any isomorphism $g$ between $(H(C),\lambda,r)$ and
  $(H(C'),\lambda,r)$ coincides with $f$, i.e., for any vertex $v$ of
  $H(C)$, $f(v)=g(v)$.
\end{claim}

\begin{proof}
  Let $m_C\in P(C)$ and $m_{C'}\in P(C')$ be the meets of $C$ and
  $C'$. Let $g$ be an isomorphism between $(H(C),\lambda,r)$ and
  $(H(C'),\lambda,r)$.  By induction on $k:=d(v,m_C)$ we prove that
  $g(v)=f(v)$. Since any isomorphism maps $m_C$ to $m_{C'}$, this is
  true for $k=0$. Since the labeling $\lambda$ is deterministic, any
  isomorphism map the edge $m_{C}v$ to the unique edge $m_{C'}v'$
  labeled $\lambda(m_{C}v)$. Therefore $g(v)=v'=f(v)$ for any neighbor
  $v$ of $m_C$ in $H(C)$. Therefore, our assertion is also true for
  $k=1$. Suppose that it is true for all vertices $v$ of $H(C)$ such
  that $d(v,m_C)\le k$ and pick a vertex $w$ with $d(w,m_C)=k+1$. Let
  $v$ be a neighbor of $w$ at distance $k$ from $m_C$. By induction
  assumption, $g(v)=f(v)$, denote this vertex by $v'$. Set $w':=g(w)$
  and $w'':=f(w)$. Since $vw$ is an edge outgoing from $v$, $v'w'$ and
  $v'w''$ are edges outgoing from $v'$, both labeled by $\lambda(vw)$
  ($v'w'$ is an edge of $H(C')$ but the vertex $w''$ and the edge
  $v'w''$ are not necessarily in $H(C')$).  Since the labeling
  $\lambda$ is deterministic, this is possible only if $w'=w''$,
  whence $g(w)=f(w)$.
\end{proof}

The following claim concludes the proof of the implication
(5)$\Rightarrow$(6):

\begin{claim}\label{f-end-isomorphism}
  $f$ is an end-isomorphism between $\Upsilon(C)$ and $\Upsilon(C')$.
\end{claim}

\begin{proof}
  Since $\Upsilon(C)$ is the union of all principal filters
  $\cF(v), v\in C$, we have $\Upsilon(C)\subseteq \cF(m_C)$, thus $f$
  is well-defined on $\Upsilon(C)$. Since $f$ is a bijective map
  between $\cF(m_C)$ and $\cF(m_{C'})$, $f$ is an injective map from
  $\Upsilon(C)$ to $\cF(m_{C'})$.  By
  Claim~\ref{recent-past-isomorphism}, $f$ maps $C$ to $C'$, thus the
  $f$-image of any principal filter $\cF(v)$ with $v\in C$ is a
  principal filter $\cF(f(v))$ with $f(v)\in C'$, thus
  $f(\Upsilon(C))\subseteq \Upsilon(C')$. Since any vertex of
  $\Upsilon(C')$ belongs to at least one principal filter $\cF(v')$
  with $v'\in C'$ and $f$ bijectively maps $C$ to $C'$, $f$ is a
  surjective map from $\Upsilon(C)$ to $\Upsilon(C')$. Since $f$ is
  also injective on $\Upsilon(C)$, $f$ is a bijection between
  $\Upsilon(C)$ and $\Upsilon(C')$. Since any edge $xy$ of
  $\Upsilon(C)$ belongs to at least one principal filter $\cF(v)$ with
  $v\in C$, $f$ maps $xy$ to an edge $x'y'$ of $\Upsilon(C')$ and
  $\lambda(x'y')=\lambda(xy)$ holds.  Since the same property holds
  for edges of $\Upsilon(C')$, this establishes that $f$ is an
  end-isomorphism between $\Upsilon(C)$ and $\Upsilon(C')$.
\end{proof}

\subsection{Proof of Proposition~\ref{mso-graph-complex}}

For a covering map $\varphi:Y \rightarrow X$, an automorphism
$\alpha: Y \to Y$ is a \emph{deck transformation} of $\varphi$ if
$\varphi \circ \alpha = \varphi$.  The set of deck transformations of
$\varphi$ forms a group under composition.  A covering map
$\varphi:Y \rightarrow X$ is called \emph{normal} (or \emph{regular})
if for each pair of lifts $y,y' \in Y$ of $x\in X$ there is a deck
transformation mapping $y$ to $y'$. If there is a normal covering map
$\varphi:Y \rightarrow X$, then $Y$ is called a \emph{normal} cover of
$X$.  Every universal cover is normal and its deck transformation
group is isomorphic to the fundamental group $\pi_1(X)$ of $X$
(see~\cite[Proposition 1.39]{Hat}). Since the special cube complex
$X_N$ of a net system $N$ is finite and its universal cover $\tX_N$ is
normal, this implies that the automorphism group of $\tX_N$ has a
finite number of orbits. Therefore, the automorphism group of the
directed labeled graph $\oG(\tbX_N)$ also has a finite number of
orbits.  Consequently, to $\oG(\tbX_N)$ we can apply
Theorem~\ref{th-kuskelohrey} of Kuske and Lohrey~\cite{KuLo} and
deduce Proposition~\ref{mso-graph-complex}.

\subsection{The MSO Theory of Hairings of Event Structures}

The goal of this subsection is to prove
Proposition~\ref{barycentric-subdivision},
Proposition~\ref{prop-hairing-N}, and
Theorem~\ref{barycentric-subdivision-MSO}.

\subsubsection{Proof of Propositions~\ref{barycentric-subdivision}}

Assume that $\lambda$ is a trace-regular labeling of $\cE$ over $M =
(\Sigma,I)$ and let $\dotlambda$ be the labeling of $\dotcE$ defined as
above. We first show that $\dotlambda$ is a trace labeling of $\dotcE$,
i.e., that $\dotlambda$ satisfies (LES1), (LES2), and (LES3).
For any two events $e, e' \in E$, the properties (LES1)-(LES3) are
satisfied because $\lambda$ is a trace-regular labeling of
$\cE$. Suppose now that $e'$ is a hair event. Hence $\dotlambda(e') = h$
and for any $a \in \Sigma \cup \{h\}$,
$(\dotlambda(e'),a) \notin \dotI$. Consequently (LES2) trivially
holds. Since a hair event is not concurrent with any other event,
(LES3) also trivially holds. If $e$ is in minimal conflict with $e'$,
then $e$ cannot be a hair event and thus
$\dotlambda(e) \neq \dotlambda(e')$, establishing (LES1).

We now show that $\dotlambda$ is a regular labeling of $\dotcE$. Consider
a configuration $\dotc$ of $\dotcE$ and observe that if $\dotc$ contains
a hair event $e_{c'}$ associated with a configuration $c'$ of $\cE$,
then $\dotc = c' \cup \{e_{c'}\}$. Consequently, for any such
configuration $\dotc$, $\dotcE^{\dotlambda}\setminus \dotc$ is
empty. Therefore, all such configurations are equivalent for
$R_{\dotcE^{\dotlambda}}$. Observe that any configuration $c_0$ of $\dotcE$
that does not contain a hair event is also a configuration of
$\cE$. Consider two configurations $c_0, c_0' \in \dotcE$ such that
$c_0 R_{\cE^{\lambda}} c_0'$ and let $f$ be an isomorphism from
$\cE \setminus c_0$ to $\cE \setminus c_0'$. We define an isomorphism
$\dotf$ from $\dotcE \setminus c_0$ to $\dotcE \setminus c_0'$ as follows.
For any event $e \in \cE \setminus c_0$, let $\dotf(e) = f(e)$. For
any configuration $c$ of $\cE\setminus c_0$,
$f(c) = \{f(e) : e \in c\}$ is a configuration of $\cE\setminus c_0'$
and we let $\dotf(e_c) = e_{f(c)}$. Observe that in any case,
$\dotlambda(\dotf(e))=\dotlambda(e)$. Consider any two events
$e_1, e_2 \in \dotcE\setminus \{c_0\}$. If
$e_1, e_2 \in \cE\setminus c_0$, then $\dotf(e_1) \dotleq \dotf(e_2)$
iff $e_1 \dotleq e_2$ and $\dotf(e_1) \dotdiese \dotf(e_2)$ iff
$e_1 \dotdiese e_2$ since $f$ is an isomorphism from
$\cE\setminus c_0$ to $\cE\setminus c_0'$. Suppose now that $e_2$ is a
hair event $e_c$ associated to a configuration $c$ of
$\cE\setminus c_0$. Then $e_1 \dotleq e_c$ if $e_1 \in c$ and
$e_1 \dotdiese e_c$ otherwise. In the first case, $f(e_1) \in f(c)$
and consequently $\dotf(e_1) = f(e_1) \dotleq e_{f(c)} =
\dotf(e_c)$. In the second case, $f(e_1) \notin f(c)$ and thus
$\dotf(e_1) = f(e_1) \# e_{f(c)} = \dotf(e_c)$. Since $f$ is
bijective, $\dotf$ is also bijective and thus $\dotf$ is an
isomorphism from $\dotcE\setminus c_0$ to $\dotcE\setminus
c_0'$. Consequently, since $R_{\cE^\lambda}$ has finite index, so does
$R_{\dotcE^{\dotlambda}}$, showing that $\dotlambda$ is a trace-regular
labeling of $\dotcE$. This ends the proof of
Proposition~\ref{barycentric-subdivision}.

\subsubsection{Proof of Proposition~\ref{prop-hairing-N}}

We first show that $X_{\dotN}$ and $\dotX_N$ are isomorphic.  For any
marking $m$ of $N$, consider the marking
$\dotm = m \cup \{p_a:a \in \Sigma\}$ of $\dotN$.  First note that for
all markings $m,m'$ of $N$ and $b \in \Sigma$, we have
$\dotm \xlongrightarrow{b} \dotm'$ in $\dotN$ iff
$m \xlongrightarrow{b} m'$ in $N$, i.e., the graph induced by
$\{\dotm : m \text{ is a vertex of } X_N\}$ in $X_{\dotN}$ is
isomorphic to $\oG_{N}$.  Moreover, for each vertex $m$ of $X_N$, we
have $\dotm = m \cup \{p_a:a \in \Sigma\} \xlongrightarrow{h} m$ in
$\dotN$, i.e., $(\dotm,m)$ is an arc of $\oG_{\dotN}$. Furthermore,
the vertex $m$ has out-degree $0$ since no transition can be fired
from $m$ in $\dotN$ as $m \cap \{p_a:a \in \Sigma\} = \varnothing$,
i.e., $(\dotm,m)$ is a pendant arc of $\oG_{\dotN}$.  Consequently,
the special cube complexes $X_{\dotN}$ and $\dotX_N$ coincide and thus
$\cF_{\tildo}(\tm_{0},\tdotX{_N^{(1)}})=
\cF_{\tildo}(\tdotm_{0},\tX{_{\dotN}^{(1)}})$.
To show that the event structures $\dotcE_N$ and $\cE_{\dotN}$ are
isomorphic, we show that their domains are isomorphic. By
Theorem~\ref{Petri-to-special-isomorphism},
$\cD(\dotcE_N) = \cD(\dotcE_{X_N})$ and
$\cD(\cE_{X_{\dotN}}) = \cD(\cE_{\dotN})$. Using
Equation~(\ref{eq-DcE}), we obtain
\[
  \cD(\dotcE_N) = \cD(\dotcE_{X_N}) =
  \cF_{\tildo}(\tm_{0},\tdotX{_N^{(1)}})=
  \cF_{\tildo}(\tdotm_{0},\tX{_{\dotN}^{(1)}}) = \cD(\cE_{X_{\dotN}})
  = \cD(\cE_{\dotN}).
\]

\subsubsection{Proof of Theorem~\ref{barycentric-subdivision-MSO}}

The proof of Theorem~\ref{barycentric-subdivision-MSO} is based on
Theorem~\ref{mso-graph} and Propositions~\ref{prop-MSOb-MSOG}
and~\ref{prop-MSOG-MSO}.


\begin{proposition}\label{prop-MSOb-MSOG}
  For a trace-regular event structure ${\cE}=(E,\le, \#,\lambda)$, if
  $\MSO(\dotcE)$ is decidable, then $\MSO(\oG(\cE))$ is decidable.
\end{proposition}

\begin{proof}
  We transform any formula $\varphi_G(\bvv,\bX) \in \MSO(\oG(\cE))$,
  where $\bvv= \{v_1,\ldots,v_n\}$ and $\bX = \{X_1,\ldots,X_m\}$,
  into a formula $\varphi_{\dotES}(\be,\bX) \in \MSO(\dotcE)$ where
  $\be = \{e_{v_1},\ldots,e_{v_n}\}$. The variables representing the
  vertices of $\oG(\cE)$ will be replaced by variables representing
  the corresponding hair events of $\dotcE$.
  We proceed by induction on the structure of
  $\varphi_G(v_1,\ldots,v_n,X_1,\ldots,X_m)$. We need to explain the
  transformations for atomic formulas, Boolean combinations of
  formulas, and existential quantifications over vertices and sets of
  vertices.

  We first consider atomic formulas.

  If $\varphi_G(v,X) = (v \in X)$, then we set
  $\varphi_{\dotES}(e_v,X) := (e_v \in X)$.

  If $\varphi_G(u,v) = (u = v)$, then we set
  $\varphi_{\dotES}(e_{u},e_{v}) := (e_{u} = e_{v})$.

  If $\varphi_G(u,v) = ((u,v) \in E_a)$ for some letter
  $a \in \Sigma$, then we set
  \[
    \varphi_{\dotES}(e_{u},e_{v}) := \exists e \big((\dotlambda(e) = a)
    \wedge (e \dotdiese_\mu e_u) \wedge (e \dot\lessdot e_v)\big).
  \]

  If $\varphi_G(\bvv,\bX)$ is a boolean combination of formulas of
  $\MSO(\oG(\cE))$, then $\varphi_{\dotES}(\be,\bX)$ is the same boolean
  combination of the corresponding formulas in $\MSO(\dotcE)$.

  If $\varphi_G(\bvv,\bX) = \exists v \varphi_G'(\{v\}\cup\bvv,\bX)$
  and $\varphi_{\dotES}'(\{e_v\}\cup \be,\bX)$ is the formula obtained
  from $\varphi_G'(\{v\}\cup \bvv,\bX)$, then we set
  \[
    \varphi_{\dotES}(\be,\bX) := \exists e_v \Big((\dotlambda(e_v) = h)
    \wedge \varphi_{\dotES}'\big(\{e_v\} \cup \be,\bX\big)\Big).
  \]

  If $\varphi_G(\bvv, \bX) = \exists X \varphi_G'(\bvv,\{X\}\cup \bX)$
  and $\varphi_{\dotES}'(\be,\{X\}\cup \bX)$ is the formula obtained
  from $\varphi_G'(\bvv, \{X\}\cup \bX)$, then we set
  \[
    \varphi_{\dotES}(\be,\bX) := \exists X \Big(\big(\forall e_v \in X
    (\dotlambda(e_v) = h)\big) \wedge \varphi_{\dotES}'\big(\be,\{X\}\cup
    \bX\big)\Big).
  \]

  For every sentence $\varphi_G$ in $\MSO(\oG(\cE))$, the sentence
  $\varphi_{\dotES}$ obtained by our construction is a sentence in
  $\MSO(\dotcE)$. Moreover, by induction on the structure of the
  sentence, one can show that for any trace-regular event structure
  $\cE$, $\oG(\cE)$ satisfies $\varphi_G$ if and only if $\dotcE$
  satisfies $\varphi_{\dotES}$.
\end{proof}

\begin{proposition}\label{prop-MSOG-MSO}
  For a trace-regular event structure ${\cE}=(E,\le, \#,\lambda)$, if
  $\MSO(\oG(\cE))$ is decidable, then $\MSO(\cE)$ is decidable.
\end{proposition}


\begin{proof}
  Consider any formula
  $\varphi(\be,\bX) = \varphi_{ES}(\be,\bX) \in \MSO(\cE)$, where
  $\be = \{e_1,\ldots,e_n\}$ and $\bX = \{X_1,\ldots,X_m\}$.  We first
  transform this formula into another formula of $\MSO(\cE)$ as
  follows. Since $\leq$ is the transitive closure of $\lessdot$ and
  since the transitive closure of any binary relation expressible in
  MSO can also be expressed in MSO (see~\cite[Section 5.2.2]{CouEn}),
  we can assume that the atomic formulas of $\varphi(\be,\bX)$ are of
  the type $e \in X$, $e_1 = e_2$, $R_a(e)$ for $a \in \Sigma$, and
  $e_1 \lessdot e_2$.

  We now transform the formula $\varphi(\be,\bX)$ in such a way that
  each event variable (respectively, each set variable) has a label
  $a \in \Sigma$, i.e., it can be interpreted only by an event labeled
  by $a$ (respectively, by a subset of events labeled by $a$). Assume
  in the following that $\Sigma = \{a_1,\ldots,a_n\}$.

  We transform $\varphi(\be,\bX)$ into $\varphi'(\be,\bX)$ in an
  inductive way as follows.

  If $\varphi(x) = R_a(x)$ for $a \in \Sigma$, then
  $\varphi'(x) := R_a(x)$.

  If $\varphi(x_1,x_2) = x_1 \lessdot x_2$, then
  $\varphi'(x_1,x_2) := x_1 \lessdot x_2$.

  If $\varphi(\be,\bX) = \neg \varphi_1(\be,\bX)$ and if
  $\varphi_1'(\be',\bX')$ is the formula obtained from
  $\varphi_1(\be,\bX)$, then
  \[
    \varphi'(\be',\bX') := \neg \varphi_1'(\be',\bX').
  \]

  If
  $\varphi(\be,\bX) = \varphi_1(\be_1,\bX_1) \vee
  \varphi_2(\be_2,\bX_2)$ with $\be = \be_1 \cup \be_2$ and
  $\bX = \bX_1 \cup \bX_2$ and if $\varphi_1'(\be'_1,\bX'_1)$ and
  $\varphi_2'(\be'_2,\bX'_2)$ are the formulas obtained respectively
  from $\varphi_1(\be_1,\bX_1)$ and $\varphi_2(\be_2,\bX_2)$, then
  \[
    \varphi'(\be' = \be'_1 \cup \be'_2, \bX' = \bX'_1 \cup \bX_2') :=
    \varphi_1'(\be_1',\bX_1') \vee \varphi_2'(\be_2',\bX_2').
  \]

  If $\varphi(e,X) = (e \in X)$, then let
  $\bX = \{X_a : {a \in \Sigma}\}$ and set
  \[
    \varphi'(e,\bX) := \bigwedge_{a\in \Sigma} \Big(R_a(e)
    \Leftrightarrow e \in X_a\Big).
  \]

  If $\varphi(\be,\bX) = \exists e \varphi_1(\{e\} \cup \be,\bX)$ and
  if $\varphi_1'(\{e\} \cup \be',\bX')$ is the formula obtained from
  $\varphi_1(\{e\} \cup \be,\bX)$, then
  \[
    \varphi'(\be',\bX') := \bigvee_{a\in \Sigma} \bigg(\exists e_a
    \Big(R_a(e_a) \wedge \varphi_1'\big(\{e_a\} \cup
    \be',\bX'\big)\Big)\bigg).
  \]

  If $\varphi(\be,\bX) = \exists Y \varphi_1(\be, \{Y\}\cup \bX)$ and
  if $\varphi_1'(\be', \bY \cup \bX')$ is the formula obtained from
  $\varphi_1(\be,\{Y\} \cup \bX')$ where
  $\bY = \{Y_{a_1}, \ldots, Y_{a_n}\}$, then
  \[
    \varphi'(\be',\bX') := \exists Y_{a_1} \ldots \exists Y_{a_n}
    \bigwedge_{a_i \in \Sigma} \bigg(\big(\forall e \in Y_{a_i}
    (R_{a_i}(e))\big)\bigg) \wedge \varphi_1'\big(\be', \bY \cup
    \bX'\big).
  \]

  Consequently, from now on, we consider only formulas
  $\varphi_{ES}(\be,\bX)$ in which every variable $e_a$ or $S_a$ is
  indexed by a letter $a \in \Sigma$, meaning that it can only be
  interpreted by an event (or a set of events) labeled by $a$.
  Given such a formula $\varphi_{ES}(\be,\bX) \in \MSO(\cE)$, we
  construct a formula $\varphi_{G}(\bS,\bX) \in \MSO(\oG(\cE))$ where
  each event variable $e$ is replaced by a second order variable
  representing a set of vertices $S$. The idea of the transformation
  is that an event variable $e$ can be interpreted in $\cE$ by an
  event $f$ if and only if the set $S$ can be interpreted in
  $\oG(\cE)$ by the set of sources of precisely those edges which are
  dual to the hyperplane $\cH_f$. Similarly, a set of events will be
  represented by the set of sources of the edges dual to the
  corresponding hyperplanes.  We proceed by induction on the structure
  of $\varphi_{ES}(\be,\bX) \in \MSO(\oG(\cE))$.

  We first consider atomic formulas.

  If $\varphi_{ES}(\{e_a\},\emptyset) = (R_b(e_a))$, then we set
  $\varphi_{G}(\{S_{e_a}\}) := \bot$ if $a \neq b$ and
  $\varphi_{G}(\{S_{e_a}\}) := \top $ otherwise.

  If $\varphi_{ES}(\{e_a\},\{X_b\}) = (e_a \in X_b)$, then we set
  $\varphi_{G}(\{S_{e_a}\},\{X_b\}) := \bot$ if $a \neq b$ and
  $\varphi_{G}(\{S_{e_a}\},\{X_b\}) := (S_{e_a} \subseteq X_b)$
  otherwise.

  If $\varphi_{ES}(\{e_a,e_b\},\emptyset) = (e_a = e_b)$, then we set
  $\varphi_{G}(\{S_{e_a},S_{e_b}\}) := \bot$ if $a \neq b$ and
  $\varphi_{G}(\{S_{e_a},S_{e_b}\}) := (S_{e_a} = S_{e_b})$ otherwise.

  If $\varphi_{ES}(\{e_a,e_b\},\emptyset) = (e_a \lessdot e_b)$, then
  we set
  \[
    \varphi_{G}(\{S_{e_a},S_{e_b}\}) := \Big(\exists s \in S_{e_a}
    \exists t \in S_{e_b} ((s,t) \in E_a)\Big) \wedge \neg \Big(\exists
    s \in S_{e_b} \exists t \in S_{e_a} ((s,t) \in E_b)\Big).
  \]

  If $\varphi_{ES}(\be,\bX)$ is a boolean combination of formulas of
  $\MSO(\cE)$, then $\varphi_{G}(\bS,\bX)$ is the same boolean
  combination of the corresponding formulas in $\MSO(\oG(\cE))$.

  In the following, we need to ensure that the set of vertices
  representing the edges dual to a hyperplane (or to a set of
  hyperplanes) labeled by $a$ are indeed the sources of edges labeled
  by $a$. Given a second order variable $Y_a$, this is ensured by the
  following formula $H_a(Y_a)$:
  \[
    H_a(Y_a) := \forall s \in Y_a \exists t ((s,t) \in E_a).
  \]

  We also need to ensure that the variables representing the edges
  dual to a hyperplane (or to a set of hyperplanes) labeled by $a$
  represent sets of edges that are closed in the directed median graph
  $\oG(\cE)$ under the parallelism relation. Given a second order
  variable $Y_a$, this is ensured by the following formula
  $PC_a(Y_a) \in \MSO(\oG(\cE))$:
  \begin{align*}
    PC_a(Y_a) := \bigwedge_{b \in \Sigma} \bigg(\forall s_1 \forall s_2
    \forall t_1 \forall t_2 \Big((s_1,t_1), (s_2,t_2) \in E_a\Big)
    & \wedge \Big((s_1,s_2), (t_1,t_2) \in E_b\Big)\\
    & \Rightarrow \Big((s_1 \in Y_a) \Leftrightarrow (s_2 \in Y_a)\Big)\bigg).
  \end{align*}

  We first consider second order existential quantification in
  $\MSO(\cE)$. We want to replace a variable representing a set of
  events by a variable representing the set of sources of the edges
  dual to these events.  If
  $\varphi_{ES}(\be,\bX) = \exists X_a \psi_{ES}(\be,\{X_a\}\cup \bX)$
  and $\psi_G(\bS,\{X_a\}\cup \bX)$ is the formula obtained from
  $\psi_{ES}(\be,\{X_a\}\cup \bX)$, then we set
  \[
    \varphi_{G}(\bS,\bX) := \exists X_a \Big(H_a(X_a) \wedge PC_a(X_a)
    \wedge \psi_G\big(\bS,\{X_a\}\cup \bX\big)\Big).
  \]

  We now consider first order existential quantification in
  $\MSO(\cE)$ and we use the previous transformation and the fact that
  an event is a minimal non-empty subset of events. The following
  formula ensures that $S_a$ is a non-empty set of sources of edges
  dual to a set of events labeled by $a$:
  \[
    NSH(S_a) := \Big( H_a(S_a) \wedge PC_a(S_a) \wedge (\exists s \in
    S_a)\Big).
  \]

  If
  $\varphi_{ES}(\be,\bX) = \exists e_a \psi_{ES}(\{e_a\}\cup \be,
  \bX)$ and $\psi_G(\{S_a\}\cup \bS,\bX)$ is the formula obtained from
  $\psi_{ES}(\{e_a\}\cup\be,\bX)$, then we set:
  \[
    \varphi_{G}(\bS,\bX) := \exists S_a \bigg(NSH(S_a) \wedge
    \Big(\forall S'_a \big(NSH(S'_a) \wedge (S'_a \subseteq S_a)\big)
    \Rightarrow (S_a = S'_a)\Big) \wedge \psi_G\big(\{S_a\}\cup\bS,
    \bX\big)\bigg).
  \]

  For every sentence $\varphi_{ES}$ in $\MSO(\cE)$, the sentence
  $\varphi_{G}$ obtained by our construction is a sentence in
  $\MSO(\oG(\cE))$. Moreover, by induction on the structure of the
  sentence, it can be shown that for any trace-regular event structure
  $\cE$, $\cE$ satisfies $\varphi_{ES}$ if and only if $\oG(\cE)$
  satisfies $\varphi_G$.
\end{proof}

The ``if'' implication of Theorem~\ref{barycentric-subdivision-MSO} is
the content of Proposition~\ref{prop-MSOb-MSOG}. To prove the converse
implication, consider a trace-regular event structure
${\cE}=(E,\le, \#,\lambda)$, such that $\MSO(\oG(\cE))$ is
decidable. By Theorem~\ref{mso-graph}, $G(\cE)$ has finite
treewidth. This implies that $G(\dotcE)$ has also finite treewidth.
By Theorem~\ref{mso-graph}, $\MSO(\oG(\dotcE))$ is decidable, and
thus, by Proposition~\ref{prop-MSOG-MSO}, $\MSO(\dotcE)$ is decidable.

\begin{remark}
  Notice that the converse of Proposition~\ref{prop-MSOG-MSO} is not
  true: the MSO theory of trace conflict-free event structures is
  decidable~\cite{Mad}, however the graphs of their domains may have
  infinite treewidth and thus an undecidable MSO theory. For example,
  the event structure $\cE=(E,\le, \#)$ consisting of two pairwise
  disjoint sets $X=\{x_0,x_1,x_2,\ldots\},Y=\{y_0,y_1,y_2,\ldots\}$ of
  events, such that $x_0<x_1<x_2<\cdots$ and $y_0<y_1<y_2<\cdots$, and
  all events of $X$ are concurrent with all events of $Y$, is
  conflict-free and trace-regular, but its domain $\cD(\cE)$ is the
  infinite square grid.
\end{remark}

\section{Counterexamples to Conjecture~\ref{MSO}}\label{sec-cex}

In this section, we use the general results obtained in
Section~\ref{MSO-domain} to construct a counterexample to
Thiagarajan's Conjecture~\ref{MSO}.  In view of
Theorem~\ref{barycentric-subdivision-MSO}, it suffices to find a
trace-regular event structure ${\cE}$ whose graph $G(\cE)$ has
unbounded treewidth (i.e., it contains arbitrarily large square grid
minors) and whose hairing $\dotcE$ is grid-free (as an event
structure). To build such an example, as in~\cite{CC-thiag}, we start
by constructing a finite NPC square complex. Namely, we consider an
NPC square complex $Z$ with one vertex, four edges, and three squares,
and we show that $Z$ is virtually special. This implies that the
principal filter of the universal cover $\tZ$ of $Z$ is the domain
$\cD(\cE_Z)$ of a trace-regular event structure (i.e., $\cE_Z$ is the
event structure unfolding of a net system $N_Z$). We prove that the
median graph $G(\cE_Z)$ of the domain has unbounded treewidth. On the
other hand, to prove that $\dotcE_Z$ is grid-free we show that the
graph $G(\cE_Z)$ has bounded hyperbolicity. In conclusion, we obtain
the following result:

\begin{theorem}\label{th-counterexample}
  There exists a virtually special NPC square complex $Z$ and a
  trace-regular event structure $\cE_Z$ such that the domain
  $\cD(\cE_Z)$ is a principal filter of $\tZ$, the hairing $\dotcE_Z$
  is grid-free, and the median graphs $G(\cE_Z)$ and $G(\dotcE_Z)$
  have unbounded treewidth. Consequently, $\MSO(\dotcE_Z)$ is
  undecidable and thus Thiagarajan's Conjecture~\ref{MSO} is false.
\end{theorem}

Badouel et al.~\cite[pp. 144--146]{BaDaRa} described a trace-regular
event structure that has a domain that is not context-free. Using the
results of Section~\ref{MSO-domain}, we show that the hairing of this
event structure is also a counterexample to Conjecture~\ref{MSO}.

\subsection{Proof of Theorem~\ref{th-counterexample}}

The square complex $Z$ consists of three squares $Q_1,Q_2,Q_3$, one
vertex $v_0$, and four edges, colored and directed as in
Fig.~\ref{figZ}. The four edges of $Z$ are colored orange (color $a$),
black (color $b$), blue (color $x$), and red (color $y$) as indicated
in the figure.  Since $Z$ is a VH-complex, $Z$ is nonpositively
curved. Let $\tbZ=(\tZ,\tildo,\tc)$ denote the directed and colored
universal cover of $Z$. Pick any vertex $\tv_0$ of $\tbZ$ ($\tv_0$ is
a lift of $v_0$) and let $\cE_Z$ denote the event structure whose
domain is the principal filter
$\cD_Z=(\cF_{\tildo}(\tv_0,\tZ^{(1)}),\prec_{\tildo})$ of
$(\tZ,\tildo)$. Let also $\oG(\cE_Z)$ and $G(\cE_Z)$ denote the
directed and the undirected 1-skeletons of $\cD_Z$. Finally, denote by
$\dotcE_Z$ the hairing of $\cE_Z$.

\begin{figure}
  \includegraphics[scale=0.85]{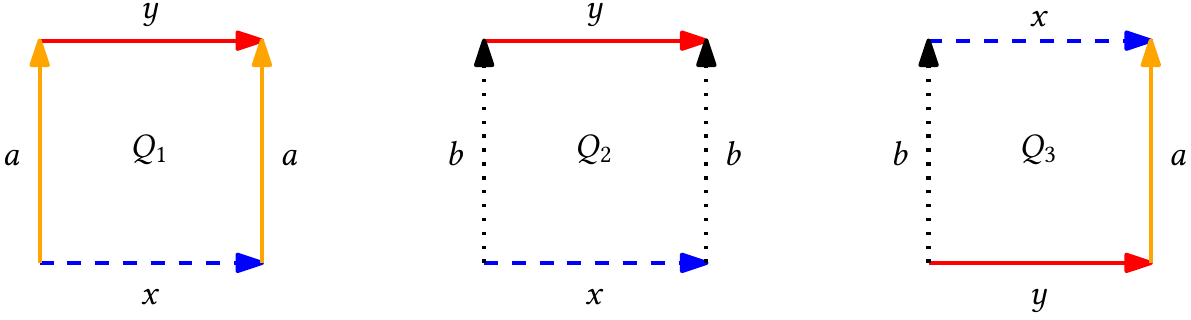}%
  \caption{The three squares defining the VH-complex $Z$}%
  \label{figZ}
\end{figure}

First we investigate the properties of the complexes $Z$ and $\tbZ$,
of the graphs $\oG_Z$ and $G_Z$, and of the event structure
$\cE_Z$. First, even if $Z$ is not special, we show that it is
virtually special:

\begin{lemma}\label{virtually-special}
  The NPC square complex $Z$ is virtually special.  Consequently, the
  event structures $\cE_Z$ and $\dotcE_Z$ are trace-regular.
\end{lemma}

\begin{proof}
  Let $Z'$ be the square complex represented in Fig.~\ref{fig-scZ}.
  As in Fig.~\ref{fig-petrinet-XN}, one has to merge the left and
  right sides, as well as the lower and the upper sides. Consider the
  map $\varphi$ sending all vertices of $Z'$ to the unique vertex of
  $Z$, and each edge of $Z'$ to the unique edge of $Z$ with the same
  color.

  The complex $Z'$ has 8 vertices, 32 edges, and 24 squares. In $Z'$,
  a 4-cycle is the boundary of a square if opposite edges have the
  same label and if the colors of the boundary of this square
  correspond to the colors of the boundary of one of the three squares
  of $Z$. In Fig.~\ref{fig-scZ}, the number ($2$ or $4$) in the middle
  of each $4$-cycle represent the number of squares of $Z'$ on the
  vertices of this $4$-cycle.  This implies that $\varphi$ is a
  covering map from $Z'$ to $Z$.
  Observe that two edges are dual to the same hyperplane of $Z'$ iff
  they have the same label. Using this, it is easy to check that $Z'$
  is special.

  By Theorem~\ref{special-trace-bis} and
  Proposition~\ref{special-trace1}, we have that for any vertex
  $\tv \in \tZ'$, $\cF(\tv,\tZ')$ is the domain of a trace-regular
  event structure $\cE_{Z'}$. Since $\tZ$ and $\tZ'$ coincide and
  since all vertices of $\tZ = \tZ'$ are lifts of the unique vertex of
  $Z$, $\cF(\tv,\tZ')$ is independent of the choice of
  $\tv$. Consequently, $\cE_Z = \cE_{Z'}$ is a trace-regular event
  structure.  The fact that $\dotcE_Z$ is trace-regular follows from
  Proposition~\ref{barycentric-subdivision}.
\end{proof}

\begin{figure}
  \includegraphics[scale=0.85]{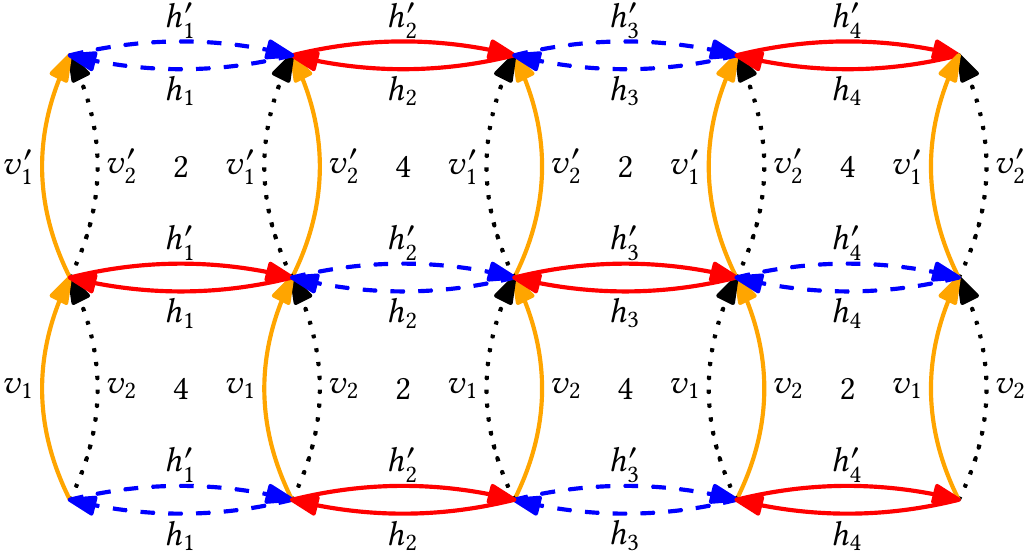}%
  \caption{A finite special cover $Z'$ of the complex $Z$.}%
  \label{fig-scZ}
\end{figure}

\begin{remark}
  Observe that $Z'$ coincides with the special cube complex $X_{N^*}$
  of the net system $N^*$ from Examples~\ref{ex-petrinet}
  and~\ref{ex-petrinet-XN}. Consequently, $\cE_Z$ coincides with the
  event structure unfolding $\cE_{N^*}$ of $N^*$.
  To obtain a net system $\dotN^*$ corresponding to $\dotcE_Z$, one
  can use Proposition~\ref{prop-hairing-N}. However in this case, it
  is enough to add a single transition $h$ to $N^*$ such that
  $\buh = \{C_1,C_2,C_3,C_4\}$ and $\hbu = \varnothing$. In the
  resulting $\dotN^*$, for any firing sequence $\sigma$ of $N^*$,
  $\sigma h$ is a firing sequence of $\dotN^*$ and no transition can
  be fired once $h$ has been fired. Using this property, with a proof
  similar to the proof of Proposition~\ref{prop-hairing-N}, one can
  show that $\dotcE_{Z}$ and $\cE_{\dotN^*}$ are isomorphic.
\end{remark}

\begin{lemma}\label{GZ-hyperbolic}
  The graph $G(\cE_Z)$ is hyperbolic and thus the event structures
  $\cE_Z$ and $\dotcE_Z$ are grid-free.
\end{lemma}

\begin{proof}
  Since $Z$ is a VH-complex, its universal cover is a CAT(0) square
  complex. Thus any isometric grid of $G(\cE_Z)$ is flat.  Suppose by
  way of contradiction that $G(\cE_Z)$ contains a large $n\times n$
  flat grid $\Lambda$.  Since $\Lambda$ is flat, $\Lambda$ is a convex
  and thus a gated subgraph of $G(\cE)$. Let $\tv$ denotes the gate of
  $\tv_0$ in $\Lambda$. By Lemma~\ref{directed-median} the direction
  of the edges of the graph $\oG(\cE_Z)$ coincide with the basepoint
  order $\le_{\tv_0}$. This implies that the direction of the edges of
  the grid $\Lambda$ in $\oG(\cE_Z)$ coincides with the basepoint
  order of $\Lambda$ with $\tv$ as the basepoint. In particular, this
  implies that $\tv$ is the unique source of $\Lambda$. Consequently,
  there exists a corner $\tv'$ of $\Lambda$ such that the interval
  $I(\tv,\tv')$ spans an $n'\times n''$ subgrid $\Lambda'$ of $\Lambda$
  with $n'\ge \frac{n}{2}$ and $n''\ge \frac{n}{2}$.

  Let $\ovr{\tv\tu}$ and $\ovr{\tv\tu'}$ be the
  two outgoing from $\tv$ edges in $\Lambda'$.  Consider the square
  $Q$ of $\Lambda'$ containing those two edges. Suppose without loss
  of generality that $\ovr{\tv\tu}$ is upward vertical and
  $\ovr{\tv\tu'}$ is horizontal and to the right. The
  vertex $\tv$ is the unique source of $Q$. Denote by $\tw$ the vertex
  of $Q$ opposite to $\tv$. We will analyze in which way one can now
  extend the square $Q$ to the grid $\Lambda'$. Notice that the square
  $Q$ as well as any other square of $\Lambda$ is one of the three
  squares $Q_1,Q_2,Q_3$ of the complex $Z$.

  First suppose that $Q = Q_1$. Since $\Lambda'$ is directed according
  to $\le_{\tv}$, one can extend $Q$ horizontally only by adding a new
  square $Q_1$ to the right. Also we can extend $Q$ vertically only by
  adding the square $Q_3$ on the top of $Q$. But then we cannot extend
  the resulting union of three squares to a $2\times 2$ grid because
  there is no square in $\{Q_1,Q_2,Q_3\}$ where the outgoing edges of
  the source are orange (color $a$) and red (color $y$).
  Now suppose that $Q=Q_3$. Then we can extend $Q$ only by setting
  $Q_1$ to the right. From the case when $Q=Q_1$ we know that we
  cannot extend $Q_1$ to a $2\times 2$ grid. This show that we cannot
  extend $Q$ to the $2\times 3$ grid.
  Finally, suppose that $Q=Q_2$. The single possibility to extend $Q$
  vertically is to set a copy of $Q_3$ on top. From the case when
  $Q=Q_3$, we know that we cannot extend $Q_3$ to a $2\times 3$
  grid. This show that we cannot extend $Q$ to the $3\times 3$ grid.

  We deduce that in all cases we have $n'\le 2$ or $n''\le 2$,
  establishing that if $G(\cE_Z)$ contains an $n\times n$ isometric
  grid, then $n\le 4$.  This proves that $G(\cE_Z)$ is hyperbolic. The
  graph $G(\dotcE_Z)$ is also hyperbolic because any grid of $G(\dotcE_Z)$
  comes from a grid of $G(\cE_Z)$. By
  Proposition~\ref{hyperbolic-gridfree}, the event structures $\cE_Z$
  and $\dotcE_Z$ are thus grid-free, since $\cE_Z$ and $\dotcE_Z$ have
  respectively degrees $4$ and $5$.
\end{proof}

\begin{remark}
  Since $\dotcE_Z = \cE_{\dotN^*}$, we can also establish that
  $\dotcE_Z$ is grid-free by considering the net system $\dotN^*$ (as
  suggested by one of the referees of this paper). Using the
  symmetries of $\dotN^*$ and some case analysis, one can show that
  there exist no reachable marking $m$ of $\dotN^*$ and firing
  sequences $\sigma, \sigma'$ such that
  $m \xlongrightarrow{\sigma} m$, $m \xlongrightarrow{\sigma'} m$, and
  $(a,a') \in I$ for any transition $a$ and $a'$ appearing
  respectively in $\sigma$ and
  $\sigma'$. By~\cite[Corollary~5]{ThiaYa}, we conclude that the net
  system $\dotN^*$ is grid-free.
\end{remark}

\begin{lemma}\label{GZ-treewidth}
  The graph $G(\cE_Z)$ has infinite treewidth, i.e., the directed
  graph $\oG(\cE_Z)$ is not context-free. Consequently, the theories
  $\MSO(\oG(\cE_Z))$, $\MSO_2(G(\cE_Z))$, and $\MSO(\dotcE_Z)$ are
  undecidable.
\end{lemma}

\begin{proof}
  The proof of this assertion in some sense is similar to the proof of
  implication (4)$\Rightarrow$(5) of Theorem~\ref{mso-graph}. As in
  the proof of the implication we will show that the graph $G(\cE_Z)$
  has the infinite half-grid $\frac{1}{2}\Gamma$ as a minor. We will
  also denote by $z_{i,j}, i,j\ge 0$, the vertices of
  $\frac{1}{2}\Gamma_n$ and by $Z_{i,j},i,j\ge 0$, the connected
  subgraph of $G(\cE_Z)$ which will be mapped (contracted) to
  $z_{i,j}$. The subgraphs $Z_{i,j}$ are also paths laying in two
  consecutive spheres $S(\tv_0,k-1)\cup S(\tv_0,k)$. The difference is
  that in the proof of implication (4)$\Rightarrow$(5) of
  Theorem~\ref{mso-graph} we first constructed the union $F^*$ of all
  fences in a downward way and then constructed the paths
  $Z_{i,j}\subset F^*$ in a upward way. For the current claim, we will
  build the paths $Z_{i,j}$ level-by-level, in an upward manner.

  For this we use the fact that $\oG(\cE_Z)$ is the graph of the
  principal filter
  $\cD_Z=(\cF_{\tildo}(\tv_0,\tZ^{(1)}),\prec_{\tildo})$ of the
  universal cover $(\tZ,\tildo)$ of $Z$ (here $\tv_0$ is an arbitrary
  but fixed lift of $v_0$). Since $Z$ has one vertex $v_0$, all
  vertices $\tv$ of $\oG(\cE_Z)$ are lifts of $v_0$. Analogously to
  $v_0$, each such vertex $\tv$ is incident to four outgoing and to
  four incoming colored edges in $\tZ$. However, in the graph
  $\oG(\cE_Z)$ of the domain, each vertex $\tv$ has at most two
  incoming edges (otherwise, there exists a 3-cube in the interval
  $I(\tv_0,\tv)$, but this is impossible since $\tZ$ is
  $2$-dimensional).  The four outgoing edges define three squares
  $Q_1,Q_2,Q_3$ having $\tv$ as the source (for an illustration, see
  the last figure in Fig.~\ref{fig-claim-pavage}). Moreover,
  $(\tZ,\tildo)$ satisfies the following determinism property: if two
  edges $\ovr{e'},\ovr{e''}$ outgoing from a
  vertex $\tv$ of $\tZ$ have the same color as the edges outgoing from
  the source of a square $Q_i$ of $Z$, then $\ovr{e'}$ and
  $\ovr{e''}$ belong in $(\tZ,\tildo)$ to a
  $Q_i$-square. Using this fact, one can see that there exists an
  infinite directed path $P_a$ with $\tv_0$ as the origin and in which
  all edges have color orange (color $a$). Analogously, there exists
  an infinite directed path $P_y$ with $\tv_0$ as the origin and in
  which all edges have color red (color $y$). Since $Z$ is a
  VH-complex, the paths $P_a$ and $P_y$ are locally-convex paths of
  $\tZ$. Since $G(\cE)$ is median, by Lemma~\ref{convex}, $P_a$ and
  $P_y$ are convex paths, thus shortest paths, of $G(\cE)$.  Let
  $P_a=(\tu_0=\tv_0,\tu_1,\tu_2,\ldots)$ and
  $P_y=(\tv_0,\tv_1,\tv_2\ldots)$ (recall again that all vertices of
  these paths as well as all vertices of $G(\cE)$ are lifts of $v_0$).

  We continue with an auxiliary claim that we need in order to define
  the paths $Z_{i,j}$:
  \begin{claim}\label{claim-loc-pavage}
    For any vertex $\tv \in S(\tv_0,k-1)$, for any outgoing edges
    $\ovr{\tv\tu}, \ovr{\tv\tu'}$, there exist
    $0<p \leq 4$ distinct vertices
    $\tu_1=\tu, \tu_2, \ldots, \tu_p=\tu' \in S(\tv_0,k)$ and $p-1$
    distinct vertices $\tw_1, \ldots, \tw_{p-1} \in S(\tv_0,k+1)$ such
    that for every $i$, $\ovr{\tu_i\tw_{i}}$ and
    $\ovr{\tu_i\tw_{i-1}}$ are directed edges of
    $\oG(\cE_Z)$, and such that the following holds:
    \begin{itemize}
    \item if $\ovr{\tv\tu}$ and $\ovr{\tv\tu'}$ are colored with the
      same colors as the outgoing edges from the source of a square
      $Q \in \{Q_1,Q_2,Q_3\}$, then $p=2$ and $\ovr{\tu\tw_1}$ and
      $\ovr{\tu'\tw_1}$ are colored as the corresponding edges of $Q$.
    \item if $\ovr{\tv\tu}$ and $\ovr{\tv\tu'}$
      are colored respectively blue (color $x$) and red (color $y$),
      then $p =3$ and $\ovr{\tu\tw_1}$,
      $\ovr{\tu'\tw_2}$, $\ovr{\tv\tu_2}$ are
      colored respectively orange (color $a$), black (color $b$), and
      black (color $b$);
    \item if $\ovr{\tv\tu}$ and $\ovr{\tv\tu'}$
      are colored respectively black (color $b$) and orange (color
      $a$), then $p =3$ and $\ovr{\tu\tw_1}$,
      $\ovr{\tu'\tw_2}$, $\ovr{\tv\tu_2}$ are
      colored respectively red (color $y$), red (color $y$), and blue
      (color $x$);
    \item if $\ovr{\tv\tu}$ and $\ovr{\tv\tu'}$
      are colored respectively orange (color $a$) and red (color $y$),
      then $p =4$ and $\ovr{\tu\tw_1}$,
      $\ovr{\tu_2\tw_1}$, $\ovr{\tu_3\tw_3}$,
      $\ovr{\tu'\tw_3}$, $\ovr{\tv\tu_2}$,
      $\ovr{\tv\tu_3}$ are colored respectively red (color
      $y$), orange (color $a$), blue (color $x$), orange (color $a$),
      blue (color $x$), and black (color $b$).
    \end{itemize}
  \end{claim}

  \begin{proof}
    In the universal cover $(\tZ,\tildo)$ of $Z$, all vertices of
    $\oG(\cE_Z)$ are lifts of the unique vertex $v_0$ of
    $Z$. Consequently, each of them has four outgoing edges colored
    with the four different colors. Taking this into account, the
    proof of the claim follows from Fig.~\ref{fig-claim-pavage}.
  \end{proof}

  \begin{figure}
    \centering
    \includegraphics[scale=0.5]{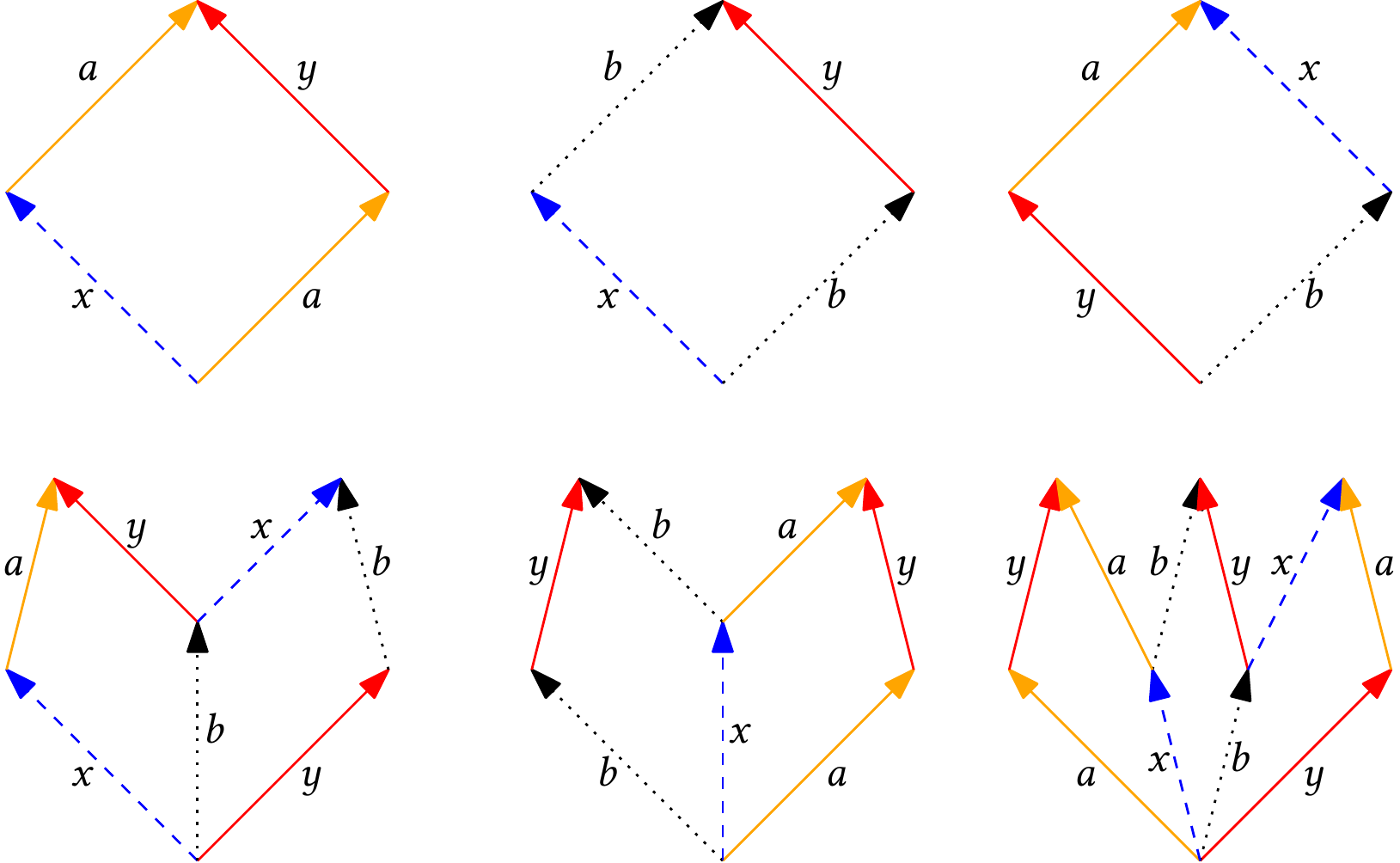}%
    \caption{To the proof of Claim~\ref{claim-loc-pavage}}%
    \label{fig-claim-pavage}
  \end{figure}

  For each $k$, we construct iteratively a simple path
  $P_k=P(\tu_k,\tv_k)=(\tu_k=\tp_{k,1}, \tq_{k,1}, \ldots,
  \tp_{k,\ell-1},$ $\tq_{k,\ell-1}, \tp_{k,\ell}=\tv_k)$ such that
  $\ovr{\tq_{k,1}\tp_{k,1}}$ is colored orange (color $a$),
  $\ovr{\tq_{k,\ell-1}\tp_{k,\ell}}$ is colored red (color $y$), and
  for each $i$, $\tp_{k,i} \in S(\tv_0,k)$ and
  $\tq_{k,i} \in S(\tv_0,k-1)$. This path plays a role similar to the
  one of the path $P_2(u_k,v_k)$ in the proof of
  Theorem~\ref{mso-graph}.
  Let $P_1 = (\tu_1,\tv_0,\tv_1)$ and suppose that the simple path
  $P_k = P(\tu_k,\tv_k)$ has been defined. We define the path
  $P_{k+1} = P(\tu_{k+1},\tv_{k+1})$ in two steps.  First, let
  $P'_{k+1}$ be the path obtained by concatenating the paths obtained
  by applying Claim~\ref{claim-loc-pavage} to each vertex $q_{k,i}$ of
  $P_k\cap S(\tv_0,k)$ and its two outgoing edges in $P_k$. Note that
  the first edges of $P_k$ and $P'_{k+1}$ are consecutive edges in a
  square $Q$ of $\oG(\cE_Z)$. Since the first edge of $P_k$ is orange
  (color $a$), necessarily $Q = Q_1$ and the first edge of $P'_{k+1}$
  is red (color $y$). Analogously, the last edges of $P_k$ and
  $P'_{k+1}$ are consecutive edges in a square $Q'$ of
  $\oG(\cE_Z)$. Since the last edge of $P_k$ is red (color $y$), then
  necessarily $Q' = Q_3$ and the last edge of $P'_{k+1}$ is orange
  (color $a$).

  \begin{claim}\label{claim-simple-path}
    $P_{k+1}'$ is a simple path.
  \end{claim}

  \begin{proof}
    Let
    $P'_{k+1} = (\tu_{k+1} = \tq_{k+1,1}, \tp_{k+1,1}, \ldots,
    \tq_{k+1,\ell'-1}, \tp_{k+1,\ell'-1}, \tq_{k+1,\ell'} =
    \tv_{k+1})$. Suppose first that there exists $i < j$ such that
    $\tq_{k+1,i}=\tq_{k+1,j}$. By convexity of $P_a$ and $P_y$, we
    have $2 \leq i < j \leq \ell'-1$. Since the path $P_k$ is simple,
    $\tq_{k+1,i}$ and $\tq_{k+1,j}$ cannot both belong to $P_k$. First
    suppose that one of them belongs to $P_k$, say $\tq_{k+1,i}$. By
    construction, $\tq_{k+1,j}$ has a neighbor $\tq_{k,j'} \in P_k$
    and $\tq_{k+1,i}$ has two distinct neighbors
    $\tq_{k,i'}, \tq_{k,i'+1} \in P_k$.  Since
    $\tq_{k+1,j} = \tq_{k+1,i}$ has at most two incoming edges in
    $\oG(\cE_Z)$, we get that $\tq_{k,j'} = \tq_{k,i'}$ or
    $\tq_{k,j'} = \tq_{k,i'+1}$. Since the path $P_k$ is simple, it
    means that $i' \leq j' \leq i'+1$, but this is impossible by the
    construction of $P_{k+1}'$ and Claim~\ref{claim-loc-pavage}.  Now
    suppose that both vertices $\tq_{k+1,i}$ and $\tq_{k+1,j}$ do not
    belong to $P_k$. Then, by construction, $\tq_{k+1,i}$ has a
    neighbor $\tq_{k,i'} \in P_k$ and $\tq_{k+1,j}$ has a neighbor
    $\tq_{k,j'} \in P_k$. Moreover, by Claim~\ref{claim-loc-pavage},
    the arcs $\ovr{\tq_{k,i'}\tq_{k+1,i}}$ and
    $\ovr{\tq_{k,j'}\tq_{k+1,j}}$ are blue (color $x$) or
    black (color $b$). Since the path $P_k$ is simple, these edges are
    distinct and thus have distinct colors. By the quadrangle
    condition, these two edges are incident to the sink
    $\tq_{k+1,i} = \tq_{k+1,j}$ of a square $Q$. However, there is no
    square in $Z$ (or in $\tZ$) where the sink is incident to a black
    and a blue edge (See Fig.~\ref{figZ}).
    Assume now that there exist $i < j$ such that
    $\tp_{k+1,i}=\tp_{k+1,j}$. By the construction, $\tp_{k+1,i}$ is
    adjacent to $\tq_{k+1,i},\tq_{k+1,i+1}$ and $\tp_{k+1,j}$ is
    adjacent to $\tq_{k+1,j},\tq_{k+1,j+1}$. By the previous case,
    these four vertices are distinct. Consequently,
    $\tp_{k+1,i} = \tp_{k+1,j}$ has four incoming edges, which is
    impossible since $\tZ$ is 2-dimensional.
  \end{proof}

  The path $P_{k+1} = P(\tu_{k+1},\tv_{k+1})$ is obtained from
  $P'_{k+1}$ by concatenating the orange (color $a$) edge
  $\ovr{\tu_k\tu_{k+1}}$ at the beginning of $P_{k+1}'$ and
  the red (color $y$) edge $\ovr{\tv_k\tv_{k+1}}$ at the
  end of $P_{k+1}'$. Since $\tu_{k+1} \in P_a$ and $P_a$ is a convex
  path, $\tu_{k+1}$ cannot coincide with any vertex of $P_{k+1}'$. For
  the same reason, $\tv_{k+1}$ is different from any vertex of
  $P_{k+1}'$ and different from $\tu_{k+1}$. Consequently, the path
  $P_{k+1}$ is  simple.

  \begin{figure}
    \centering
    \includegraphics[scale=0.65]{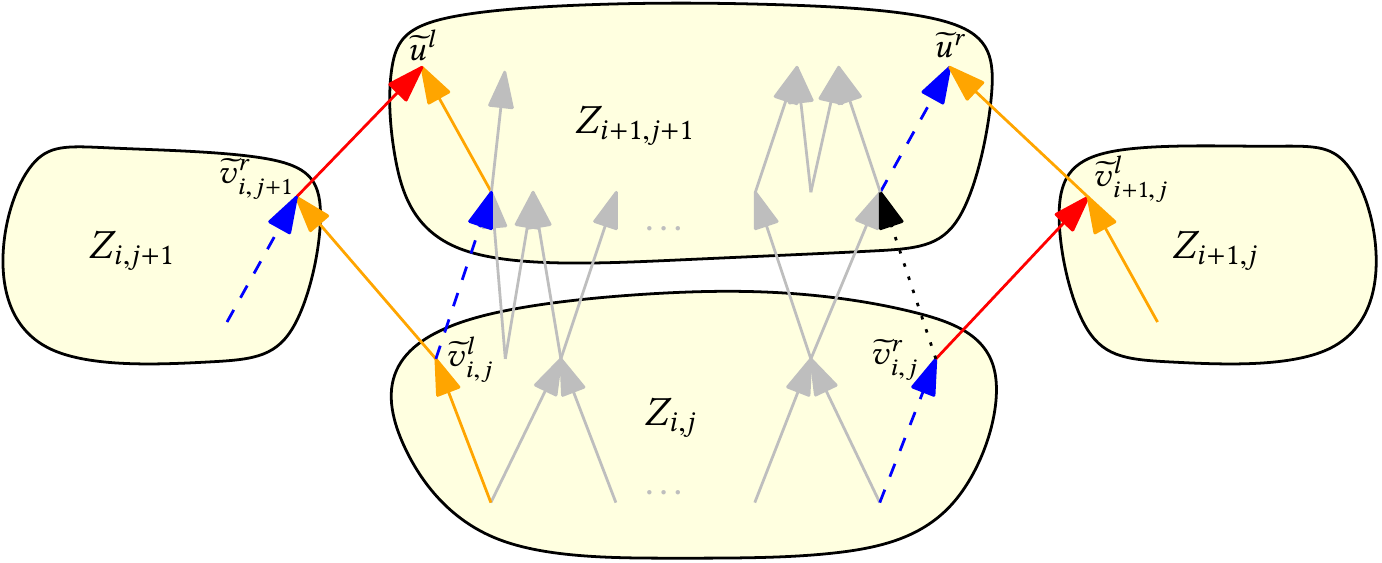}%
    \caption{Construction of the path $Z_{i+1,j+1}$.}%
    \label{figtZ}
  \end{figure}

  Now, for each $k$, we construct iteratively the paths $Z_{i,j}$ with
  $i+j = k$ by selecting subpaths of $P(\tu_k,\tv_k)$. We require that
  the paths $Z_{i,j}$ satisfy the following properties (See
  Fig.~\ref{figtZ} and~\ref{figtZ-minor}):
  \begin{enumerate}
  \item $Z_{0,j}=\{ \tu_j\}$ and $Z_{i,0}=\{ \tv_i\}$ for each
    $i, j \in \{0,\ldots,n\}$;
  \item for each $k$, if $i+j=k$, then $Z_{i,j}$ is a subpath of
    $P_k$;
  \item for each $i, j$ with $i+j= k-1$, the last vertex
    $\tv_{i,j+1}^r$ of the path $Z_{i,j+1}$ appears in $P_k$ before
    the first vertex $\tv_{i+1,j}^l$ of $Z_{i+1,j}$;
  \item each $Z_{i,j}$ with $i, j\geq 1$ has its two end-vertices in
    $S(\tv_0,k)$ and its first edge is orange (color $a$) and its last
    edge is blue (color $x$);
  \item for each pair $(i,j)$ with $i+j=k$, the leftmost vertex
    $\tv^l_{i,j}$ of the path $Z_{i,j}$ is adjacent to the rightmost
    vertex $\tv^r_{i,j+1}$ of the path $Z_{i,j+1}$ by an orange (color
    $a$) edge belonging to $P_{k+1}$ and the rightmost vertex
    $\tv^r_{i,j}$ of $Z_{i,j}$ is adjacent to the leftmost vertex
    $\tv^l_{i+1,j}$ of the path $Z_{i+1,j}$ by a red (color $y$) edge
    belonging to $P_{k+1}$;
  \item any two distinct paths $Z_{i,j}$ and $Z_{i',j'}$ are disjoint.
  \end{enumerate}

  Recall that the path $P'_2$ is obtained by applying
  Claim~\ref{claim-loc-pavage} to $P_1 = (\tu_1,\tv_0,\tv_1)$ and that
  $P_2'$ is a $(\tu_1,\tv_1)$-path starting with a red edge (color
  $y$) and ending by an orange edge (color $a$). Let $Z_{1,1}$ be the
  path obtained from $P'_2$ by removing these two edges. It is easy to
  see that Conditions (1)-(3) and (5)-(6) hold, and Condition (4)
  holds by the construction of $P'_2$ and
  Claim~\ref{claim-loc-pavage}.

  Suppose now that the paths $Z_{i,j}$ satisfying the previous
  conditions have been defined for all pairs $(i,j)$ such that
  $i+j\le k+1$ and we have to define the paths $Z_{i,j}$ with
  $i+j=k+2$ (See Fig.~\ref{figtZ} for an illustration of the
  construction described below).
  By induction hypothesis, the edge $\tv^l_{i,j}\tv^r_{i,j+1}$ is
  orange (color $a$) and the edge $\tv^r_{i,j}\tv^l_{i+1,j}$ is red
  (color $y$). These two edges belong to $P_{k+1}$. There exists
  $\tu^l$ such that $\tv^l_{i,j}\tv^r_{i,j+1}$ and
  $\tv^r_{i,j+1}\tu^l$ are consecutive in a square $Q^l$ of
  $\tZ$. Consequently, $\tv^r_{i,j+1}\tu^l$ is red (color $y$) and the
  opposite edge of $\tv^l_{i,j}\tv^r_{i,j+1}$ in $Q^l$ is orange
  (color $a$). Note that by construction, this edge also belongs to
  $P_{k+1}$.  Analogously, there exists $\tu^r$ such that
  $\tv^r_{i,j}\tv^l_{i+1,j}$ and $\tv^l_{i+1,j}\tu^r$ are consecutive
  in a square $Q^r$ of $\tZ$. Consequently, $\tv^l_{i+1,j}\tu^r$ is
  orange (color $a$) and the opposite edge of
  $\tv^r_{i,j}\tv^l_{i+1,j}$ in $Q^r$ is blue (color $x$). Note that
  by construction, this edge also belongs to $P_{k+1}$.  We let
  $Z_{i+1,j+1}$ be the subpath of $P_{k+2}$ comprised between $\tu^l$
  and $\tu^r$. By the construction of $P_{k+2}$ and the properties of
  $Q^l$ and $Q^r$, the first edge of $Z_{i+1,j+1}$ is orange and the
  last one is blue, i.e., $Z_{i+1,j+1}$ satisfies Conditions (2) and
  (4). Observe that $Z_{i,j+1}$ and $Z_{i+1,j}$ also satisfy Condition
  (5).

  We continue with Condition (3). Consider any path $Z_{i',j'}$ with
  $i'+j' = i+j+2 = k+2$ and assume that $i+1 < i'$. From the
  construction, $Z_{i+1,j+1}$ is included in the
  $(\tu_{k+2},\tv^l_{i+1,j})$-subpath of $P_{k+2}$ while $Z_{i',j'}$
  is included in the $(\tv^r_{i+1,j},\tv_{k+2})$-subpath of
  $P_{k+2}$. Since these two subpaths are disjoint by
  Claim~\ref{claim-simple-path}, $Z_{i+1,j+1}$ is disjoint from
  $Z_{i',j'}$, establishing Condition (3).

  It remains to show that Condition (6) holds, i.e., that the path
  $Z_{i+1,j+1}$ is disjoint from any other path $Z_{i',j'}$ with
  $i'+j' \leq i+j+2 = k+2$. If $i'+j'=i+j= k+2$, this follows from
  Condition (3).  If $i'+j' \leq k$, then this is trivially true since
  $Z_{i',j'} \subseteq P_k \subseteq S(\tv_0,k-1)\cup S(\tv_0,k)$ and
  $Z_{i+1,j+1} \subseteq P_{k+2} \subseteq S(\tv_0,k+1)\cup
  S(\tv_0,k+2)$. Assume finally that $i'+j'= k+1$. In this case if
  $Z_{i',j'}$ and $Z_{i+1,j+1}$ have a common vertex $\tw$, then
  $\tw \in S(\tv_0,k+1)$.  If $i' \leq i$, by induction hypothesis,
  $Z_{i',j'}$ is included in the $(\tu_{k+1},\tv^r_{i,j+1})$-subpath
  of $P_{k+1}$. Note that by construction all the vertices of
  $P_{k+1} \cap S(\tv_0,k+1)$ that appear before $\tv^r_{i,j+1}$ in
  $P_{k+1}$ also appear before $\tv^r_{i,j+1}$ in $P_{k+2}$. Since all
  vertices of $Z_{i+1,j+1}$ appear after $\tv^r_{i,j+1}$ in $P_{k+2}$
  and since $P_{k+2}$ is simple, necessarily $Z_{i+1,j+1}$ is disjoint
  from $Z_{i',j'}$. If $i' \geq i+2$, we obtain the same result by a
  symmetric argument.
\end{proof}

\begin{figure}
  \centering
  \includegraphics[scale=0.55]{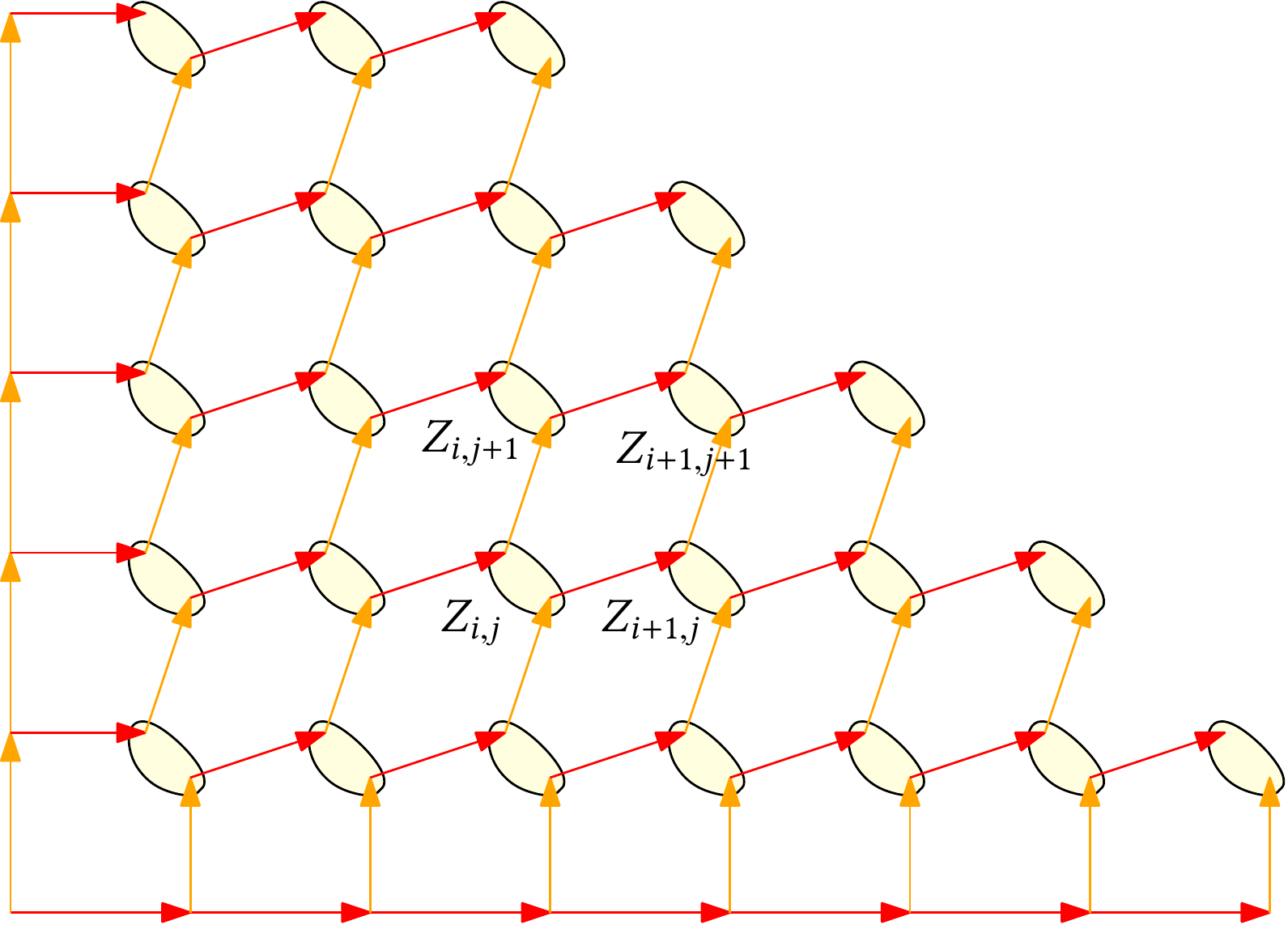}%
  \caption{Part of the half-grid resulting from the contraction of the paths
    $Z_{i,j}$.}%
  \label{figtZ-minor}
\end{figure}

Consequently, by Lemma~\ref{GZ-hyperbolic}, the event structure
$\dotcE_Z$ is grid-free and by Lemma~\ref{GZ-treewidth}, $\MSO(\dotcE_Z)$
is undecidable. This concludes the proof of
Theorem~\ref{th-counterexample}.

\begin{remark}\label{rem-tore-bazar}
  By construction, the event structure $\dotcE_Z$ is strongly regular,
  and $\dotcE_Z$ is hyperbolic by Lemma~\ref{GZ-hyperbolic}. However,
  $\dotcE_Z$ is not strongly hyperbolic-regular because $\tZ$ (and
  thus $\dot\tZ$) is not hyperbolic. Indeed, in $\tZ$, it is possible
  to build an infinite grid by repeating the pattern described in
  Fig.~\ref{fig-grille-pas-dirigee}. Due to the orientation of the
  edges of this grid, this grid cannot appear in any principal filter
  of $(\tZ,\tildo)$. Consequently, $\tZ$ is not hyperbolic, but any
  principal filter of $\tZ$ is hyperbolic.

  This leads to the following open question: \emph{Can one construct a
    finite directed special complex $X$ such that $\tX$ is hyperbolic
    and some principal filter of $\tX$ is not context-free?}
\end{remark}

\begin{figure}
  \centering
  \includegraphics[scale=0.7]{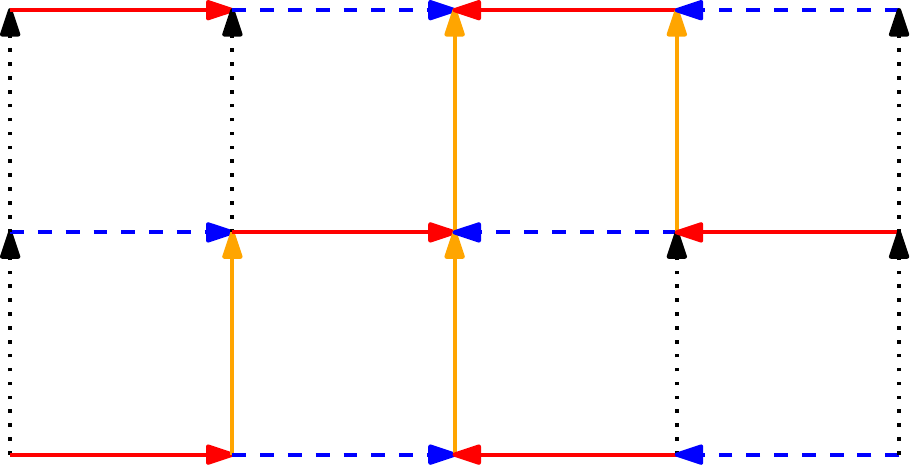}%
  \caption{Part of a infinite grid in $\tZ$}%
  \label{fig-grille-pas-dirigee}
\end{figure}

\subsection{Another Counterexample to Conjecture~\ref{MSO}}

Another counterexample to Conjecture~\ref{MSO} can be derived from the
hairing $\dotcE_{BDR}$ of the trace-regular event structure
$\cE_{BDR}$ described by Badouel et al.~\cite[pp. 144--146 and
Fig. 5--9]{BaDaRa}. The domain $\cD(\cE_{BDR})$ of $\cE_{BDR}$ is a
plane graph defined recursively as a tiling of the quarterplane with
origin $v_0$ by tiles consisting of two squares sharing an edge (see
Fig.~\ref{fig-tuilesBDR}, left). Namely, we start with two infinite
directed paths with common origin $v_0$, and at each step, we insert
the tile in each free angle (see Fig.~\ref{fig-tuilesBDR}, right). As
observed in~\cite{BaDaRa}, the hyperplanes of $G(\cE_{BDR})$ can be
represented by an arrangement of axis-parallel pseudolines in the
plane (see Fig.~\ref{fig-hyperplane-BDR}).

\begin{figure}
  \centering
  \includegraphics[scale=0.7]{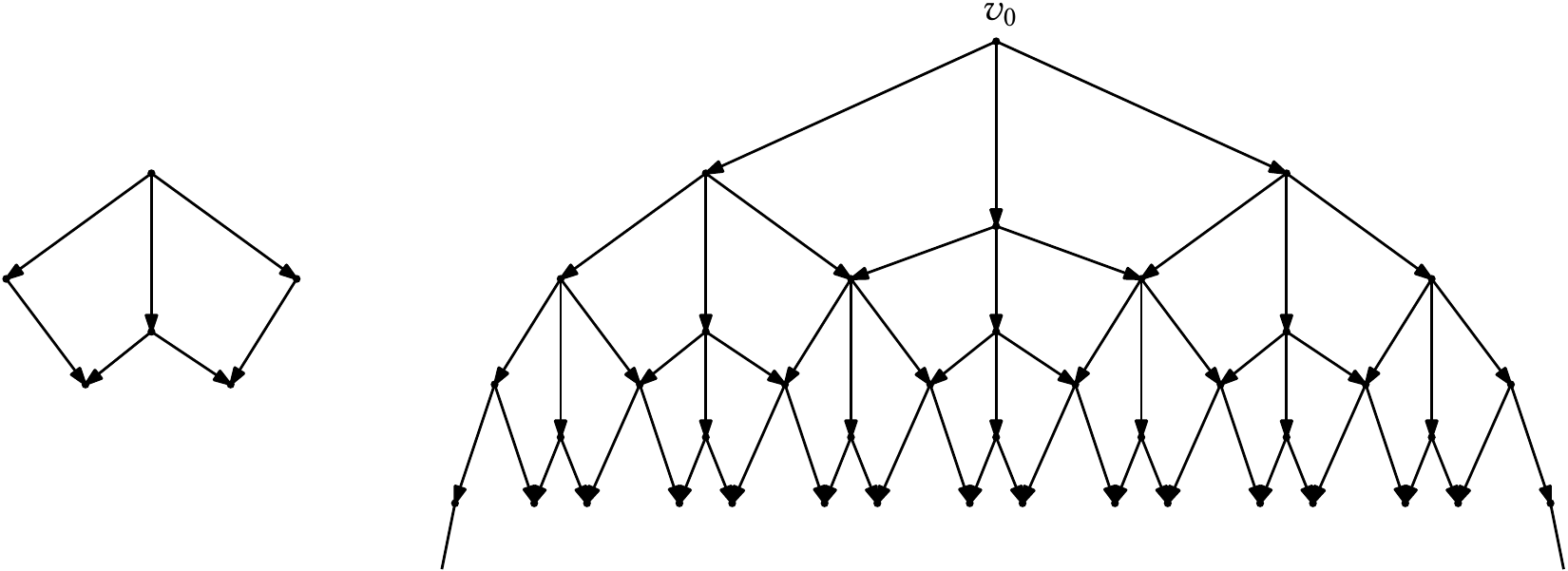}%
  \caption{The tile that is recursively inserted to build
    $\cD(\cE_{BDR})$ and the first four steps of the construction.}%
  \label{fig-tuilesBDR}
\end{figure}

\begin{figure}
  \centering
  \includegraphics[scale=0.7]{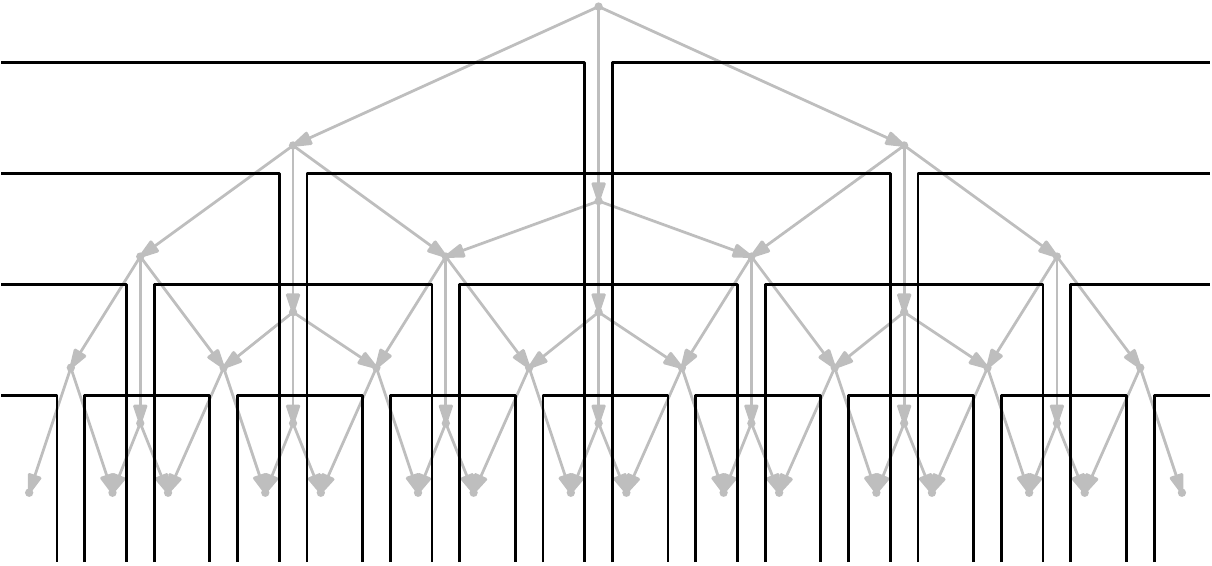}%
  \caption{The hyperplanes obtained during the first four steps of the
    construction of the domain of $\cE_{BDR}$.}%
  \label{fig-hyperplane-BDR}
\end{figure}

Badouel et al.~\cite{BaDaRa} showed that the directed graph
$\oG(\cE_{BDR})$ is not context-free. Indeed, for each $k$, there is a
unique level $k$ cluster coinciding with the sphere $S(v_0,k)$ of
radius $k$ and the diameters of spheres increase together with their
radius. By Theorem~\ref{mso-graph}, this shows that the graph
$G(\cE_{BDR})$ has infinite treewidth.  On the other hand one can
easily show that the planar graph $G(\cE_{BDR})$ has bounded
hyperbolicity. Indeed, suppose by way of contradiction that
$G(\cE_{BDR})$ contains a $3\times 3$ isometric square grid
$\Gamma$. Since the cube complex of $G(\cE_{BDR})$ is 2-dimensional,
$\Gamma$ is a convex and thus gated subgraph of $G(\cE_{BDR})$. Let
$v$ be the gate of $v_0$ in $\Gamma$. Then $\Gamma$ contains a
$2 \times 2$ directed grid $\Gamma'$ having $v$ as a source. Let $v'$
be the center of $\Gamma'$ and observe that $v'$ has two incoming and
two outgoing arcs. Since $G(\cE_{BDR})$ is planar, the four squares of
$\Gamma'$ around $v'$ are the unique faces of the planar graph
$G(\cE_{BDR})$ incident to $v'$. Consequently, $v'$ is the source of
only one square in $\Gamma'$ and thus in $G(\cE_{BDR})$. But in
$G(\cE_{BDR})$, each inner vertex is the source of two distinct
squares (defined by the three outgoing edges at $v'$), a
contradiction. By Proposition~\ref{hyperbolic-gridfree}, the event
structure $\cE_{BDR}$ is grid-free since $\cE_{BDR}$ has degree $3$.

\begin{figure}
  \centering
  \includegraphics[scale=0.7]{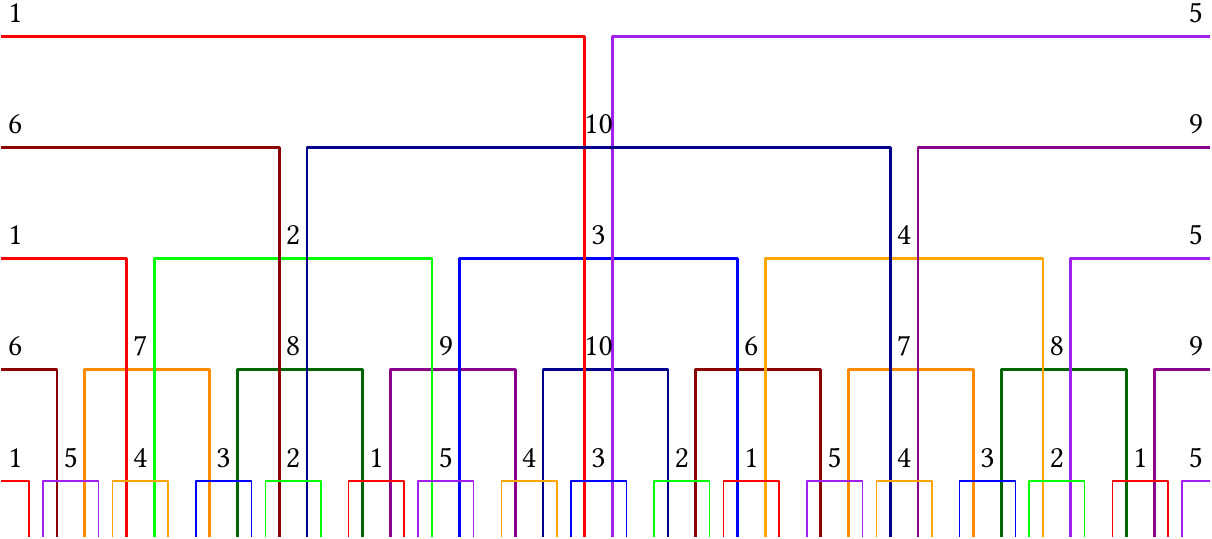}%
  \caption{The trace labeling of the events (hyperplanes) of
    $\cE_{BDR}$ obtained after five steps of the construction.}%
  \label{fig-coloredhyperplanes--BDR}
\end{figure}

Finally, the fact that $\cE_{BDR}$ admits a trace-regular labeling was
established in~\cite{BaDaRa}. Badouel et al.\ showed that the domain
$\cD(\cE_{BDR})$ of $\cE_{BDR}$ is the domain of a finite trace
automaton (see~\cite{BaDaRa} or~\cite{Schmitt} for the definition) and
thus $\cD(\cE_{BDR})$ admits a regular nice labeling. Using the result
of Schmitt~\cite{Schmitt} (or the more general one of
Morin~\cite{Morin}), this implies that $\cE_{BDR}$ is a trace-regular
event structure.
The labeling of the events (hyperplanes) of $\cE_{BDR}$ is given in
Fig.~\ref{fig-coloredhyperplanes--BDR} for the events obtained after
five steps of the construction.  The idea is that the events
constructed at step $4i+1$ are labeled consecutively from left to
right $1, 5, 4, 3, 2, 1, \ldots, 1, 5$, those constructed at step
$4i+2$ are labeled $6, 10, 9, 8, 7, 6, \ldots, 10, 9$, those
constructed at step $4i+3$ are labeled $1, 2, 3, 4, 5, \ldots, 4, 5$,
and those constructed at step $4i+4$ are labeled
$6, 7, 8, 9 , 10, \ldots, 8, 9$. A tedious check of the construction
shows that this labeling gives 40 types of labeled principal
filters\footnote{In~\cite{BaDaRa}, only 20 types of labeled principal
  filters are mentioned}.

Consequently, $\cE_{BDR}$ is a grid-free trace-regular event structure
whose graph $G(\cE_{BDR})$ has infinite treewidth. By
Theorem~\ref{barycentric-subdivision-MSO}, the MSO theory
$\MSO(\dotcE_{BDR})$ of the hairing of $\cE_{BDR}$ is undecidable.

\begin{remark}
  By Corollary~\ref{cor-tracereg-implies-stronglyreg}, the domain of
  $\cE_{BDR}$ is the principal filter of the universal cover of some
  finite (virtually) special cube complex. However, we do not even
  have an explicit construction of a small NPC square complex
  $X_{BDR}$ such that the domain of $\cE_{BDR}$ is a principal filter
  of the universal cover of $X_{BDR}$. To produce such a NPC square
  complex $X_{BDR}$, one can use the result of Schmitt~\cite{Schmitt}
  (or Morin~\cite{Morin}) to find a trace-regular labeling of
  $\cE_{BDR}$, then the result of
  Thiagarajan~\cite{Thi_regular,Thi_conjecture} to construct a net
  system $N_{BDR}$ such that its event structure unfolding
  $\cE_{N_{BDR}}$ is $\cE_{BDR}$, and finally
  Theorem~\ref{Petri-to-special} to construct a finite special cube
  complex $X_{BDR}$ from $N_{BDR}$. The first two steps of this
  approach significantly increase the number of labels used to label
  the events of $\cE_{BDR}$ and it is not clear how to avoid this
  combinatorial explosion.
\end{remark}

\begin{remark}
  In view of Remark~\ref{rem-tore-bazar}, one can ask whether there
  exists a NPC square complex $X_{BDR}$ with a hyperbolic universal
  cover $\tX_{BDR}$ such that $\cD(\cE_{BDR})$ is a principal filter
  of $\tX_{BDR}$.
\end{remark}

We do not know if the hairing operation is necessary in order to
obtain grid-free trace-regular event structures with undecidable MSO
theories. In particular, we wonder whether $\MSO(\cE_Z)$ and
$\MSO(\cE_{BDR})$ are decidable. If this is not the case, this would
provide counterexamples to Conjecture~\ref{MSO} that are not based on
encoding MSO formulas over the domain by MSO formulas over the hair
events.

\section{Conclusion}

The three Thiagarajan's conjectures were a driving force in authors
research for a long time.  Our motivation to work on those conjectures
was their intrinsic beauty and fundamental nature (finite versus
infinite and decidability versus undecidability, both expressible in
combinatorial way) and also our expertise in median graphs and CAT(0)
cube complexes. This expertise allowed us to work with the domain of
the event structure instead of the event structure itself and perform
geometric operations on the domain which preserve the property to be
median or CAT(0). This also allowed us to use the rich and deep theory
of median graphs, CAT(0) cube complexes, and, more importantly, of
special cube complexes. We strongly believe that those three
ingredients are essential in the understanding of Thiagarajan's
conjectures.

Even if we found counterexamples to the three Thiagarajan's
conjectures, the work on them raised many interesting open questions
and lead to a better understanding of trace-regularity and to a
surprising link between 1-safe Petri nets and finite special cube
complexes. We think that the characterization of trace-regular event
structures provided by this bijection can be viewed as a positive
answer to Thiagarajan's Conjecture~\ref{conjecture_regular}. The open
questions related to the first two conjectures are presented in the
papers~\cite{CC-thiag} and~\cite{ChHa}. The Questions 3 and 4
from~\cite{ChHa} are related to Conjecture~\ref{conj-nice} and to the
embedding question (which may be easier than the nice labeling
question). The questions and conjectures from~\cite{CC-thiag} describe
several conjectured classes of event structures for which
Conjecture~\ref{conjecture_regular} is true (hyperbolic and
confusion-free event structures) and relate the decidability of
existence of finite regular nice labelings with the decidability of
the question of whether a finite cube complex is virtually special, an
open question formulated in~\cite{Agol_ICM} and~\cite{BrWi}.

We conclude this paper with a speculation about
Question~\ref{question3} and Conjecture~\ref{MSO}. Our counterexample
to Conjecture~\ref{MSO} shows that grid-freeness of a trace-regular
event structure $\cE_N$ does not implies the decidability of
$\MSO(\cE_N)$. On the other hand, we proved that decidability of
$\MSO(\oG(\cE_N))$ (or of $\MSO(G(\cE_N))$) is equivalent to finite
treewidth and implies the decidability of $\MSO(\cE_N)$. We also
showed that the decidability of $\MSO(\oG(\cE_N))$ is equivalent with
the decidability of $\MSO(\dotcE_N)$, where $\dotcE_N$ is the hairing
of $\cE_N$. On the other hand, we know that conflict-free event
structures (whose domains may have infinite treewidth) have decidable
MSO theory.  Therefore, in the attempt to correct the formulation of
Conjecture~\ref{MSO}, we think that it is necessary to define the
``haired'' version of the treewidth of the domain of the event
structure $\cE_N$. We know that bounded treewidth is characterized by
bounded square grid minors. In a similar way, a \emph{haired grid
  minor} of $G(\cE_N)$ is a minor of $G(\cE_N)$ which is a haired
square grid. A \emph{haired square grid} is a square grid in which to
each vertex is added a pendant edge (hair). Notice that the hairs are
edges of $G(\cE_N)$ and thus correspond to events of $\cE_N$. We
require additionally that in a haired square grid minor, the hairs
correspond to pairwise conflicting events.  The \emph{haired
  treewidth} of $G(\cE_N)$ is the supremum of the sizes of haired
square grid minors of $G(\cE_N)$. We wonder whether the MSO theory of
a trace-regular event structures $\cE_N$ is decidable whenever the
graph $G(\cE_N)$ has finite haired treewidth.

\section*{Note}
After the completion of this paper, Didier Caucal informed us that he
also has constructed a counterexample to Conjecture~\ref{MSO}.

\begin{acks}
  The authors would also like to thank the anonymous referees for
  their valuable comments and helpful suggestions for improving
  readability.
  
  The work on this paper was supported by
  \grantsponsor{ANR}{ANR}{https://anr.fr/} project
  DISTANCIA~(\grantnum{ANR}{ANR-17-CE40-0015}).
\end{acks}

\bibliographystyle{ACM-Reference-Format}
\bibliography{refs}

\end{document}